\documentclass[12pt]{article}
\usepackage[authoryear]{natbib}
\usepackage{graphicx,bm,multirow}
\usepackage{amsthm,amsfonts,amssymb,color,comment}
\usepackage{amsmath}
\usepackage{comment}
\usepackage[hidelinks]{hyperref} 
\usepackage{booktabs, siunitx}
\usepackage[a4paper, total={6.5in, 9in}]{geometry}
\usepackage{setspace}
\usepackage{algorithm}
\usepackage{algpseudocode}
\usepackage{multirow}
\usepackage{enumitem}
\usepackage{caption}
\usepackage{xcolor}
\usepackage{rotating}
\usepackage[affil-it]{authblk}
\usepackage{adjustbox}

\newtheorem{Theorem}{Theorem}
\newtheorem{Lemma}{Lemma}

\newtheorem{Assumption}{Assumption}
\newtheorem{Remark}{Remark}

\usepackage{tikz}

\usepackage{subcaption}

\numberwithin{equation}{section}

\usepackage{chngcntr}

\definecolor{Red}{rgb}{.9,0,0}

\begin{document}

\title{B-CONCORD - A scalable Bayesian high-dimensional precision matrix estimation procedure}
\author{Peyman Jalali, Kshitij Khare and George Michailidis}
\affil{University of Florida} 
\date{}

\maketitle

\begin{abstract}
Sparse estimation of the precision matrix under high-dimensional scaling constitutes a canonical problem in 
statistics and machine learning. Numerous regression and likelihood based approaches, many frequentist and some Bayesian
in nature have been developed. Bayesian methods provide direct uncertainty quantification of the model parameters through the
posterior distribution and thus do not require a second round of computations for obtaining debiased estimates of the model parameters and their confidence intervals. However, they are computationally expensive for settings involving more than 500 variables. To that end, we develop B-CONCORD for the problem at hand, a Bayesian analogue of the CONvex CORrelation selection methoD (CONCORD) introduced by Khare et al. (2015). B-CONCORD leverages the CONCORD generalized likelihood function
together with a spike-and-slab prior distribution to induce sparsity in the precision matrix parameters. We establish model selection and estimation consistency under high-dimensional scaling; further, we develop a procedure that refits only the non-zero parameters of the precision matrix, leading to significant improvements in the estimates in finite samples.
Extensive numerical work illustrates the computational scalability of the proposed approach vis-a-vis competing Bayesian methods, as well as its accuracy.
\end{abstract}


\section{Introduction}

\noindent
Graphical models capture conditional dependence relationships between a set of random variables
\cite{buhlmann2011statistics}. The emergence of high dimensional data, wherein researchers have measured a large
number of variables $p$ on a relative small number of samples $n$ led to the study of estimating such models
under {\it sparsity} constraints, namely that the number of true non-zero parameters is less than the sample size.
A rich body of work on algorithms and the associated theoretical considerations emerged addressing this problem
\cite{wainwright2019high}. A key development was the introduction of the neighborhood selection method 
\cite{Meinshausen:Buhlmann:2006} which for Gaussian graphical models leverages the connection between the $(i,j)^{th}$ entry of the precision matrix $\mathbf{\Omega} = \mathbf{\Sigma}^{-1}$ -the model parameter of interest- to the partial correlation between the $i^{th}$ and $j^{th}$ variable; the latter can be estimated through a regression model even under sparsity constraints. This regression approach was used for graphical models for binary variables, as well as mixed measurement variables (e.g. numerical, binary, count) \cite{chen2015selection}, in addition to the Gaussian case.

As previously mentioned, numerous approaches have been developed for estimation of a sparse precision matrix, either
based on the neigborhood selection idea or leveraging the Gaussian likelihood; e.g., see  
\cite{Meinshausen:Buhlmann:2006, yuan2007model, friedman2008sparse, peng2009partial, CLL:2011, khare2015convex} 
and references therein. These approaches come with statistical guarantees expressed in the form of high probability
error bounds for selecting the correct non-zero model parameters and for the norm difference between the estimated
model parameters and the data generating ones. The Bayesian paradigm provides comprehensive uncertainty quantification 
of the model parameters through the posterior distribution. To that end, several Bayesian 
counterparts to the penalized (Gaussian) likelihood based methods have been proposed in the literature; e.g., see
\cite{DLR:2011, wang2012bayesian, Cheng:Lenkoski:2012, Wang:2015}. However, a key challenge for these approaches is their
scalability to settings involving a large number of variables (e.g. $p\geq 500$).

The main goal of this paper is to develop a highly \textit{scalable} Bayesian approach for sparse precision matrix estimation in high-dimensional settings with thousands of variables by leveraging the neighborhood selection method. This is accomplished by leveraging the regression based generalized likelihood function in \cite{khare2015convex} and combined with a spike-and-slab 
prior distribution on the precision matrix $\mathbf{\Omega}$ parameters, to obtain a generalized posterior distribution.
A key advantage of the generalized likelihood function is that the resulting posterior distributions for the elements
of $\mathbf{\Omega}$ are available in closed form (up to a normalizing constant), unlike full Gaussian likelihood approaches.
This enables derivation of a scalable Gibbs sampler that works well in settings involving thousands of variables (see numerical evidence in Section \ref{simulations}) and outperforms state of the art approaches in the literature. Further, we establish posterior consistency both for selecting the correct non-zero elements of $\mathbf{\Omega}$ and for the norm difference between the estimated model parameters and the data generating ones. In addition, to improve estimation accuracy of the magnitude of the non-zero elements in small sample settings, we introduce a novel refitting procedure with a modified prior distribution that achieves this objective, a novel development of independent interest for Bayesian methods for high-dimensional sparse estimation problems.

The remainder of the paper is organized as follows. The problem formulation based on the generalized likelihood function and the spike-and-slab prior distribution, together with the development of a scalable Gibbs sampler for sampling from the resulting posterior distribution is presented in Section \ref{generalized:likelihood:spike:slab}. The novel refitting procedure
is discussed in Section \ref{refitting:technique}. High-dimensional selection consistency and convergence rates for the model parameters are established in Section \ref{sec:consistency}. Extensive numerical evaluation of the proposed algorithm and an illustration to a Omics data set is given in Section \ref{simulations}. Finally, the proofs of all technical results and some background on the generalized likelihood function are delegated to the Appendix.

\section{The Bayesian CONCORD (B-CONCORD) framework for precision matrix estimation} 
\label{generalized:likelihood:spike:slab}

The key building block in the proposed framework is the CONCORD generalized likelihood function introduced in 
\cite{khare2015convex} that is motivated by the regression based neighborhood regression approach for estimation
of Gaussian graphical models introduced in \cite{Meinshausen:Buhlmann:2006} ({\it a brief introduction and 
derivationof it is given in Supplemental Section \ref{background}}). Let 
$\mathcal{Y} :=\left(\{\mathbf{y}_{i:}\}_{i=1}^{n}\right)$ be independent and identically distributed observations 
from a $p$-variate (continuous) distribution, with mean $\boldsymbol{0}$ and covariance matrix 
$\mathbf{\Omega}^{-1}$. Then, the CONCORD generalized likelihood function is defined as follows. 
\begin{equation} \label{eq2}
\mathcal{L}_{CONCORD} ({\bf \Omega}) := \exp \left( - \frac{1}{2}\sum\limits_{j=1}^{p}
{\sum\limits_{i=1}^{n}{\left( \omega_{jj}y_{ij} + \sum\limits_{k \ne j}{\omega_{jk}y_{ik}} \right)^2}} + n \sum\limits_{j=1}^p 
\log \omega_{jj} \right). 
\end{equation}

\noindent
Similarly to the neighborhood selection approach, for computational reasons we require that $\Omega \in \mathbf{\mathbb{M}}_{p}^+$, the space of real $p\times p$ symmetric matrices with positive diagonal elements, but {\it not} necessarily positive definite.

\medskip 
\noindent
{\bf A Spike and Slab Prior Distribution for B-CONCORD.} \\
\noindent
For every off-diagonal entry of ${\bf \Omega} \in \mathbf{\mathbb{M}}_p$, we assume the following: 
particular 
\begin{equation}\label{spikeslab}
\omega_{jk} \sim (1-q)I_{\{0\}}(\omega_{jk}) + q \phi_{\lambda_{jk}} (\omega_{jk}) I_{\mathbb{R} \setminus \{0\}}(\omega_{jk}) 
\end{equation}

\noindent
independently for every $1 \leq j < k \leq p$, where $\phi_\lambda$ denotes the normal density with mean zero and 
variance $1/\lambda$. Further, we impose an independent Exponential$(\gamma)$ prior distribution on all the 
diagonal entries of ${\bf \Omega}$. Hence, the (shrinkage) hyperparameters are $\{\lambda_{jk}\}_{1 \leq j < k \leq p}$ and 
$\gamma$. We will discuss the choice of these hyperparameters in Remark \ref{hyperparameter:selection}. 

We now introduce additional notation for ease of exposition, in particular for the asymptotic analysis in 
Section \ref{sec:consistency}. Let 
\begin{equation} \label{theta}
{\boldsymbol \xi} = (\omega_{12}, \omega_{13}, \cdots, \omega_{1,p}, \omega_{2,3}, \cdots, \omega_{p-1,p}) 
\end{equation}

\noindent
denote the collection of off-diagonal entries of the symmetric matrix ${\bf \Omega}$. Further, let $\boldsymbol{l} \in 
\{0,1\}^{{p \choose 2}}$ be a generic sparsity pattern for ${\boldsymbol \xi}$. There are $2^{{p \choose 2}}$ such 
sparsity patterns. For example, when $p = 3$, there are ${3 \choose 2} = 3$ off-diagonal entries, and the 
$2^{{p \choose 2}} = 8$ possible sparsity patterns in the off-diagonal entries are  
\begin{align*}
& (0,0,0), \; (1,0,0), \; (0,1,0), \; (0,0,1)\\
& (1,1,0), \; (1,0,1), \; (0,1,1), \; (1,1,1). 
\end{align*}

\noindent
For every sparsity pattern $\boldsymbol{l}$, let $d_{\boldsymbol{l}}$ be the density (number of non-zero entries) of 
$\boldsymbol{l}$, and $\mathcal{M}_{\boldsymbol{l}}$ be the space where ${\boldsymbol \xi}$ varies when 
restricted to follow the sparsity pattern $\boldsymbol{l}$. in the $p=3$ example discussed above, $d_{(0,0,0)} = 0$ and 
$d_{(1,0,0)} = 1$. 

Using straightforward calculations, the independent spike-and-slab priors for the off-diagonal entries (specified in 
(\ref{spikeslab}), can be alternatively represented as a hierarchical prior distribution as follows: 
\begin{equation} \label{hierarchical prior 1}
\pi \left( {\boldsymbol \xi} | \boldsymbol{l}\right) = \frac{|\mathbf{\Lambda}_{\boldsymbol{l}\boldsymbol{l}}|^{\frac{1}{2}}}{\left(2\pi \right)^{\frac{d_{\boldsymbol{l}}}{2}}} \exp \left(-\frac{{\boldsymbol \xi}'\mathbf{\Lambda}{\boldsymbol \xi}}{2} \right)I_{\left( {\boldsymbol \xi} \in \mathcal{M}_{\ell}\right)},
\end{equation}

\noindent
where $\mathbf{\Lambda}$ is a diagonal matrix with diagonal entries $(\lambda_{jk})_{1 \leq j < k \leq p}$, and 
$\mathbf{\Lambda}_{\boldsymbol{l}\boldsymbol{l}}$ is a sub-matrix of $\mathbf{\Lambda}$ obtained after removing the 
rows and columns corresponding to the zeros in ${\boldsymbol \xi} \in \mathcal{M}_{\boldsymbol{l}}$. In other words, 
given the sparsity pattern ${\boldsymbol{l}}$, the non-zero entries of ${\boldsymbol \xi}$ follow a 
$d_{\boldsymbol{l}}$-dimensional multivariate normal distribution with mean ${\bf 0}$ and covariance matrix 
$\Lambda_{\boldsymbol{l}\boldsymbol{l}}^{-1}$. The marginal distribution of $\boldsymbol{l}$ is given by 

\begin{equation}\label{hierarchical prior 2}
\pi(\boldsymbol{l}) \propto \begin{cases} 
q^{d_{\boldsymbol{l}}}(1-q)^{\binom{p}{2} - d_{\boldsymbol{l}}} & d_{\boldsymbol{l}} \leq \tau,\\
0 & d_{\boldsymbol{l}} > \tau. 
\end{cases}
\end{equation}
which puts zero mass on {\it unrealistic sparsity patterns}, i.e., sparsity patterns with more than $\tau$ non-zero entries. In the subsequent theoretical analysis, we discuss appropriate values for $\tau$. 

Using (\ref{hierarchical prior 1}) and (\ref{hierarchical prior 2}), the marginal prior distribution on ${\boldsymbol \xi}$ can be obtained as a mixture of multivariate normal distributions as follows:
\begin{equation}\label{prior on theta}
\begin{split}
\pi\left( {\boldsymbol \xi} \right)  = \sum\limits_{ \boldsymbol{l} \in \mathcal{L}}\pi\left({\boldsymbol \xi} | \boldsymbol{l}\right)\pi\left(\boldsymbol{l}\right)
&\propto \sum\limits_{ \boldsymbol{l} \in \mathcal{L}} q^{d_{\boldsymbol{l}}}(1-q)^{\binom{p}{2} - d_{\boldsymbol{l}}} \left\{ \frac{|\mathbf{\Lambda}_{\boldsymbol{l}\boldsymbol{l}}|^{\frac{1}{2}}}{\left(2\pi \right)^{\frac{d_{\boldsymbol{l}}}{2}}} \exp \left(-\frac{{\boldsymbol \xi}'\mathbf{\Lambda}{\boldsymbol \xi}}{2} \right)I_{\left( {\boldsymbol \xi} \in \mathcal{M}_{\ell}\right)} \right\}. 
\end{split}
\end{equation}

\noindent
Note that the vector ${\boldsymbol \xi}$ only incorporates the off-diagonal entries of $\mathbf{\Omega}$. Regarding the 
diagonal entries, we define $\boldsymbol{\delta}$ to be the vector of all diagonal elements $\mathbf{\Omega}$, i.e.
\begin{equation}\label{delta}
\boldsymbol{\delta} = \left( \omega_{11}, ..., \omega_{pp}\right). 
\end{equation}

\noindent
Note that an independent Exponential($\gamma$) prior distribution is assigned on each coordinate of $\boldsymbol{\delta}$, i.e,
\begin{equation}\label{prior on delta}
\pi\left( \boldsymbol{\delta} \right) \propto \exp \left( -\gamma \boldsymbol{1}'\boldsymbol{\delta} \right) I_{\mathbb{R}_{+}^{p}}
\left( \boldsymbol{\delta} \right).
\end{equation}


\subsection{Computing the Posterior Distribution} \label{posterior:computation:gibbs:sampler}

\noindent
Combining (\ref{eq2}) and (\ref{spikeslab}), it is easy to check that the generalized posterior density of ${\bf \Omega}$ is given by 
\begin{eqnarray}
\pi \left\{ \mathbf{\Omega} | \mathcal{Y}\right\} \propto & \ \text{exp} \left( n\sum\limits_{j=1}^{p}{\text{log}\omega_{jj}} - \frac{n}{2}\text{tr}\left( \mathbf{\Omega}^2\mathbf{S}\right) - 
\lambda \mathop{\sum\sum}\limits_{1 \le j < k \le p}\frac{\omega_{jk}^2}{2} - \lambda \sum_{j=1}^p\omega_{jj}\right) 
\nonumber\\
\times & \prod \limits_{1 \le j < k \le p} \left( I_{\{0\}} (\omega_{jk}) + \frac{q\sqrt{\lambda_{jk}}}{(1-q)\sqrt{2 \pi}} 
I_{\mathbb{R} \setminus \{0\}}(\omega_{jk}) \right). \label{BSSC posterior}
\end{eqnarray}

\noindent
The generalized posterior density in (\ref{BSSC posterior}) is intractable, in the sense that it is not feasible to draw exact 
samples from such density. However, we will use the conditional posterior density of each element $\omega_{jk}$ of 
$\mathbf{\Omega}$, $1 \leq j \leq k \leq p$, given the remaining elements denoted by $\mathbf{\Omega}_{-(jk)}$, to introduce 
an entry-wise Gibbs sampler that can generate approximate samples from the generalized posterior density in 
(\ref{BSSC posterior}). In order to compute the conditional posterior density of the off-diagonal elements $\omega_{jk}, \  
1 \le j \leq k \le p$, we first note that 
\begin{equation}\label{eq6}
 n \ \text{tr}\left( \mathbf{\Omega}^2\mathbf{S}\right)
=\sum\limits_{i=1}^{n}\| \left( \mathbf{\Omega} \mathbf{y}_{i:}\right) \|^2
= \sum\limits_{i=1}^{n}{\sum\limits_{j=1}^{p}{\left( \sum\limits_{k=1}^{p}{\omega_{jk}y_{ik}}\right)^2}} =\sum\limits_{j=1}^{p} \| 
\left( \sum\limits_{k=1}^{p}{\omega_{jk} \mathbf{y}_{:k}}\right) \|^2, 
\end{equation}
\noindent
where $||\cdot||$ denotes the $\ell_2$ norm of a vector.

\noindent
Thus, in view of (\ref{eq6}), straightforward algebra shows that
\begin{equation*}
\begin{split}
\pi (\omega_{jk}| \mathbf{\Omega}_{-(jk)}, \mathcal{Y}) &\propto \exp \left\{ -\frac{n}{2} \left( a_{jk}\omega_{jk}^{2} + 2 b_{jk}
\omega_{jk}\right)\right\} \left( I_{\{0\}} (\omega_{jk}) + \frac{q\sqrt{\lambda}}{(1-q)\sqrt{2 \pi}} 
I_{\mathbb{R} \setminus \{0\}}(\omega_{jk}) \right)\\
&=I_{\{0\} }(\omega_{jk}) + c_{jk}\frac{\sqrt{na_{jk}}}{\sqrt{2\pi}} \exp\left\{ -\frac{na_{jk}}{2} \left( \omega_{jk} + 
\frac{b_{jk}}{a_{jk}}\right)^2\right\}I_{\mathbb{R} \setminus \{0\}}(\omega_{jk})
\end{split}
\end{equation*}

\noindent
with,
\begin{equation*}
\begin{split}
a_{jk} &= s_{jj} + s_{kk} + \frac{\lambda_{jk}}{n}, \quad \quad \quad 
b_{jk} = \sum_{k' \neq k} \omega_{jk'} s_{kk'} + \sum_{j' \neq j} \omega_{j'k} s_{jj'}\\
c_{jk} &= \frac{q \sqrt{\lambda_{jk}}}{(1-q) \sqrt{na_{jk}}}\exp\left( \frac{nb_{jk}^2}{2a_{jk}}\right).
\end{split}
\end{equation*}
Next, letting $p_{jk} = \frac{c_{jk}}{1+c_{jk}}$, we can then write
\begin{equation}\label{cond. omegas_jk BSSC}
\begin{split}
(\omega_{jk}| \mathbf{\Omega}_{-(jk)}, \mathcal{Y}) \sim (1-p_{jk})I_{\{0\}}(\omega_{jk}) + p_{jk}N(-\frac{b_{jk}}{a_{jk}}, \frac{1}{na_{jk}})I_{\mathbb{R} \setminus \{0\} }(\omega_{jk}), \quad \quad 1 \le j < k \le p.
\end{split}
\end{equation}

Moreover, the diagonal elements $\omega_{jj}$ are conditionally independent and the conditional density of $ \omega_{jj}$ given $\mathbf{\Omega_{-(jj)}}, \ 1 \le j \le p$, is given by
\begin{equation}\label{cond. diags BSSC}
\begin{split}
f(\omega_{jj}| \mathbf{\Omega}_{-(jj)}, \mathcal{Y}) \propto \omega_{jj}^{n}\text{exp}\left\{ -\omega_{jj}^{2}\left( \frac{n}{2} s_{jj}
\right) - \omega_{jj}\left( \lambda_{jk} + n b_j \right) \right\},
\end{split}
\end{equation}

\noindent
where 
$$
b_j = \sum_{j' \neq j} \omega_{jj'}' s_{jj'}. 
$$

\noindent
Note that the density in \ref{cond. diags BSSC} is not a standard density, but using the fact that it has a unique mode at
\begin{equation}\label{mode of diag dist}
\frac{-(\lambda_{jk} + nb_j) + \sqrt{(\lambda + nb_j)^2 + 4n^2 s_{ii}^k}}{2n s_{ii}^k}, 
\end{equation}

\noindent
one can use a discretization technique to generate samples from it. However, we have observed in our extensive numerical work that the density in (\ref{cond. diags BSSC}) puts most of its mass around the mode. As a result, when appropriate, one can simply approximate it using a degenerate density with a point mass at it's mode, given in \ref{mode of diag dist}. This approximation allows faster implementation of the algorithm without sacrificing accuracy. Using the distributions in (\ref{cond. omegas_jk BSSC}) and (\ref{cond. diags BSSC}), we develop a component-wise Gibbs sampler to generate 
approximate samples from the joint posterior density in (\ref{BSSC posterior}). Given the current value of $\mathbf{\Omega}$, a single iteration of this Gibbs sampler -henceforth referred to as Bayesian Spike and Slab CONCORD (BSSC)- is described in Algorithm \ref{GibbsBSSC}.
\begin{algorithm}
\caption{Entry wise Gibbs Sampler for BSSC}
\begin{algorithmic}
\Procedure{BSSC}{$\mathbf{S}$}\Comment{Input the data}
\For{$j=1, ..., p-1$}
\For{$k=j+1, ..., p$}
\State $\omega_{jk} \sim (1-p_{jk})I_{\{0\} }(\omega_{jk}) + p_{jk}N(-\frac{b_{jk}}{a_{jk}}, \frac{1}{a_{jk}})I_{\mathbb{R} \setminus 
\{0\}}(\omega_{jk})$
\EndFor
\EndFor
\For{$j=1, ..., p$}
\State $\omega_{jj} \gets \frac{-(\lambda_{jk} + n\mathbf{\Omega}_{-jj}'\mathbf{S}_{-jj}) + \sqrt{(\lambda_{jk} + 
nb_j)^2 + 4n^2 s_{ii}^k}}{2n s_{ii}^k}$
\EndFor
\State \textbf{return} $ \mathbf{\Omega}$\Comment{Return $\mathbf{\Omega}$}
\EndProcedure
\end{algorithmic}\label{GibbsBSSC}
\end{algorithm}

Let $\left\{ \hat{\mathbf{\Omega}}^{(t)} \right\}_{t=1}^T$ denote the iterates obtained by running the Gibbs sampler (with a sufficiently long burn-in period). For each $1 \leq j < k \leq p$, we compute the proportion of times the corresponding entry was chosen to be non-zero, i.e.,
$$
\hat{p}_{jk} = \frac{1}{T} \sum_{t=1}^T 1_{\{\hat{\omega}_{jk}^{(t)} \neq 0\}}. 
$$

\noindent
If $\hat{p}_{jk}$ is greater than a pre-specified threshold $\upsilon \in (0,1)$, $(j,k)$, the $(j,k)^{th}$ entry is considered to be non-zero in the estimated sparsity pattern for the precision matrix. Note that by the ergodic theorem $\hat{p}_{jk}$ converges to the posterior probability of $\omega_{jk}$ being non-zero as $T \rightarrow 0$. Hence, an entry is classified as non-zero, if the posterior probability of being non-zero is above $\upsilon$. We denote the resulting sparsity pattern estimate by 
$\hat{\boldsymbol{l}}_{\upsilon, BSSC}$, and use the threshold $\upsilon = 0.5$ in our numerical work to obtain 
the {\it median probability model/sparsity pattern}. The user has the flexibility to select more conservative or relaxed threshold values. 

\begin{Remark} (Selection of hyperparameters) \label{hyperparameter:selection}
A good selection of shrinkage hyperparameters $\{\lambda_{jk}\}_{1 \leq j < k \leq p}$ and $\{\gamma_j\}_{j=1}^p$, along with the mixing probability $q$ is important. Following 
\cite{park2008bayesian,wang2012bayesian}, independent gamma prior distributions are assigned
on each shrinkage parameter $\lambda_{jk}$ and $\gamma_j$, i.e., 
\begin{align*}
& \lambda_{jk} \sim \text{Gamma}(r, s) \quad \quad \mbox{for} \quad 1 \leq j < k \leq p, 
& \gamma_j \sim \text{Gamma}(r, s) \quad \quad \mbox{for} \quad 1 \leq j \leq p 
\end{align*} 

\noindent
for some fixed $r, s > 0$. Straightforward calculations demonstrate that $\{\lambda_{jk}\}_{1 \leq j < k \leq p}, 
\{\gamma_j\}_{j=1}^p$ are conditionally mutually independent given $\mathbf{\Omega}, \mathcal{Y}$, i.e., 
\begin{align*}
& \lambda_{jk} \mid \mathbf{\Omega}, \mathcal{Y} \sim \text{Gamma}(r + 0.5, 0.5 \omega_{jk}^2+ s), 
& \gamma \mid \mathbf{\Omega}, \mathcal{Y} \sim \text{Gamma}(r + 1, \omega_{jj} + s). 
\end{align*}

\noindent
Since $\mathbb{E}\{\lambda_{jk} \mid \mathbf{\Omega}, \mathcal{Y}\} = \frac{r + 0.5}{0.5 \omega_{jk}^2+ s}$ and 
$\mathbb{E}\{\gamma_j \mid \mathbf{\Omega}, \mathcal{Y}\} = \frac{r + 1}{\omega_{jj} + s}$, this approach 
selects the respective shrinkage parameters based on the current $\omega$-values in a way that 
larger (smaller) entries are regularized less (more) on average. For the parameters $r$ and $s$ of the  
Gamma prior distribution, in absence of any prior information, we recommend the non-informative choices $r=10^{-4}$ and 
$s=10^{-8}$, which come very close to flat prior distributions for the $\lambda$ and $\gamma$ values, and are based 
on the suggestions made in \cite{wang2012bayesian}. Extensive numerical work suggests (see Section 
\ref{simulations}) that these are satisfactory choices. 

The default choice for the mixing probability $q$ is the objective one, namely $q = 1/2$. Based on the consistency results in 
Section \ref{sec:consistency}, one can use the choice $q = 1/p$ in really high-dimensional settings to further 
encourage sparser models. 
\end{Remark}

\subsection{Estimating magnitudes of non-zero entries: Correcting for bias using refitting} \label{refitting:technique}

\noindent
As previously mentioned, the key objective of the B-CONCORD methodology is the identification of the correct sparsity pattern in the precision matrix $\mathbf{\Omega}$. However, a good estimate of the ``strength" (magnitude) of the conditional association between two variables is also required for downstream analysis in many applications. 
Note that such estimates can be obtained from Algorithm \ref{GibbsBSSC}, and Theorem \ref{estimation:consistency:refitted:posterior} establishes their asymptotic accuracy and convergence rates under high-dimensional scaling. Nevertheless, in finite sample settings these estimates exhibit bias, an issue also noticed in the frequentist literature and resolved through the development of debiasing procedures \cite{Zhang:Zhang:2014, Javanmard:Montanari:2014, VanDeGeer:2019}). The numerical work in Section \ref{simulations} also provides evidence for the presence of bias. 

\noindent
A popular approach for addressing this problem for regularized estimates in a frequentist setting is to employ a refitting step, wherein only the non-zero entries of the precision matrix are re-estimated (e.g., see \cite{ma2016joint}). Next, we propose a Bayesian refitting step for obtaining debiased/improved estimates of the magnitudes of the non-zero entries of $\mathbf{\Omega}$ (as specified by the sparsity pattern obtained from Algorithm \ref{GibbsBSSC}). Figure \ref{fig:bias motivation} is an illustration of the effectiveness of the refitting technique, described later in this section. The plot depicts the estimation accuracy for one coefficient of a precision matrix of dimension $p=50$ based on sample sizes $n=100$ (panel \ref{fig:bias motivation sub1}) and $n=1000$ (panel \ref{fig:bias motivation sub2}). The posterior density of the coefficient was estimated using both BSSC and the refitting approach, colored in red and blue, respectively. The true coefficient value is corresponds to the green vertical line. Figure \ref{fig:bias motivation} shows that the true coefficient is located within the range of the estimated posterior density generated by the refitting approach. Theorem \ref{estimation:consistency:refitted:posterior} theoretically establishes posterior estimation consistency of the refitting approach. In addition, in section \ref{accuracy}, we further demonstrate the numerical significance of the refitting approach in reducing the estimation error. 

\begin{figure}[h]
\centering
\caption{Illustrating the presence of bias in the output of Algorithm \ref{GibbsBSSC}) and its correction through a refitting step for a coefficient from a $p=50$ dimensional precision matrix, estimated from $n=100$ (panel (a)) and $n=1000$ (panel (b)) samples.}
\begin{subfigure}{.5\textwidth}
  \centering
  \includegraphics[height=2in,width=3in,angle=0]{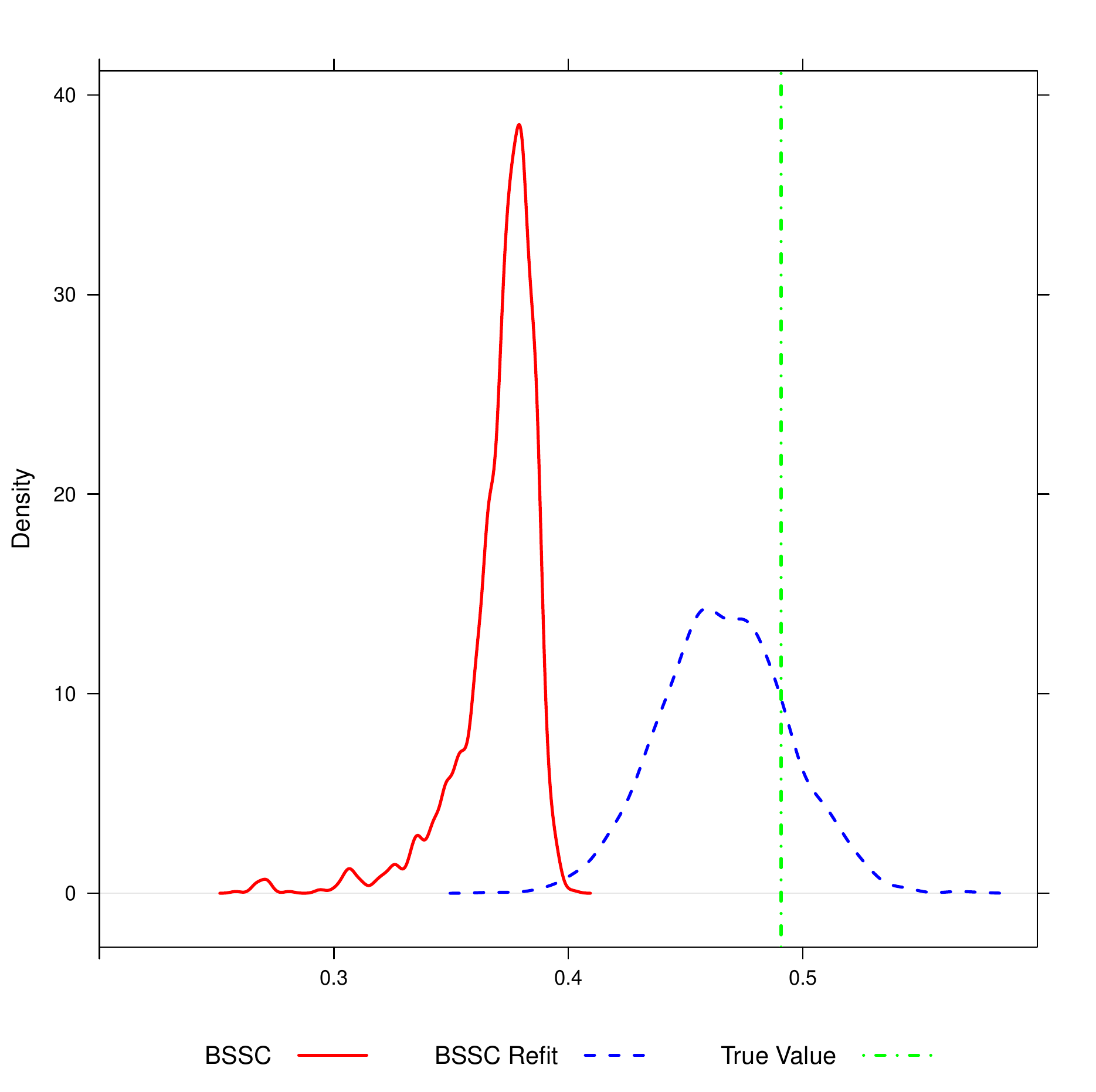}
  \caption{}
  \label{fig:bias motivation sub1}
\end{subfigure}%
\begin{subfigure}{.5\textwidth}
  \centering
  \includegraphics[height=2in,width=3in,angle=0]{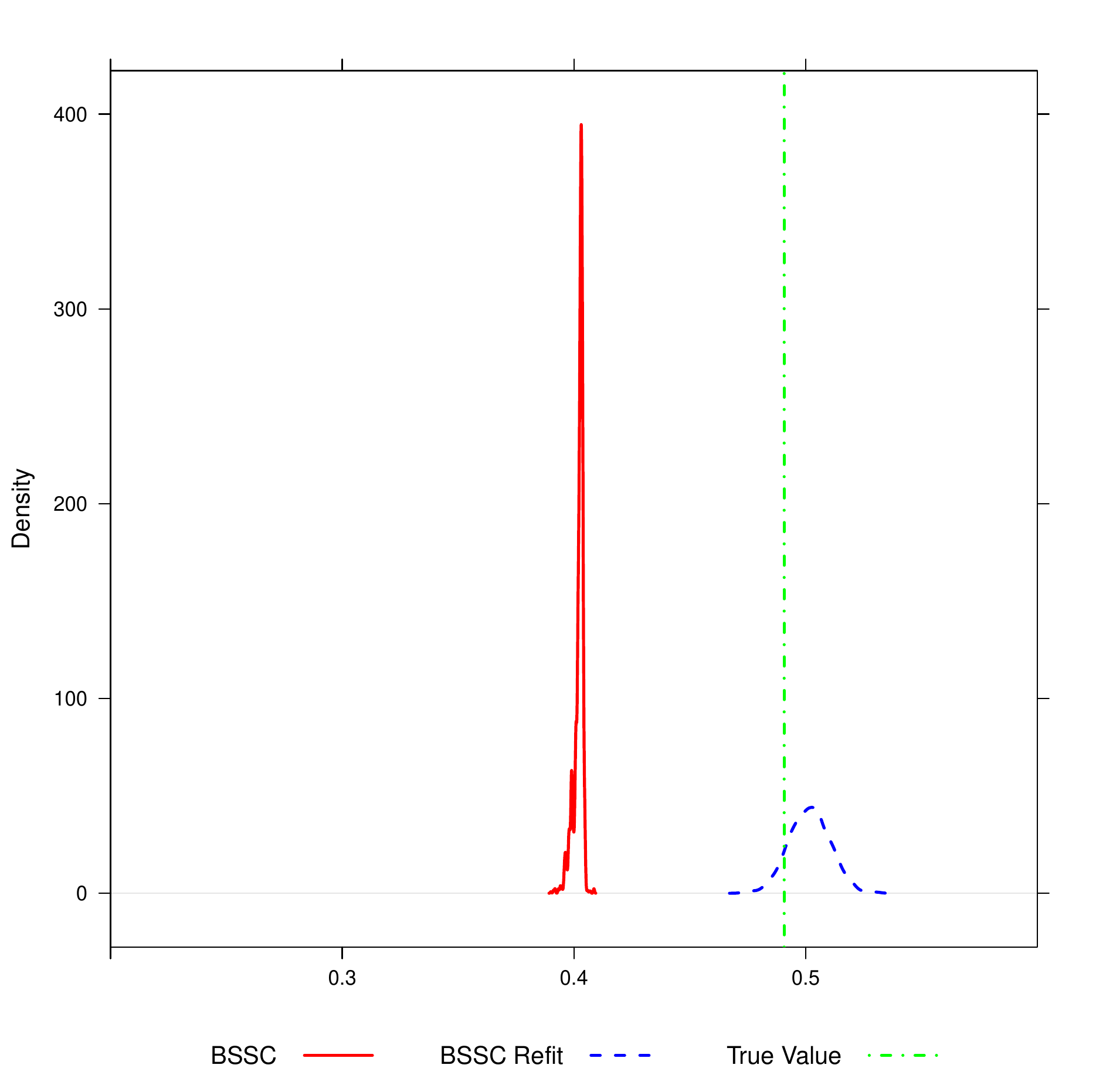}
  \caption{}
  \label{fig:bias motivation sub2}
\end{subfigure}
\label{fig:bias motivation}
\end{figure}

Recall from the discussion at the end of Section \ref{posterior:computation:gibbs:sampler} that using the output of 
Algorithm \ref{GibbsBSSC}, the proportion of times the $(j,k)^{th}$ entry was non-zero in the iterates of the sampler (denoted 
by $\hat{p}_{jk}$) can be computed. Let $\upsilon$ be the user specified threshold, and let 
\begin{equation} \label{estimated:edge:set}
\hat{E} = \{ (j,k): \; \hat{p}_{jk} > \upsilon\}
\end{equation}

\noindent
denote the collection of indices classified as non-zero. Let $\hat{G}$ denote the graph with vertices $\{1,2, \cdots, p\}$ and 
edge set $\hat{E}$. In other words, $\hat{G}$ encodes the sparsity pattern of $\hat{\mathbf{\Omega}}$. Let 
$\mathbb{M}_{\hat{G}}$ be defined as
$$
\mathbb{M}_{\hat{G}} = \left\{ \mathbf{\Omega} \in \mathbb{M}_p^+: M_{ij} = 0 \mbox{ if } i \neq j, (i,j) \notin \hat{E} \right\},
$$

\noindent
which is the space of symmetric $p \times p$ matrices with positive diagonal entries obeying the sparsity pattern 
encoded in $\hat{G}$. The goal is to construct a debiased/improved estimate (and corresponding credible region) for 
$\mathbf{\Omega}\in \mathbb{M}_{\hat{G}}$ by specifying an appropriate prior distribution on $\mathbb{M}_{\hat{G}}$. 

An immediate question that might arise in the mind of the reader is that while the positive definite constraint can be relaxed for the purpose of estimating the sparsity pattern in $\mathbf{\Omega}$, this constraint should be imposed for estimating the 
magnitudes of the entries of $\mathbf{\Omega}$. One can certainly do this by further restricting to the space 
$\mathbb{P}_{\hat{G}}$ (which is defined to be the intersection of $\mathbb{M}_{\hat{G}}$ with positive definite matrices), 
specifying a prior on $\mathbb{P}_{\hat{G}}$, and then using the subsequent posterior distribution. This however leads to 
significant computational challenges and involves inversion of $p-1$ dimensional matrices. Hence, for computational reasons, 
we first construct our prior on $\mathbb{M}_{\hat{G}}$, and project the resulting estimator (and associated credible region) on 
the space of positive definite matrices by modifying the diagonal entries (see the end of this section for details). 

We start by specifying an improper prior distribution on $\mathbb{M}_{\hat{G}}$, and explain why we expect the resulting 
posterior distribution to lead to good estimators. We specify independent (improper) uniform priors on the off-diagonal 
entries $\{\omega_{jk}\}_{(j,k) \in E}$ and for the diagonal entries we independently specify the following improper priors: 
\begin{equation}
\begin{split}
\pi(\omega_{jj}) &\propto \exp\left( n\omega_{jj} - n\log\omega_{jj}\right), \quad \omega_{jj} > 0, \quad \text{for} \quad 1 \leq j 
\leq p.
\end{split}
\end{equation}

\noindent
Then, the joint posterior distribution on $\mathbb{M}_{\hat{G}}$ is given by 
\begin{equation} \label{refitted:posterior:density}
\pi_{\text{refitted}} \left(\mathbf{\Omega}|\mathcal{Y} \right) \propto \exp\{n \text{tr}\left(\mathbf{\Omega}\right)-\frac{n}{2} \text{tr}
\left( \mathbf{\Omega}^2 \mathbf{S})\right\} \mbox{ for } \mathbf{\Omega} \in \mathbb{M}_{\hat{G}}. 
\end{equation}

\noindent
We refer to this posterior distribution as the {\bf refitted posterior}, since it is defined conditional on the sparsity pattern estimated from Algorithm \ref{GibbsBSSC} and encoded in $\hat{G}$. The following lemma addresses the propriety and 
unimodality of the posterior distribution, and its proof can be found in the Supplement.
\newtheorem{lem}{Lemma}
\begin{lem} \label{propriety}
If the degree of $\hat{G}$ (maximum number of edges shared by any vertex) is less than $n$, then the following holds. 
\begin{enumerate}
\item The refitted posterior density in (\ref{refitted:posterior:density}) can be normalized to a proper probability density. 
\item This refitted posterior density is log-concave and has a unique mode. 
\end{enumerate}
\end{lem}

\noindent
The next lemma provides insights on why the mode of the refitted posterior distribution is a good/improved estimator of 
$\mathbf{\Omega}$. 
\begin{lem} \label{optimizer}
Suppose $\mathbf{K} \in \mathbb{P}_{\hat{G}}$. Then,
$$
\mathbf{K} = \mbox{arg min}_{\mathbf{\Omega} \in \mathbb{M}_{\hat{G}}} \left\{ \frac{n}{2} \text{tr} \left( \mathbf{\Omega}^2 
\mathbf{K}^{-1} \right) - n \text{tr} (\mathbf{\Omega}) \right\}. 
$$
\end{lem}

\noindent
Suppose the true precision matrix $\mathbf{\Omega}^0 = (\mathbf{\Sigma}^0)^{-1}$ belongs to $\mathbb{M}_G$ for some graph 
$G$. Then, it follows from Lemma \ref{optimizer} that 
$$
\mathbf{\Omega}^0 = \mbox{arg min}_{\mathbf{\Omega} \in \mathbb{M}_G} \left\{ \frac{n}{2} \text{tr} \left( \mathbf{\Omega}^2 
\mathbf{\Sigma}^0 \right) - n \text{tr} (\mathbf{\Omega}) \right\}. 
$$

Note that under mild regularity assumptions, the maximum entry-wise difference between the sample covariance matrix and 
the population covariance matrix is of the order $\sqrt{\log p/n}$. Hence, if $\hat{G}$ (obtained from 
Algorithm \ref{GibbsBSSC}) is an accurate estimate of the true underlying graph $G$ for the true precision matrix 
$\mathbf{\Omega}^0$, we expect that the mode of the refitted posterior density in (\ref{refitted:posterior:density}), given by 
$$
\hat{\mathbf{\Omega}}_{\mbox{mode, refitted}} = \mbox{arg min}_{\mathbf{\Omega} \in \mathbb{M}_{\hat{G}}} \left\{ \frac{n}{2} \text{tr} 
\left( \mathbf{\Omega}^2 S \right) - n \text{tr} (\mathbf{\Omega}) 
\right\}, 
$$

\noindent
is close to $\mathbf{\Omega}^0$. This heuristic analysis is formalized in a high-dimensional setting in 
Theorem \ref{estimation:consistency:refitted:posterior}. 

It follows from Lemma \ref{propriety} that the mode of the refitted posterior density is available in closed from. To compute 
credible intervals, we observe that the full conditional posterior densities of the non-zero elements in 
$\mathbb{M}_{\hat{G}}$ can be derived in a straightforward way. Let $\hat{E}$ denote the edge set for $\hat{G}$. 
In particular, it can be shown that the full conditional (refitted) posterior density of $\omega_{jk}$ for $(j,k) \in \hat{E}$ is 
\begin{equation}
N \left( -\frac{b_{jk}}{a_{jk}}, \frac{1}{na_{jk}} \right)
\end{equation}

\noindent
where 
\begin{equation*}
\begin{split}
a_{jk} &= s_{jj} + s_{kk} , \quad \quad \quad 
b_{jk} = \sum_{k' \neq k, (j,k') \in E} \omega_{jk'}' s_{kk'} + \sum_{j' \neq j, (j',k) \in E} \omega_{j'k} s_{jj'}. \end{split}
\end{equation*}

\noindent
Further, the full conditional (refitted) posterior density of $\omega_{jj}$ is 
\begin{equation}
\text{trunc}N(\frac{1- b_{j}}{s_{jj}}, \frac{1}{ns_{jj}}, 0, \infty),
\end{equation}

\noindent
where 
\begin{equation*}
b_{j} = \sum_{j' \neq j, (j,j') \in E} \omega_{jj'}' s_{jj'}
\end{equation*}

\noindent
and $\text{trunc}N(\mu, \sigma^2, u, v)$ denotes a normal distribution with mean $\mu$ and variance $\sigma^2$, truncated 
in the interval $(u, v)$. Hence, one can generate approximate samples from the refitted posterior density using a 
Gibbs sampling approach. These samples can be subsequently used to generate a posterior credible region. 

As previously discussed, the refitted posterior density is supported on $\mathbb{M}_{\hat{G}}$, and hence the mode/mean 
is not guaranteed to be positive definite. However, our numerical work shows that as long as the sample size is not 
too small -e.g., $n>p/2$-, the resulting estimated precision matrix will actually be positive definite.

In case the sample size was really small, the resulting estimated precision matrix may not be positive definite. In such 
circumstances, a crude but simple solution, if needed, is to ``project'' the posterior mode/mean (and the associated credible 
region) on the space of positive definite matrices by using the following transformation 
$$
B(\mathbf{\Omega}) = \begin{cases}
\mathbf{\Omega} & \mbox{if } \mathbf{\Omega} \mbox{ is positive definite}, \cr
\mathbf{\Omega} -  \lambda_{\min} (\mathbf{\Omega}) I_p + \epsilon I_p & \mbox{ if } \lambda_{\min} (\mathbf{\Omega}) \leq 0, 
\end{cases}
$$

\noindent
where $\epsilon > 0$ is a user-defined small positive number. The function $B$ leaves positive definite matrices in 
$\mathcal{M}_{\hat{G}}$ invariant, and appropriately increases the diagonal entries of matrices that are not 
positive definite. 

As previously mentioned, the main goal of this paper is selecting the sparsity pattern in $\mathbf{\Omega}$. The refitting and 
projection based method developed in this section for estimating the magnitude of the non-zero entries (post sparsity 
selection) performs well in the simulations in Section \ref{simulations}, and indeed reduces bias. Developing more 
sophisticated, and yet computationally effective methods for projecting into the space of (sparse) positive definite matrices is a topic of current research. 

\section{High Dimensional Sparsity Selection Consistency and Convergence Rates for B-CONCORD} 
\label{sec:consistency}

\noindent
In this section, we establish selection and estimation consistency properties for B-CONCORD, under high dimensional scaling wherein the number of variables $p = p_n$ increases with the sample size $n$. The observations ${\bf y}_{1:}^n, {\bf y}_{2:}^n, \cdots, {\bf y}_{n:}^n \in \mathbb{R}^{p_n}$ form an i.i.d. sample from a distribution with mean 
${\bf 0}$ and precision matrix $\mathbf{\Omega}^0_n$. Let $G^0_n = (\{1,2, \cdots, p_n\}, E_{0,n})$ denote the graph 
encoding the sparsity pattern in $\mathbf{\Omega}^0_n$. Let $\boldsymbol{t} = \boldsymbol{t}_n \in \{0,1\}^{{p \choose 2}}$ 
denote the sparsity pattern in the true precision matrix $\mathbf{\Omega}^0_n$, and $d_{\boldsymbol{t}, n}$ denote the 
number of non-zero entries in $\boldsymbol{t}$. For ease of exposition, we will often suppress the dependence of the 
quantities $p_n, \mathbf{\Omega}^0_n, G^0_n, \boldsymbol{t}_n, d_{\boldsymbol{t}, n}$ on $n$, and simply denote them by 
$p, \mathbf{\Omega}^0, G^0, \boldsymbol{t}, d_{\boldsymbol{t}}$, respectively. 

Recall from (\ref{theta}) and (\ref{delta}) that ${\boldsymbol \xi}$ and $\boldsymbol{\delta}$ denote the vectorized versions 
of the off-diagonal and diagonal entries of $\mathbf{\Omega}$. Let ${\boldsymbol \xi}^0$ denote the vectorized 
version of the off-diagonal elements of the true precision matrix $\mathbf{\Omega}^0$. Using (\ref{prior on theta}), 
(\ref{prior on delta}), (\ref{BSSC posterior}), and straightforward calculations, the generalized posterior distribution in 
terms of the $({\boldsymbol \xi}, \boldsymbol{\delta})$ can be expressed as 
\begin{eqnarray}
\pi \left\{ {\boldsymbol \xi}, \boldsymbol{\delta} \mid \mathcal{Y} \right\}
& \propto \exp \left\{ n \boldsymbol{1}'\log \left( \boldsymbol{\delta} \right) - \frac{n}{2}\left[ \left( {\begin{array}{*{20}{c}}
{{\boldsymbol \xi}'} &
{\boldsymbol{\delta}'} \\
\end{array} } \right) \left( {\begin{array}{*{20}{c}}
{\mathbf{\Phi}} & {\mathbf{A}} \\
{\mathbf{A}'} & {\mathbf{D}} \\
\end{array} } \right) \left( {\begin{array}{*{20}{c}}
{{\boldsymbol \xi}} \\
{\boldsymbol{\delta}} \nonumber\\
\end{array} } \right) \right] \right\}  \exp \left(-\frac{{\boldsymbol \xi}'\mathbf{\Lambda}{\boldsymbol \xi}}{2} \right) \\
 &\times \sum\limits_{ \boldsymbol{l} \in \mathcal{L}} \left\{ \frac{|\mathbf{\Lambda}_{\boldsymbol{l}\boldsymbol{l}} |^{\frac{1}{2}}}{\left(2\pi \right)^{\frac{d_{\boldsymbol{l}}}{2}}} I_{\left( {\boldsymbol \xi} \in \mathcal{M}_{\ell}\right)}\left[ q^{d_{\boldsymbol{l}}}(1-q)^{\binom{p}{2} - d_{\boldsymbol{l}}} \right]\right\} \exp \left( -\gamma \boldsymbol{1}'\boldsymbol{\delta} \right). \label{posterior theta and delta}
\end{eqnarray}

\noindent
where, $\mathbf{\Phi}$ is a $\frac{p(p-1)}{2} \times \frac{p(p-1)}{2}$ symmetric matrix which after indexing the rows and 
columns using $(12, 13, ..., p-1p)$, is given by 
\begin{equation}\label{upsilon}
\mathbf{\Phi}_{\left(ab, cd\right)} = \left\{
	\begin{array}{ll}
		s_{aa} + s_{bb}  & \mbox{if } a=b \ \& \ c=d , \\
				s_{ac}  & \mbox{if } b=d \ \& \ a \neq c,\\
						s_{bd}  & \mbox{if } a=c \ \& \ b\neq d,\\
								0  & \mbox{if } a \neq b \ \& \ c \neq d ,\\
	\end{array}
\right. \quad \text{for} \quad 1 \leq a < b \leq p, \ \text{and} \ 1 \leq c < d \leq p. 
\end{equation}

\noindent
For example, when $p=5$, $\mathbf{\Phi}$ is as follows. 
\begingroup\makeatletter\def\f@size{8}\check@mathfonts\begin{equation*}
\left( {\begin{array}{*{20}{c}}
{s_{11} + s_{22}} & {s_{23}} & {s_{24}} & {s_{25}} & {s_{13} } & {s_{14} } & {s_{15} } & {0} & {0} & {0} \\ 
{s_{23} } & {s_{11} +s_{33} } & {s_{34} } & {s_{35} } & {s_{12} } & {0} & {0} & {s_{14} } & {s_{15} } & {0} \\ 
{s_{24} } & {s_{34} } & {s_{11}  + s_{44} } & {s_{45} } & {0} & {s_{12} } & {0} & {s_{13} } & {0} & {s_{15} } \\ 
{s_{25} } & {s_{35} } & {s_{45} } & {s_{11} +s_{55} } & {0} & {0} & {s_{12} } & {0} & {s_{13} } & {s_{14} } \\ 
{s_{13} } & {s_{12} } & {0} & {0} & {s_{22} +s_{33} } & {s_{34} } & {s_{35} } & {s_{24} } & {s_{25} } & {0} \\ 
{s_{14} } & {0} & {s_{12} } & {0} & {s_{34} } & {s_{22} +s_{44} } & {s_{45} } & {s_{23} } & {0} & {s_{25} } \\ 
{s_{15} } & {0} & {0} & {s_{12} } & {s_{35} } & {s_{45} } & {s_{22} +s_{55} } & {0} & {s_{23} } & {s_{24} } \\ 
{0} & {s_{14} } & {s_{13} } & {0} & {s_{24} } & {s_{23} } & {0} & {s_{33} +s_{44} } & {s_{45} } & {s_{35} } \\ 
{0} & {s_{15} } & {0} & {s_{13} } & {s_{25} } & {0} & {s_{23} } & {s_{45} } & {s_{33} +s_{55} } & {s_{34} } \\ 
{0} & {0} & {s_{15} } & {s_{14} } & {0} & {s_{25} } & {s_{24} } & {s_{35} } & {s_{34} } & {s_{44} +s_{55} } \\ 
\end{array} } \right),
\end{equation*}
\endgroup

\noindent
In addition, the vector $\boldsymbol{a}$ is a vector of length $\frac{p(p-1)}{2}$ given by 
\begin{equation}\label{a}
\boldsymbol{a}=(s_{12}(\omega_{11}+\omega_{22}), ..., s_{1p} (\omega_{11}+\omega_{pp}),..., s_{p-1p}(\omega_{p-1p-1}+
\omega_{pp}))',
\end{equation}

\noindent
$\mathbf{A}$ is a $\frac{p(p-1)}{2}\times p$ matrix such that $\mathbf{A}\boldsymbol{\delta} = \boldsymbol{a}$ and 
$\mathbf{D}$ is a $p \times p$ diagonal matrix with entries $\left\{ s_{ii} \right\}_{1 \leq i \leq p}$. Recall that 
$\mathcal{L}$ denotes the space of all the $2^{{p \choose 2}}$ sparsity patterns for ${\boldsymbol \xi}$, and $\mathcal{M}_{\boldsymbol{l}}$ denotes the space in which the parameter ${\boldsymbol \xi}$ varies when restricted to the sparsity pattern $\boldsymbol{l}$. 

Note that our main objective is to correctly select that sparsity pattern in the off-diagonal entries. Hence, as commonly done 
for generalized likelihood based high-dimensional consistency proofs - see \cite{khare2015convex,peng2009partial} - we 
assume the existence of accurate estimates for the diagonal elements, i.e., estimates $ \hat{\omega}_{ii}$ are available, 
such that for any $\eta>0$, there exists a constant $C>0$, such that 
\begin{equation} \label{accurate-estimation}
\max_{ 1 \leq k \leq K} \| \hat{\omega}_{ii} - \omega_{ii}\| \leq C \left( \sqrt{\frac{\log p}{n}}\right),
\end{equation}

\noindent
with probability at least $1-O(n^{-\eta})$. One way to get such estimates of the diagonal entries is discussed in Lemma 4 of 
\cite{khare2015convex}. Denote the resulting estimates of the vectors $\boldsymbol{\delta}$ and $\boldsymbol{a}$ by 
$\hat{\boldsymbol{\delta}}$ and $\hat{\boldsymbol{a}}$, respectively. For the remainder of the section, we assume that the 
entries of $\mathbf{\Lambda}$ are fixed. 

In view of (\ref{posterior theta and delta}), the conditional posterior distribution of the vector of off-diagonal entries 
${\boldsymbol \xi}$ given $\hat{\boldsymbol{\delta}}$ is as follows. 
\begin{equation}\label{posterior theta given delta} 
\begin{split}
&\pi \left\{ {\boldsymbol \xi} |\hat{\boldsymbol{\delta}}, \mathcal{Y}\right\} \propto \exp \left\{ - \frac{1}{2}\left[ 
{\boldsymbol \xi}'\left( n\mathbf{\Phi} + \mathbf{\Lambda}\right) {\boldsymbol \xi} + 2n{\boldsymbol \xi}' 
\hat{\boldsymbol{a}} \right] \right\} \\
&\times \sum\limits_{ \boldsymbol{l} \in \mathcal{L}} \left\{ \frac{|\mathbf{\Lambda}_{\boldsymbol{l}\boldsymbol{l}}|^{\frac{1}{2}}}{\left(2\pi \right)^{\frac{d_{\boldsymbol{l}}}{2}}} I_{\left( {\boldsymbol \xi} \in \mathcal{M}_{\boldsymbol{l}}\right)}  \left[ q^{d_{\boldsymbol{{\boldsymbol{l}}}}}(1-q)^{\binom{p}{2} - d_{\boldsymbol{l}}} \right]\right\},
\end{split}
\end{equation}

\noindent
The above posterior distribution is a mixture distribution, and induces a posterior distribution on the space of sparsity patterns. Straightforward calculations (see proof of Lemma S4 in the Supplemental document) show that 
\begin{equation} \label{eqsparsity}
\pi \left\{ \boldsymbol{l} | \hat{\boldsymbol{\delta}}, \mathcal{Y} \right\} \propto q^{d_{\boldsymbol{l}}}(1-q)^{\binom{p}{2} - d_{\boldsymbol{l}}} \frac{|\mathbf{\Lambda}_{\boldsymbol{l}\boldsymbol{l}}|^{\frac{1}{2}}}{| \left( n\mathbf{\Phi} + \mathbf{\Lambda}\right)_{{\boldsymbol{l}}{\boldsymbol{l}}}|^{\frac{1}{2}}} \exp\left\{ \frac{n^2}{2}\hat{\boldsymbol{a}}_{\boldsymbol{l}}' \left( n\mathbf{\Phi} + \mathbf{\Lambda}\right)_{{\boldsymbol{l}}{\boldsymbol{l}}}^{-1}\hat{\boldsymbol{a}}_{\boldsymbol{l}}\right\}. 
\end{equation}

\noindent
for every ${\boldsymbol \ell} \in \mathcal{L}$. To establish high-dimensional asymptotic 
properties of this posterior distribution on the space of sparsity patterns, the following standard and mild regularity 
assumptions are made.
\begin{Assumption}\label{Assumption2: d_t} 
$(d_{\boldsymbol{t}} + 1) \sqrt{\frac{\log p}{n}} \to 0, \quad \quad \text{as} \quad n \to \infty$. 
\end{Assumption}

\noindent
This assumption essentially states that the number of variables $p$ has to grow slower than 
$e^{(\frac{n}{d_{\boldsymbol{t}}^2})}$. Similar assumptions have been made in other high dimensional covariance estimation 
methods e.g. \cite{banerjee2014posterior}, \cite{banerjee2015bayesian}, \cite{bickel2008regularized}, and 
\cite{xiang2015high}.

\begin{Assumption}\label{Assumption4: Sub-Gaussianity} There exists $c>0$, independent of $n$ such that 
\begin{equation*}
\mathbb{E}_0 \left[ \exp \left( {\boldsymbol{\alpha}'\mathbf{y_{i:}}} \right) \right] \leq \exp \left( c \boldsymbol{\alpha}' 
\boldsymbol{\alpha} \right),
\end{equation*}
\end{Assumption}

\noindent
where $\mathbb{E}_0$ denotes the expected value with respect to the true data generating model. The above assumption 
allows for deviations from normality. Hence, Theorem \ref{theorem (strong consistency)} below establishes that B-CONCORD is robust (in terms of consistency) under misspecification of the data generating distribution, as long as its tails are 
sub-Gaussian. 

\begin{Assumption}\label{Assumption3: Bounded eigenvalues}
(Bounded eigenvalues). There exists $\tilde{\varepsilon_{0}} > 0$, independent of $n$, such that 
\begin{equation*}
\tilde{\varepsilon_{0}} \leq \text{eig}_{\min}\left( {\mathbf{\Omega}}^0 \right) \leq \text{eig}_{\max}\left( {\mathbf{\Omega}}^0 
\right) \leq \frac{1}{\tilde{\varepsilon_{0}}}.
\end{equation*}
\end{Assumption}

\noindent
This is a standard assumption in high dimensional analysis to obtain consistency results; see for example  
\cite{buhlmann2011statistics}. 

\begin{Assumption}\label{Assumption5: Signal Strength}
(Signal Strength). Let $s_n$ be the smallest non-zero entry (in magnitude) in the vector ${\boldsymbol \xi}_0$. We assume 
$\frac{\frac{1}{2}\log n + d_{\boldsymbol{t}}\log p}{n s_n^2} \to 0$.
\end{Assumption}

\noindent
This is again a standard assumption. Similar assumptions on the appropriate signal size can be found in 
\cite{khare2015convex,peng2009partial}. 

\begin{Assumption}\label{Assumption6: Decay rate} (Decay rate of the edge probabilities).
Let $q = p^{-a_2 d_{\boldsymbol{t}}}$, where $a_2 = \frac{16 \max(1,c_0)^2}{\min(1, \tilde{\varepsilon}_0)}$. 
\end{Assumption}

\noindent
Here $c_0$ is a constant (not depending in $n$) that is specified in the proof of Lemma \ref{Lemma 4} in the Supplemental 
Document. This assumption can be interpreted as a priori penalizing sparsity patterns with too many non-zero entries. Next, 
we establish the main posterior consistency result. In particular, we show that the posterior mass assigned to the true sparsity 
pattern converges to one in probability (under the true model), if we restrict to {\it realistic sparsity patterns}, i.e., sparsity 
patterns where the number of non-zero entries is bounded by $\tau_n$, an appropriate constant multiple of 
$\sqrt{\frac{n}{\log p}}$ (see Lemma \ref{Lemma 3} in the Supplemental Document). 
\begin{Theorem}\label{theorem (strong consistency)}
(Strong Selection Consistency) Under Assumptions \ref{Assumption2: d_t} - \ref{Assumption6: Decay rate}, and restricting to 
realistic sparsity patterns, the posterior distribution on the sparsity patterns in (\ref{eqsparsity}) puts all of its mass on the true sparsity pattern ${\bf t}$ as $n \rightarrow \infty$, i.e., 
\begin{enumerate}
\item 
\begin{equation}\label{main result}
\pi \left\{ \boldsymbol{t} | \hat{\boldsymbol{\delta}},\mathcal{Y} \right\} 
\xrightarrow{\text{$\mathbb{P}_0$}} 1, \quad \quad \text{as} \quad n \to \infty.
\end{equation}

\item The sparsity pattern estimate $\hat{\boldsymbol{l}}_{\upsilon, BSSC}$, obtained by using the output of the BSSC 
Algorithm and applying the thresholding approach discussed at the end of Section \ref{posterior:computation:gibbs:sampler}, 
satisfies 
$$
\mathbb{P}_0 \left( \hat{\boldsymbol{l}}_{\upsilon, BSSC} = \boldsymbol{t} \right) \rightarrow 1 \quad \quad \text{as} \quad n 
\to \infty. 
$$
\end{enumerate}
\end{Theorem}

\noindent
While the main objective is sparsity selection consistency, we also derive a result that establish estimation consistency/
convergence rates for the estimates of the magnitudes of the non-zero entries obtained from the refitted posterior density 
defined in (\ref{refitted:posterior:density}). 

For every $1 \leq i \leq p$, let $\nu_i^0$ denote the number of structurally non-zero off-diagonal entries in the $i^{th}$ row (or 
column) of $\mathbf{\Omega}^0$. It follows that $d_{\boldsymbol{t}} = \frac{1}{2} \sum_{i=1}^p \nu_i^0$. Let 
$$
\nu_{\max} = \max_{1 \leq i \leq p} \nu^0_i 
$$

\noindent
denote the maximum number of non-zero entries in any row (or column) of $\mathbf{\Omega}^0$. The relationship between 
$\nu_{\max}$ and $d_{\boldsymbol{t}}$ (total number of structurally non-zero entries in $\mathbf{\Omega}^0$) depends on 
the underlying sparsity structure. At one extreme, $\nu_{\max}$ can be the same order as $d_{\boldsymbol{t}}$ (eg. star 
graph), while at the other extreme it can be as small as $O(d_{\boldsymbol{t}}/p)$ (eg. banded matrix). For the posterior 
convergence rate result, we use the same assumptions as Theorem \ref{theorem (strong consistency)}, except 
Assumption \ref{Assumption2: d_t}, which is replaced by the slightly stronger assumption below. 
\begin{Assumption}\label{Assumption2new}
$(d_{\boldsymbol{t}} + \sqrt{d_{\boldsymbol{t}}} \nu_{\max} + 1) \sqrt{\frac{\log p}{n}} \to 0, \quad \quad \text{as} \quad n \to 
\infty$. 
\end{Assumption}


\noindent
In many settings, where $\nu_{\max} = O(\sqrt{d_{\boldsymbol{t}}})$, Assumption \ref{Assumption2: d_t} and Assumption 
\ref{Assumption2new} are equivalent. 
\begin{Theorem} \label{estimation:consistency:refitted:posterior}
(Estimation Consistency and convergence rate for refitted posterior) Under Assumptions \ref{Assumption4: Sub-Gaussianity} - 
\ref{Assumption2new}, the refitted posterior density $\pi_{refitted}$ in (\ref{refitted:posterior:density}) 
satisfies 
$$
\mathbb{E}_{0} \left[ \pi_{refitted} \left( \left\| \hat{\mathbf{\Omega}} - \mathbf{\Omega}^0 \right\|_{\max} > K 
\nu_{\max} \sqrt{\frac{d_{\boldsymbol{t}} \log p}{n}} \right) \right] \rightarrow 0 \quad \quad \text{as} \quad n \quad \to \infty, 
$$

\noindent
for a large enough constant $K$ (not depending on $n$), and 
$$
\mathbb{E}_{0} \left[ \pi_{refitted} \left( \left\| \hat{\mathbf{\Omega}} - \mathbf{\Omega}^0 \right\| > K \nu_{\max}^2  
\sqrt{\frac{d_{\boldsymbol{t}} \log p}{n}} \right) \right] \rightarrow 0 \quad \quad \text{as} \quad n \quad \to \infty. 
$$
\end{Theorem}

\noindent
Here $\| \cdot \|_{\max}$ denotes the sup norm of a matrix (magnitude of entry with largest absolute value), and 
$\| \cdot \|$ denotes the operator norm of a matrix. The proofs of the two results above are provided in the Section S3 of 
the Supplement. 


\begin{Remark}\label{implication of assumption 1}
Posterior estimation  consistency/contraction rates for some pseudo-likelihood based Bayesian approaches have been 
studied in recent work by \cite{atchade2017contraction}. In particular, a {\bf discrete binary} graphical model (as opposed 
to a partial correlation network for continuous variables in this study) is one of the models considered in 
\cite{atchade2017contraction}. A posterior contraction rate of $\sqrt{\frac{(p + d_t) \log p}{n}}$ for $\Omega$ in the 
Frobenius norm is obtained without making an assumption similar to our assumption of accurate diagonal estimates. Note 
however that results in \cite{atchade2017contraction} do not address model selection consistency. 
\end{Remark}

\section{Performance Evaluation of B-CONCORD}\label{simulations}

\noindent
We assess the accuracy, computational speed and scalability of the BSSC Algorithm. As mentioned earlier, the main 
challenge with existing Bayesian procedures is their limited scalability. To the best of our knowledge, 
Stochastic Search Structure Learning (SSSL) introduced in \cite{Wang:2015} is the fastest Bayesian procedure available
for the problem at hand (see Section 5 of \cite{Wang:2015}). Hence, we use SSSL as a benchmark for the performance of our 
BSSC Algorithm. 

\subsection{Computational scalability: Timing and memory requirement comparison} \label{computation}

\noindent
For this task, we set $p \in \{150, 300, 500, 1000, 3000\}$. For each $p$, the {\it true} 
precision matrix $\mathbf{\Omega}^0$ is generated so as to exhibit a complete random sparsity patterns 
with $4 \%$ density of non-zero entries. 
The non-zero off-diagonal entries are generated from a Uniform distribution in the interval $[-0.6, -0.4] \cup [0.4, 0.6]$, and the diagonal entries are adjusted as needed to make the resulting precision matrix positive definite. 
We then generate $25$ data sets of size $n = p/2$ from a 
multivariate Gaussian distribution with mean ${\bf 0}$ and precision matrix $\mathbf{\Omega}^0$. 
The BSSC Algorithm (with default hyperparameter values as discussed in Remark \ref{hyperparameter:selection}), the SSSL
procedure (with default hyperparameter values as specified in \cite{Wang:2015} and graphical lasso (Glasso) (with the final value of its shrinkage parameter manually selected using exhaustive cross-validations) are used to obtain estimates of 
$\mathbf{\Omega}^0$. 

For the BSSC 
Algorithm \ref{GibbsBSSC}, 2000 iterations are used for burn-in, and 2000 more iterations to generate our estimates (standard diagnostics indicate that so many iterations are sufficient for convergence). For SSSL and with $p \in \{150, 300\}$, we used the default setting in \cite{Wang:2015} of 1000 iterations for burn-in, and 10000 iterations for computing posterior estimates; however, for higher dimensions -e.g. $p \in \{500, 1000, 3000\}$- the default setting required more than 50 GigaBytes of memory and therefore only 500 iterations for burn-in and 500 iterations for computing posterior estimates were used. Finally, the results for Glasso were achieved using the default number of iterations for convergence. The two Bayesian algorithms are compared based on {\it computing time required per iteration}. The simulations were performed using 
dedicated cores at the High Performance Computing cluster at the University of Florida. 

Table \ref{wall_time} depicts {\it computing time required per iteration} for BSSC, SSSL, and Glasso averaged over the 25 replicate data sets. The results demonstrate that SSSL becomes expensive even for $p = 500$, while BSSC easily handles settings with $p\geq 1000$. Each BSSC iterate requires less than 44 seconds, which is a fraction (1/5760) of an SSSL iterate. Further, each BSSC iterate is faster than a Glasso one, thus making the two procedures comparable in total computational time required for estimating precision matrices with $p\leq 300$. However, Glasso has an overall time execution advantage for larger $p$, since it requires on average less than 10 iterations to converge and provide estimates.

\begin{table}[h]
\centering 
\caption{Average wall-clock seconds per iteration for BSSC, Glasso, and SSSL for estimating a $p \times p$ precision matrix 
with $p \in (150, 300, 500, 1000 \ \text{and} \ 3000)$.}
\resizebox{0.5\textwidth}{!}{%
\begin{tabular}{lccccc}
\toprule
 &{$p=150$}&{$p=300$}&{$p=500$}&{$p=1000$} &{$p=3000$}\\
\midrule
{BSC} & {0.006} & {0.035} & {0.112} & {0.675} &{43.122}  \\
\midrule
{Glasso} & {0.108} & {0.293} & {0.523 } & {5.857} &{168.333} \\
\midrule
{SSSL} & {0.104} & {1.31} & {10.1} & {37.38} & {237600} \\
\bottomrule
\end{tabular}}
\label{wall_time}
\end{table}

In addition, the memory requirement for BSSC is significantly smaller than SSSL. The average memory used by BSSC and SSSL for 
different values of $p$ is summarized in Table \ref{memory usage}. Note that SSSL requires more than 50 GB for $p \geq 500$ while BSSC achieves the goal with 0.24 GB of memory. 
\begin{table}[h]
\centering 
\caption{Average memory usage (in GigaBytes) for BSSC and SSSL for estimating a $p \times p$ precision matrix with $p \in 
(150, 300, 500, 1000 \ \text{and} \ 3000)$.}
\resizebox{0.5\textwidth}{!}{%
\begin{tabular}{lccccc}
\toprule
 &{$p=150$}&{$p=300$}&{$p=500$}&{$p=1000$}&{$p=3000$}\\
\midrule
{BSC} & {0.001} & {0.22} & {0.24} & {0.32} & {1.39}  \\
\midrule
{SSSL} & {5.3} & {20.4} & {$>$50 } & {$>$50}& {$>$50} \\
\bottomrule
\end{tabular}}
\label{memory usage}
\end{table}

\noindent
The reason for the significantly superior performance of BSSC compared to SSSL is discussed next.
The SSSL algorithm requires $2(p-1)$ matrix inversions of 
$(p-1) \times (p-1)$ matrices (see Section 4.1 of \cite{Wang:2015}). The worst case computational complexity for one 
iteration of the SSSL algorithm therefore is $O(p^4)$. On the other hand, one iteration of the BSSC 
algorithm has computational complexity $O(p^3)$ (update each of the ${p \choose 2}$ entries with $O(p)$ computations, 
see Algorithm \ref{GibbsBSSC}), and requires {\it no matrix inversions}. For sparse precision matrices,
computational complexity in practice is better for both methods than the above worse case scenarios. Nevertheless,
the numerical results presented amply demonstrate the superior performance of BSSC. Note that in sparse settings, it becomes
faster to compute inner products of the form $\mathbf{\Omega}_{-jk}'\mathbf{S}_{-jj}$ needed by BSSC. 

\subsection{Estimation accuracy comparison} \label{accuracy}

\noindent
The standard performance metrics of specificity (SP), sensitivity (SE) and Matthews Correlation Coefficient (MCC), defined next, are used.
\begin{equation}\label{accuracy measures}
\begin{split}
\text{SP} &= \frac{\text{TN}}{\text{TN} + \text{FP}}, \quad \quad \quad \text{SE} = \frac{\text{TP}}{\text{TP} + \text{FN}}\\
\text{MCC} &= \frac{\text{TP} \times \text{TN} - \text{FP} \times \text{FN}}{\sqrt{(\text{TP} + \text{FP})(\text{TP} + \text{FN})(\text{TN} + \text{FP})(\text{TN} + \text{FN})}}
\end{split}
\end{equation}

\noindent
where TP, TN, FP and FN represent the number of true positives, true negatives, false positives and false negatives, 
respectively. Larger values of any of the above metrics indicate better sparsity selection obtained by the corresponding algorithm. Precision matrices of dimension $p \in (150, 300, 500, \text{and} \ 1000)$, and 
sample size $n \in (p/2, 3p/4, p, \text{and} \ 2p)$ are considered. Further, the proportion of non-zero upper off-diagonal entries of $\mathbf{\Omega}^0$ is set to 0.04 and 0.1. The true precision matrix is generated according to the same mechanism
as in Section \ref{computation}. For each combination of $p$, $n$, and edge density level, we generate $50$ data sets of size $n$ from a multivariate normal distribution with mean ${\bf 0}$ and precision matrix $\mathbf{\Omega}^0$. For each data set, we estimate the sparsity pattern using BSSC and SSSL, and subsequently calculate the SP, SE, MCC measures.
The same number of burn-in (2000) and estimation (2000) iterations are used for BSSC. For small $p \in \{150, 300\}$, we used the default setting in \cite{Wang:2015} of 1000 iterations for burn-in, and 10000 iterations for computing posterior estimates. For the reasons previously discussed, we only used 500 iterations for burn-in and 500 iterations for estimation purposes for $p=500, 1000$.

The SP, SE, MCC values, averaged over $50$ replicate 
data sets, are provided in Table \ref{comparing SSSL and BSSC}. Overall, the MCC values indicate better sparsity selection 
achieved by BSSC compared to SSSL, when the density of non-zeros entries is $0.1$, while the results are 
comparable for $0.04$ density. In summary, BSSC outperforms SSSL both in terms of estimation accuracy and especially 
of computational requirements on execution time and memory.
\begin{table}[H]
\centering 
\caption{Average sparsity selection accuracy for BSSC and SSSL for estimating a $p \times p$ precision matrix with 
$p \in (150, 300, 500, \text{and} \ 1000)$.}
\resizebox{0.8\textwidth}{!}{%
\begin{tabular}{lccccccccccc}
\toprule
 &&&&{BSSC}&&&&{SSSL}&\\
\cmidrule(lr){3-6} 
\cmidrule(lr){7-10} 
 & & {$p=150$} & {$p=300$} & {$p=500$} & {$p=1000$} & {$p=150$} & {$p=300$} & {$p=500$} & {$p=1000$} \\
\cmidrule(lr){3-6} 
\cmidrule(lr){7-10}   
$n$ &&&&{$Density = 0.04$}&&&&{$Density = 0.04$}&\\
\cmidrule(lr){1-2}  
\cmidrule(lr){3-6} 
\cmidrule(lr){7-10} 
{\multirow{3}{*}{$p/2$} } 
& {SP\%}& {99} & {100} & {100} & {100} & {100} & {100} & {100} & {100}   \\
& {SE\%}& {26} & {28} & {29} & {30} & {20} & {22} & {17} & {17}   \\
& {MC\%}& {41} & {44} & {45} & {46} & {39} & {42} & {37} & {36}   \\
\cmidrule(lr){1-2}  
\cmidrule(lr){3-6} 
\cmidrule(lr){7-10} 
{\multirow{3}{*}{$3p/4$} } 
& {SP\%}& {99} & {100} & {100} & {100} & {100} & {100} & {100} & {100}    \\
& {SE\%}& {41} & {43} & {45} & {47} & {34} & {35} & {29} & {28}   \\
& {MC\%}& {55} & {58} & {60} & {62} & {54} & {56} & {50} & {49} \\
\cmidrule(lr){1-2}  
\cmidrule(lr){3-6} 
\cmidrule(lr){7-10} 
{\multirow{3}{*}{$p$} } 
& {SP\%}& {100} & {100} & {100} & {100} & {100} & {100} & {100} & {100}   \\
& {SE\%}& {53} & {57} & {60} & {61} & {45} & {48} & {41} & {38}   \\
& {MC\%}& {67} & {70} & {71} & {73} & {64} & {67} & {61} & {58}  \\
\cmidrule(lr){1-2}  
\cmidrule(lr){3-6} 
\cmidrule(lr){7-10} 
{\multirow{3}{*}{$2p$} } 
& {SP\%}& {100} & {100} & {100} & {100} & {100} & {100} & {100} & {100}   \\
& {SE\%}& {86} & {89} & {91} & {92} & {78} & {85} & {79} & {74}    \\
& {MC\%}& {89} & {91} & {92} & {93} & {87} & {91} & {88} & {85}   \\
\midrule
$n$ &&&&{$Density = 0.1$}&&&&{$Density = 0.1$}&\\
\cmidrule(lr){1-2}  
\cmidrule(lr){3-6} 
\cmidrule(lr){7-10} 
{\multirow{3}{*}{$p/2$} } 
& {SP\%}& {99} & {99} & {99} & {99} & {100} & {100} & {100} & {100}    \\
& {SE\%}& {12} & {11} & {10} & {11} & {9} & {4} & {6} & {5}    \\
& {MC\%}& {24} & {24} & {23} & {24} & {23} & {16} & {19} & {17}   \\
\cmidrule(lr){1-2}  
\cmidrule(lr){3-6} 
\cmidrule(lr){7-10} 
{\multirow{3}{*}{$3p/4$} } 
& {SP\%}& {99} & {99} & {99} & {99} & {100} & {100} & {100} & {100}    \\
& {SE\%}& {18} & {17} & {17} & {18} & {15} & {9} & {11} & {9}   \\
& {MC\%}& {33} & {33} & {32} & {33} & {33} & {25} & {28} & {24}    \\
\cmidrule(lr){1-2}  
\cmidrule(lr){3-6} 
\cmidrule(lr){7-10} 
{\multirow{3}{*}{$p$} } 
& {SP\%}& {99} & {99} & {99} & {99} & {100} & {100} & {100} & {100}   \\
& {SE\%}& {25} & {24} & {23} & {25} & {21} & {13} & {16} & {13}   \\
& {MC\%}& {40} & {41} & {40} & {41} & {40} & {31} & {34} & {29}   \\
\cmidrule(lr){1-2}  
\cmidrule(lr){3-6} 
\cmidrule(lr){7-10} 
{\multirow{3}{*}{$2p$} } 
& {SP\%}& {99} & {99} & {99} & {99} & {100} & {100} & {100} & {100}    \\
& {SE\%}& {49} & {49} & {48} & {51} & {44} & {30} & {35} & {27}   \\
& {MC\%}& {63} & {64} & {63} & {65} & {62} & {51} & {55} & {47}   \\
\bottomrule
\end{tabular}}
\label{comparing SSSL and BSSC}
\end{table}

\noindent
Next, we assess the effectiveness of the refitting technique developed in 
Section \ref{refitting:technique} for reducing the estimation bias in the magnitude of the non-zero entries. This is
accomplished by examining improvements in the relative error of the Frobenius norm, namely $\frac{\|\hat{\mathbf{\Omega}} - 
\mathbf{\Omega}^0\|_{F}}{\mathbf{\|\Omega}^0\|_{F}}$, for the original and refitted estimates. Specifically, let 
$\hat{E}$ denote the collection of indices selected as non-zero using the thresholding procedure described at the end of 
Section \ref{posterior:computation:gibbs:sampler}, and $\left\{ \hat{\mathbf{\Omega}}^(t) \right\}_{t=1}^T$ denote 
the sequence of iterates obtained by running the Gibbs sampler in Algorithm \ref{GibbsBSSC}. Our first estimate of 
$\mathbf{\Omega}$ is given by 
\begin{equation} \label{estimate}
\hat{\mathbf{\Omega}}_{BSSC} = \begin{cases}
\frac{\sum_{t=1}^T \hat{\omega}_{jk}^{(t)}}{\sum_{t=1}^T 1_{\{\hat{\omega}_{jk}^{(t)} \neq 0\}}} & \mbox{if } (j,k) \in \hat{E}, 
\mbox{ or } j=k, \cr
0 & \mbox{if } (j,k) \notin \hat{E}. 
\end{cases}
\end{equation}

\noindent
The second estimate, denoted by $\hat{\mathbf{\Omega}}_{refitted}$, is the posterior mean of the refitted posterior density 
in (\ref{refitted:posterior:density}), which can again be computed by the modified Gibbs sampling procedure on 
$\mathbb{M}_{\hat{G}}$ discussed at the end of Section \ref{refitting:technique}. Note that both 
$\hat{\mathbf{\Omega}}_{BSSC}$ and $\hat{\mathbf{\Omega}}_{refitted}$ set the indices not in $\hat{E}$ to be zero, 
and differ only in the magnitudes of indices classified as non-zero, i.e., indices in $\hat{E}$. We consider settings with
$p \in (50, 100, 150, 200, \text{and} \ 300)$, and $n= p$. The true precision matrix $\mathbf{\Omega}^0$ is generated by 
using the same mechanism as in Section \ref{computation}, and for each $p$, the relative Frobenius norm of the two estimates described above are averaged over the $50$ replicate data sets. The results are shown in Table \ref{debiasing sim}, and clearly 
demonstrate the improvement obtained by refitting, especially for larger values of $p$. 
\begin{table}[h]
\centering 
\caption{Summary of average relative error for estimation of the magnitudes of the precision matrix entries for estimates 
directly obtained from Algorithm \ref{GibbsBSSC} vs. estimates obtained by using the refitting technique in 
Section \ref{refitting:technique}}
\resizebox{0.8\textwidth}{!}{%
\begin{tabular}{lccccc}
\toprule
 &{$p=n=50$}&{$p=n=100$}&{$p=n=150$}&{$p=n=200$}&{$p=n=300$}\\
\midrule
{BSSC} & {0.38} & {0.45} & {0.45 } & {0.45 } &  {0.49} \\
\midrule
{BSSC with refitting} & {0.36} & {0.35} & {0.32 } & {0.32 } &  {0.29} \\
\bottomrule
\end{tabular}}
\label{debiasing sim}
\end{table}

\subsection{An Application of B-CONCORD to Inflammatory Bowel Disease Metabolomics data}

There are a number of factors that impact the stool metabolome, including diet, gut flora and gut function. The data analyzed next come from 208 female subjects with inflammatory bowel disease (IBD) -which includes Crohn’s disease and ulcerative colitis, affect several million individuals worldwide- that participated in the Integrative Human Microbiome Project (iHMP) and were extracted from the Metabolomics Workbench \texttt{www.metabolomicsworkbench.org} (Study ID ST000923).
The data correspond to measurements of 428 primary (directly involved in normal growth, development, and reproduction 
cellular processes) and secondary (produced by bacteria, fungi, etc.) metabolites and lipids (fatty acids and their
derivatives). Specifically, 240 primary, 49 secondary and 139 lipids were profiled by a mass spectrometry analytical platform.

The B-CONCORD methodology was employed to estimate the interaction networks of these compounds and the results are shown
in Table \ref{tab:interaction}. It can be seen that there are relatively strong interactions between primary and secondary 
metabolites and also between primary metabolites and lipids.

\begin{table}
\centering
\begin{tabular}{l|ccc} \hline
& Secondary & Primary & Lipids \\ \hline\hline
Secondary & 32 & 337 & 65 \\
Primary & 337 & 935 & 355 \\
Lipids & 65 & 355 & 412 \\ \hline\hline
\end{tabular}
\caption{Interactions between primary, secondary metabolites and lipids in IBD specimens \label{tab:interaction}}
\end{table}

Next, we comment on some specific patterns that align with findings in the literature.
We observed that the short chain fatty acids (acetate, butyrate and propionate) have a high number of connections ($\sim 22$
on average and significantly higher than other compounds), a result consistent with their function as the main source of energy for cells lining the colon and impacting the latter's health \cite{wong2006colonic}. Further, 
the primary bile acid cholate and its glycine and taurine conjugates (glycocholate, taurocholate), as well as secondary bile acids (lithocholate and deoxycholate), were also strongly connected ($\sim 15$ connections on average) and are known to
play a role in IBD \cite{tiratterra2018role,lavelle2020gut}. Finally, sphingolipids (ceramides, phoshpocolines and sphingomyelins) form connected clusters, since they represent structural components of intestinal cell membranes and are also signaling molecules involved in cell fate decisions \cite{abdel2016fostering}.

In general, the proposed model identifies numerous interesting interactions in this rich data set that could provide insights
on how they impact molecular processes in IBD.

\section{Discussion}

This article proposes a fast scalable Bayesian framework for estimating interaction networks through Gaussian
graphical models. The use of a generalized likelihood function in combination with a spike-and-slab prior distribution
on the model parameters leads to closed form expressions for the corresponding conditional posterior distributions,
thus enabling a fast Gibbs sampler for calculating the posterior distribution. The framework also comes with statistical
guarantees on the consistency of the posterior distribution under mild regularity conditions. Another key contribution is
the introduction of a modified prior distribution that is applicable to the identified network (graphical model) from the
data, which provides improved estimates of the magnitudes of the edges in the interaction network. The provided
numerical work renders support to the strong gains in the performance of the methods vis-a-vis competing procedures.


\addcontentsline{toc}{chapter}{References}
\bibliographystyle{plainnat}
\bibliography{reference}

\begin{thebibliography}{38}
\providecommand{\natexlab}[1]{#1}
\providecommand{\url}[1]{\texttt{#1}}
\expandafter\ifx\csname urlstyle\endcsname\relax
  \providecommand{\doi}[1]{doi: #1}\else
  \providecommand{\doi}{doi: \begingroup \urlstyle{rm}\Url}\fi

\bibitem[Abdel~Hadi et~al.(2016)Abdel~Hadi, Di~Vito, and
  Riboni]{abdel2016fostering}
Loubna Abdel~Hadi, Clara Di~Vito, and Laura Riboni.
\newblock Fostering inflammatory bowel disease: sphingolipid strategies to join
  forces.
\newblock \emph{Mediators of inflammation}, 2016, 2016.

\bibitem[Atchad{\'e} et~al.(2017)]{atchade2017contraction}
Yves~A Atchad{\'e} et~al.
\newblock On the contraction properties of some high-dimensional
  quasi-posterior distributions.
\newblock \emph{The Annals of Statistics}, 45\penalty0 (5):\penalty0
  2248--2273, 2017.

\bibitem[Banerjee et~al.(2008)Banerjee, Ghaoui, and
  d'Aspremont]{banerjee2008model}
Onureena Banerjee, Laurent~El Ghaoui, and Alexandre d'Aspremont.
\newblock Model selection through sparse maximum likelihood estimation for
  multivariate gaussian or binary data.
\newblock \emph{Journal of Machine learning research}, 9\penalty0
  (Mar):\penalty0 485--516, 2008.

\bibitem[Banerjee and Ghosal(2014)]{banerjee2014posterior}
Sayantan Banerjee and Subhashis Ghosal.
\newblock Posterior convergence rates for estimating large precision matrices
  using graphical models.
\newblock \emph{Electronic Journal of Statistics}, 8\penalty0 (2):\penalty0
  2111--2137, 2014.

\bibitem[Banerjee and Ghosal(2015)]{banerjee2015bayesian}
Sayantan Banerjee and Subhashis Ghosal.
\newblock Bayesian structure learning in graphical models.
\newblock \emph{Journal of Multivariate Analysis}, 136:\penalty0 147--162,
  2015.

\bibitem[Besag(1975)]{Besag:1975}
J.~Besag.
\newblock Statistical analysis of non-lattice data.
\newblock \emph{Journal of the Royal Statistical Society: Series D (The
  Statistician)}, 24:\penalty0 179--195, 1975.

\bibitem[Bhattacharya et~al.(2015)Bhattacharya, Pati, Pillai, and
  Dunson]{bhattacharya2015dirichlet}
Anirban Bhattacharya, Debdeep Pati, Natesh~S Pillai, and David~B Dunson.
\newblock Dirichlet--laplace priors for optimal shrinkage.
\newblock \emph{Journal of the American Statistical Association}, 110\penalty0
  (512):\penalty0 1479--1490, 2015.

\bibitem[Bickel and Levina(2008)]{bickel2008regularized}
Peter~J Bickel and Elizaveta Levina.
\newblock Regularized estimation of large covariance matrices.
\newblock \emph{The Annals of Statistics}, pages 199--227, 2008.

\bibitem[B{\"u}hlmann and Van De~Geer(2011)]{buhlmann2011statistics}
Peter B{\"u}hlmann and Sara Van De~Geer.
\newblock \emph{Statistics for high-dimensional data: methods, theory and
  applications}.
\newblock Springer Science \& Business Media, 2011.

\bibitem[Cai et~al.(2011)Cai, Liu, and Luo]{CLL:2011}
T.~Cai, W.~Liu, and X.~Luo.
\newblock A constrained l1 minimization approach to sparse precision matrix
  estimation.
\newblock \emph{Journal of the American Statistical Association}, 106:\penalty0
  594--607, 2011.

\bibitem[Carvalho et~al.(2010)Carvalho, Polson, and
  Scott]{carvalho2010horseshoe}
Carlos~M Carvalho, Nicholas~G Polson, and James~G Scott.
\newblock The horseshoe estimator for sparse signals.
\newblock \emph{Biometrika}, 97\penalty0 (2):\penalty0 465--480, 2010.

\bibitem[Chen et~al.(2015)Chen, Witten, and Shojaie]{chen2015selection}
Shizhe Chen, Daniela~M Witten, and Ali Shojaie.
\newblock Selection and estimation for mixed graphical models.
\newblock \emph{Biometrika}, 102\penalty0 (1):\penalty0 47--64, 2015.

\bibitem[Cheng and Lenkoski(2012)]{Cheng:Lenkoski:2012}
Y.~Cheng and A.~Lenkoski.
\newblock Hierarchical gaussian graphical models: Beyond reversible jump.
\newblock \emph{Electronic Journal of Statistics}, 6:\penalty0 2309--2331,
  2012.

\bibitem[Dobra et~al.(2011)Dobra, Lenkoski, and Rodriguez]{DLR:2011}
A.~Dobra, A.~Lenkoski, and A.~Rodriguez.
\newblock Bayesian inference for general gaussian graphical models with
  application to multivariate lattice data.
\newblock \emph{Journal of the American Statistical Association}, 106:\penalty0
  1418--1433, 2011.

\bibitem[Friedman et~al.(2008)Friedman, Hastie, and
  Tibshirani]{friedman2008sparse}
Jerome Friedman, Trevor Hastie, and Robert Tibshirani.
\newblock Sparse inverse covariance estimation with the graphical lasso.
\newblock \emph{Biostatistics}, 9\penalty0 (3):\penalty0 432--441, 2008.

\bibitem[Javanmard and Montanari(2014)]{Javanmard:Montanari:2014}
A.~Javanmard and A.~Montanari.
\newblock Confidence intervals and hypothesis testing for high-dimensional
  regression.
\newblock \emph{Journal of Machine Learning Research}, 15:\penalty0 2869--2909,
  2014.

\bibitem[Khare et~al.(2015)Khare, Oh, and Rajaratnam]{khare2015convex}
Kshitij Khare, Sang-Yun Oh, and Bala Rajaratnam.
\newblock A convex pseudolikelihood framework for high dimensional partial
  correlation estimation with convergence guarantees.
\newblock \emph{Journal of the Royal Statistical Society: Series B (Statistical
  Methodology)}, 77\penalty0 (4):\penalty0 803--825, 2015.

\bibitem[Lavelle and Sokol(2020)]{lavelle2020gut}
Aonghus Lavelle and Harry Sokol.
\newblock Gut microbiota-derived metabolites as key actors in inflammatory
  bowel disease.
\newblock \emph{Nature Reviews Gastroenterology \& Hepatology}, pages 1--15,
  2020.

\bibitem[Ma and Michailidis(2016)]{ma2016joint}
Jing Ma and George Michailidis.
\newblock Joint structural estimation of multiple graphical models.
\newblock \emph{The Journal of Machine Learning Research}, 17\penalty0
  (1):\penalty0 5777--5824, 2016.

\bibitem[Makalic and Schmidt(2016)]{makalic2016simple}
Enes Makalic and Daniel~F Schmidt.
\newblock A simple sampler for the horseshoe estimator.
\newblock \emph{IEEE Signal Processing Letters}, 23\penalty0 (1):\penalty0
  179--182, 2016.

\bibitem[Meinshausen and Buhlmann(2006)]{Meinshausen:Buhlmann:2006}
N.~Meinshausen and P.~Buhlmann.
\newblock High dimensional graphs and variable selection with the lasso.
\newblock \emph{Annals of Statistics}, 34:\penalty0 1436--1462, 2006.

\bibitem[Neal(2003)]{neal2003slice}
Radford~M Neal.
\newblock Slice sampling.
\newblock \emph{The annals of statistics}, 31\penalty0 (3):\penalty0 705--767,
  2003.

\bibitem[Omre and Halvorsen(1989)]{omre1989bayesian}
Henning Omre and Kjetil~B Halvorsen.
\newblock The bayesian bridge between simple and universal kriging.
\newblock \emph{Mathematical Geology}, 21\penalty0 (7):\penalty0 767--786,
  1989.

\bibitem[Park and Casella(2008)]{park2008bayesian}
Trevor Park and George Casella.
\newblock The bayesian lasso.
\newblock \emph{Journal of the American Statistical Association}, 103\penalty0
  (482):\penalty0 681--686, 2008.

\bibitem[Peng et~al.(2009)Peng, Wang, Zhou, and Zhu]{peng2009partial}
Jie Peng, Pei Wang, Nengfeng Zhou, and Ji~Zhu.
\newblock Partial correlation estimation by joint sparse regression models.
\newblock \emph{Journal of the American Statistical Association}, 104\penalty0
  (486):\penalty0 735--746, 2009.

\bibitem[Polson and Scott(2010)]{polson2010shrink}
Nicholas~G Polson and James~G Scott.
\newblock Shrink globally, act locally: Sparse bayesian regularization and
  prediction.
\newblock \emph{Bayesian statistics}, 9:\penalty0 501--538, 2010.

\bibitem[Polson and Scott(2012)]{polson2012local}
Nicholas~G Polson and James~G Scott.
\newblock Local shrinkage rules, l{\'e}vy processes and regularized regression.
\newblock \emph{Journal of the Royal Statistical Society: Series B (Statistical
  Methodology)}, 74\penalty0 (2):\penalty0 287--311, 2012.

\bibitem[Rudelson and Vershynin(2013)]{rudelson2013hanson}
Mark Rudelson and Roman Vershynin.
\newblock Hanson-wright inequality and sub-gaussian concentration.
\newblock \emph{Electronic Communications in Probability}, 18, 2013.

\bibitem[Tiratterra et~al.(2018)Tiratterra, Franco, Porru, Katsanos,
  Christodoulou, and Roda]{tiratterra2018role}
Elisa Tiratterra, Placido Franco, Emanuele Porru, Konstantinos~H Katsanos,
  Dimitrios~K Christodoulou, and Giulia Roda.
\newblock Role of bile acids in inflammatory bowel disease.
\newblock \emph{Annals of gastroenterology}, 31\penalty0 (3):\penalty0 266,
  2018.

\bibitem[Van De~Geer(2019)]{VanDeGeer:2019}
Sara Van De~Geer.
\newblock On the asymptotic variance of the debiased lasso.
\newblock \emph{Electronic Journal of Statistics}, 13:\penalty0 2970--3008,
  2019.

\bibitem[Wainwright(2019)]{wainwright2019high}
Martin~J Wainwright.
\newblock \emph{High-dimensional statistics: A non-asymptotic viewpoint},
  volume~48.
\newblock Cambridge University Press, 2019.

\bibitem[Wand et~al.(2011)Wand, Ormerod, Padoan, and
  Fr{\"u}hwirth]{wand2011mean}
Matthew~P Wand, John~T Ormerod, Simone~A Padoan, and Rudolf Fr{\"u}hwirth.
\newblock Mean field variational bayes for elaborate distributions.
\newblock \emph{Bayesian Analysis}, 6\penalty0 (4):\penalty0 847--900, 2011.

\bibitem[Wang(2015)]{Wang:2015}
H.~Wang.
\newblock Scaling it up: Stochastic search structure learning in graphical
  models.
\newblock \emph{Bayesian Analysis}, 10:\penalty0 351--377, 2015.

\bibitem[Wang et~al.(2012)]{wang2012bayesian}
Hao Wang et~al.
\newblock Bayesian graphical lasso models and efficient posterior computation.
\newblock \emph{Bayesian Analysis}, 7\penalty0 (4):\penalty0 867--886, 2012.

\bibitem[Wong et~al.(2006)Wong, De~Souza, Kendall, Emam, and
  Jenkins]{wong2006colonic}
Julia~MW Wong, Russell De~Souza, Cyril~WC Kendall, Azadeh Emam, and David~JA
  Jenkins.
\newblock Colonic health: fermentation and short chain fatty acids.
\newblock \emph{Journal of clinical gastroenterology}, 40\penalty0
  (3):\penalty0 235--243, 2006.

\bibitem[Xiang et~al.(2015)Xiang, Khare, and Ghosh]{xiang2015high}
Ruoxuan Xiang, Kshitij Khare, and Malay Ghosh.
\newblock High dimensional posterior convergence rates for decomposable
  graphical models.
\newblock \emph{Electronic Journal of Statistics}, 9\penalty0 (2):\penalty0
  2828--2854, 2015.

\bibitem[Yuan and Lin(2007)]{yuan2007model}
Ming Yuan and Yi~Lin.
\newblock Model selection and estimation in the gaussian graphical model.
\newblock \emph{Biometrika}, 94\penalty0 (1):\penalty0 19--35, 2007.

\bibitem[Zhang and Zhang(2014)]{Zhang:Zhang:2014}
C.H. Zhang and S.~Zhang.
\newblock Confidence intervals for low dimensional parameters in high
  dimensional linear models.
\newblock \emph{Journal of the Royal Statistical Society: Series B (Statistical
  Methodology)}, 76:\penalty0 217--242, 2014.

\end{thebibliography}
\normalsize

\newpage


\setcounter{equation}{0}
\setcounter{figure}{0}
\setcounter{Lemma}{0}
\setcounter{table}{0}
\setcounter{section}{0}
\setcounter{page}{1}
\renewcommand{\thesection}{S \arabic{section}}
\renewcommand{\theequation}{S \arabic{equation}}
\renewcommand{\thefigure}{S \arabic{figure}}
\renewcommand{\theLemma}{S \arabic{Lemma}}

\numberwithin{equation}{section}

\section{Appendix A: Proof of Theorem 1}\label{Proof of Theorem 1}

\noindent
By Assumption 3 and the Hanson-Wright inequality \citep{rudelson2013hanson}, there exists a $c_1 >0$, independent of $n$ such 
that
\begin{equation*}\label{PC1n}
P\left\{ \max\limits_{i,j}\| s_{ij} - \sigma_{ij} \| < c_1 \sqrt{\frac{\log p}{n}}\right\} \geq 1 - \frac{1}{p^2},
\end{equation*}
and,
\begin{equation*}\label{PC2n}
P\left\{ \max\limits_{i,j}\| {\mathbf{\Omega}_{:i}^{0}}' {\mathbf{S}}_{:j}\| < c_1 \sqrt{\frac{\log p}{n}}\right\} \geq 1 - \frac{1}{p^2}.
\end{equation*}

Define the events $C_{1,n}$, $C_{2,n}$ as 

\begin{equation}\label{C1n}
C_{1,n} = \left\{ \max\limits_{i,j}\| s_{ij} - \sigma_{ij} \| < c_1 \sqrt{\frac{\log p}{n}}\right\},
\end{equation}

\begin{equation}\label{C2n}
C_{2,n} = \left\{ \max\limits_{i,j}\| {\mathbf{\Omega}_{:i}^{0}}' {\mathbf{S}}_{:j}\| < c_1 \sqrt{\frac{\log p}{n}}\right\},
\end{equation}
for the next series of lemmas, we restrict ourself to the event $C_{1,n} \cap C_{2,n}$.

The next two lemmas prove important properties of the matrix $\mathbf{\Phi}$ that appears in the generalized posterior distribution.
\begin{Lemma}\label{Lemma 2}
The following holds
\begin{equation}\label{Upsilon eigen_vals}
\text{eig}_{min}\left( \mathbf{S}\right) \le \text{eig}_{min}\left( \mathbf{\Phi}\right) \le \text{eig}_{max}\left(\mathbf{\Phi}\right) \le 
2 \text{eig}_{max}\left( \mathbf{S}\right).
\end{equation}
\end{Lemma}

\begin{proof}
Let $\boldsymbol{y} = \boldsymbol{y}\left( \mathbf{\Omega}\right)$ be a vectorized version of $\mathbf{\Omega}$ obtained by shifting the corresponding diagonal entry at the bottom of each column of $\mathbf{\Omega}$ and then stacking the columns on top of each other. Let $\mathbf{P}^i$ be the $p\times p$ permutation matrix such that $\mathbf{P}^i \boldsymbol{z} = \left( z_1, ..., z_{i-1}, z_{i+1}, ..., z_p, z_i\right)$ for every $\boldsymbol{z}\in \mathbb{R}^p$. It follows by the definition of $\boldsymbol{y}$ that 
\begin{equation*}
\boldsymbol{y} = \boldsymbol{y}\left( \mathbf{\Omega}\right) = \left( \left(\mathbf{P}^1\mathbf{\Omega}_{:1} \right)', \left(\mathbf{P}^2\mathbf{\Omega}_{:2} \right)', ..., \left(\mathbf{P}^p\mathbf{\Omega}_{:p} \right)'\right)'.
\end{equation*}
Let $\boldsymbol{x} \in \mathbb{R}^{\frac{p(p+1)}{2}}$ be the symmetric version of $\boldsymbol{y}$ obtained by removing all $\omega_{ij}$ with $i > j$. More precisely,
\begin{equation*}
\boldsymbol{x}= \left( \omega_{11}, \omega_{12}, \omega_{22}, ..., \omega_{1p},..., \omega_{pp}\right)'.
\end{equation*}

Let $\tilde{\mathbf{P}}$ be the $p^2 \times \frac{p(p+1)}{2}$ matrix such that every entry of $\tilde{\mathbf{P}}$ is either zero or one, exactly one entry in each row of $\tilde{\mathbf{P}}$ is equal to 1, and $\boldsymbol{y} = \tilde{\mathbf{P}}\boldsymbol{x}$.  

Next, define $\boldsymbol{\xi}= \left( \omega_{12}, \omega_{13}, ..., \omega_{p-1p}\right)'$ and $\boldsymbol{\delta} = \left( \omega_{11}, \omega_{22}, ..., \omega_{pp}\right)'$ and let $\tilde{\mathbf{Q}}$ be the $\frac{p(p+1)}{2} \times \frac{p(p+1)}{2}$ permutation matrix for which 
\begin{equation*}
\boldsymbol{x} = \mathbf{Q}\left( {\begin{array}{*{20}{c}}
{\boldsymbol{\xi}} \\
{\boldsymbol{\delta}} 
\end{array} } \right) .
\end{equation*}

Let $\tilde{\mathbf{S}}$ be a $p^2 \times p^2$ block diagonal matrix with $p$ diagonal blocks, the $i^{\text{th}}$ block is equal to $\tilde{\mathbf{S}}^{i}:=\mathbf{P}^i\mathbf{S}{\mathbf{P}^i}'$. It follows that 

\begin{equation*}
\begin{split}
\text{tr}\left[ \mathbf{\Omega}^2 \mathbf{S}\right] &= \sum\limits_{i=1}^{p} {\mathbf{\Omega}_{:i}}' \mathbf{S} \mathbf{\Omega}_{:i} = \sum\limits_{i=1}^{p} {\mathbf{\Omega}_{:i}}' {\mathbf{P}^i }'\mathbf{P}^i \mathbf{S} {\mathbf{P}^i }'\mathbf{P}^i\mathbf{\Omega}_{:i} = \sum\limits_{i=1}^{p} {\mathbf{\Omega}_{:i}}' {\mathbf{P}^i }'\left( \mathbf{P}^i \mathbf{S} {\mathbf{P}^i }'\right)\mathbf{P}^i\mathbf{\Omega}_{:i}\\
& = \boldsymbol{y}'\tilde{\mathbf{\Sigma}}\boldsymbol{y} = \boldsymbol{x}'\tilde{\mathbf{P}}'\tilde{\mathbf{S}}\tilde{\mathbf{P}}\boldsymbol{x}= \left( {\boldsymbol{\xi}}', \boldsymbol{\delta}' \right) \mathbf{Q}' \tilde{\mathbf{P}}'\tilde{\mathbf{S}}\tilde{\mathbf{P}} \mathbf{Q} \left( {\begin{array}{*{20}{c}}
{\boldsymbol{\xi}} \\
{\boldsymbol{\delta}} 
\end{array} } \right). 
\end{split}
\end{equation*}
There also exist appropriate matrices $\mathbf{A}$ and $\mathbf{D}$ such that 
\begin{equation*}
\text{tr}\left[ \mathbf{\Omega}^2 \mathbf{S}\right] = \left( {\boldsymbol{\xi}}', \boldsymbol{\delta}' \right) \left( {\begin{array}{*{20}{c}}
{\mathbf{\Phi}} & {\mathbf{A}} \\
{\mathbf{A}} & {\mathbf{D}} \\
\end{array} } \right) \left( {\begin{array}{*{20}{c}}
{\boldsymbol{\xi}}\\
{\boldsymbol{\delta}} 
\end{array} } \right), 
\end{equation*}
therefore, we must have 
\begin{equation*}
 \mathbf{Q}' \tilde{\mathbf{P}}'\tilde{\mathbf{S}}\tilde{\mathbf{P}} \mathbf{Q} = \left( {\begin{array}{*{20}{c}}
{\mathbf{\Phi}} & {\mathbf{A}} \\
{\mathbf{A}} & {\mathbf{D}} \\
\end{array} } \right). 
\end{equation*}

\noindent
Note that since $\tilde{\mathbf{P}}$ has orthogonal columns with $\ell_2$-norm either $1$ or $2$, and the spectrum of 
$\tilde{\mathbf{S}}$ and $\mathbf{S}$ are identical, we have that
\begin{equation*}
\text{eig}_{min}\left( \mathbf{S}\right) = \text{eig}_{min}\left( \tilde{\mathbf{S}}\right) \le \text{eig}_{min}\left( \mathbf{\Phi}
\right) \le \text{eig}_{max}\left(\mathbf{\Phi}\right) \le 2 \text{eig}_{max}\left( \tilde{\mathbf{S}}\right) = 2 \text{eig}_{max}
\left( \mathbf{S}\right).
\end{equation*}
\end{proof}

\begin{Lemma}\label{Lemma 3}
Let ${\boldsymbol{l}} \in \mathcal{L}$ be any sparsity pattern/model with $d_{\boldsymbol{l}} < \tau_n = 
\frac{\tilde{\varepsilon}_0}{4c_1}\sqrt{\frac{n}{\log p}}$, then the sub matrix 
$\mathbf{\Phi}_{{\boldsymbol{l}}{\boldsymbol{l}}}$ of $\mathbf{\Phi}$, obtained by taking out all the rows and columns 
corresponding to the zero coordinates in $\boldsymbol{\xi} \in \mathcal{M}_{{\boldsymbol{l}}}$, is positive definite. 
Specifically,
\begin{equation}
\frac{3\tilde{\varepsilon}_0}{4} \leq \text{eig}_{\min}\left( \mathbf{\Phi}_{{\boldsymbol{l}} {\boldsymbol{l}}} \right) \leq \text{eig}_{\max}\left( \mathbf{\Phi}_{{\boldsymbol{l}} {\boldsymbol{l}}} \right) \leq \frac{5}{2\tilde{\varepsilon}_0}, \quad \forall \boldsymbol{l} \in \mathcal{L}.
\end{equation}

\end{Lemma}
\begin{proof}

Let $\mathbf{\Phi}_{{\boldsymbol{l}} {\boldsymbol{l}}}^{0}$ denote the population version of $\mathbf{\Phi}_{{\boldsymbol{l}} {\boldsymbol{l}}}$. Since, we are restricted to $C_{1,n}\cap C_{2,n}$,\\ $\| \mathbf{\Phi}_{{\boldsymbol{l}} {\boldsymbol{l}}} - \mathbf{\Phi}_{{\boldsymbol{l}} {\boldsymbol{l}}}^{0} \| \leq c_1 d_{\ell} \sqrt{\frac{\log p}{n}}$, hence 

\begin{equation*}
\begin{split}
\text{eig}_{min} \left( \mathbf{\Phi}_{{\boldsymbol{l}} {\boldsymbol{l}}}\right) = \inf\limits_{|\boldsymbol{x}|=1} \boldsymbol{x}'\mathbf{\Phi}_{{\boldsymbol{l}} {\boldsymbol{l}}}\boldsymbol{x}&\geq  \inf\limits_{|\boldsymbol{x}|=1} \boldsymbol{x}'\mathbf{\Phi}_{{\boldsymbol{l}} {\boldsymbol{l}}}^{0}\boldsymbol{x} - \inf\limits_{|\boldsymbol{x}|=1} \boldsymbol{x}'\left( \mathbf{\Phi}_{{\boldsymbol{l}} {\boldsymbol{l}}} - \mathbf{\Phi}_{{\boldsymbol{l}} {\boldsymbol{l}}}^{0}\right)\boldsymbol{x} \\
&\geq  \inf\limits_{|\boldsymbol{x}|=1} \boldsymbol{x}'\mathbf{\Phi}_{{\boldsymbol{l}} {\boldsymbol{l}}}^{0}\boldsymbol{x} - 
\|\mathbf{\Phi}_{{\boldsymbol{l}} {\boldsymbol{l}}} - \mathbf{\Phi}_{{\boldsymbol{l}} {\boldsymbol{l}}}^{0}\|_2 \\
&\geq  \inf\limits_{|\boldsymbol{x}|=1} \boldsymbol{x}'\mathbf{\Phi}^{0}_{{\boldsymbol{l}} {\boldsymbol{l}}}\boldsymbol{x} -  
c_1 d_{\boldsymbol{l}} \sqrt{\frac{\log p}{n}}
\end{split}
\end{equation*}
hence, by Lemma \ref{Lemma 2},
\begin{equation*}
\begin{split}
\text{eig}_{min} \left( \mathbf{\Phi}\right)_{{\boldsymbol{l}} {\boldsymbol{l}}} &\geq \tilde{\varepsilon}_0 - c_1 d_{\boldsymbol{l}} \sqrt{\frac{\log p}{n}}\\
& \geq \tilde{\varepsilon}_0 - c_1\tau_{n} \sqrt{\frac{\log p}{n}} = \frac{3\tilde{\varepsilon}_0}{4}.
\end{split}
\end{equation*}
Similarly one can show that
\begin{equation*}
\text{eig}_{\max}\left( \mathbf{\Phi}_{{\boldsymbol{l}} {\boldsymbol{l}}} \right) \leq\frac{5}{2\tilde{\varepsilon}_0}.
\end{equation*}
\end{proof}

\noindent
By Lemma \ref{Lemma 3}, the value of the threshold $\tau_n$ which we used in building our hierarchical prior in 
\ref{hierarchical prior 2} is given as $\tau_n = \frac{\tilde{\varepsilon}_0}{4c_1}\sqrt{\frac{n}{\log p}}$. Hence by 
Assumption \ref{Assumption2: d_t}, we can write $d_t \leq \tau_n$, for any sufficiently large $n$.

\begin{Lemma}\label{Lemma 4}
Let $\boldsymbol{\xi}^0, \boldsymbol{\delta}^0$ be the true values of $\boldsymbol{\xi}, \boldsymbol{\delta}$, $\mathbf{\Phi}$ 
and $\boldsymbol{a} = \mathbf{A} \boldsymbol{\delta}^0$ be according to (\ref{upsilon}) and (\ref{a}), and 
$\hat{\boldsymbol{a}} = \mathbf{A} \hat{\boldsymbol{\delta}}$ be the estimate of $\boldsymbol{a}$ obtained by replacing 
$\boldsymbol{\delta}^0$ by the accurate diagonal estimates $\hat{\boldsymbol{\delta}}$. Then for large enough $n$, there 
exists a constant $c_0$ such that 
\begin{equation}
\|\mathbf{\Phi} \boldsymbol{\xi}^0 + \hat{\boldsymbol{a}} \|_{\max} \leq c_0 \sqrt{\frac{\log p}{n}}. 
\end{equation}

\end{Lemma}
\begin{proof}
Note that by the triangular inequality,
\begin{equation}\label{Lemma 4 proof 1}
\|\mathbf{\Phi} \boldsymbol{\xi}^0 + \hat{\boldsymbol{a}} \|_{\max} \leq \|\mathbf{\Phi} \boldsymbol{\xi}^0 + \boldsymbol{a} \|_{\max} + \| \hat{\boldsymbol{a}} - \boldsymbol{a}\|_{\max},
\end{equation}
where,  provided by Assumption 1.

Next, in view of (\ref{upsilon}), (\ref{a}), and (3.1),
\begin{equation*}
\mathbf{\Phi} \boldsymbol{\xi}^{0} + \boldsymbol{a} = \left( {\begin{array}{*{20}{c}}
{{\mathbf{\Omega}^{0}_{:1}}' {\mathbf{S}}_{:2} + {\mathbf{\Omega}^{0}_{:2}}' {\mathbf{S}}_{:1}}\\
{{\mathbf{\Omega}^{0}_{:1}}' {\mathbf{S}}_{:3} + {\mathbf{\Omega}^{0}_{:3}}' {\mathbf{S}}_{:1}}\\
{\vdots}\\
{{\mathbf{\Omega}^{0}_{:p-1}}' {\mathbf{S}}_{:p} + {\mathbf{\Omega}^{0}_{:p}}' {\mathbf{S}}_{:p-1}}
\end{array} } \right),
\end{equation*}
hence, by restricting to the event $C_{1,n} \cap C_{2,n}$, we have that

\begin{equation}\label{Lemma 4 proof 2}
\begin{split}
\| \mathbf{\Phi} \boldsymbol{\xi}^0 + \boldsymbol{a} \|_{\max} 
& \le \sqrt{ \mathop{\max}_{1 \le i < j \le p} \left( {\mathbf{\Omega}^{0}_{:i}}' {\mathbf{S}}_{:j}\right)^2}\\ 
& \le \mathop{\max}_{1 \le i < j \le p} | {\mathbf{\Omega}^{0}_{:i}}' \left( {\mathbf{S}}_{:j} - \mathbf{\Sigma}_{:j}\right) | \\
& \le 2c_1 \sqrt{\frac{\log p}{n}}.
\end{split}
\end{equation}
Moreover, by (\ref{a}) and (\ref{accurate-estimation}), it is easy to see that
\begin{equation*}
\begin{split}
\| \hat{\boldsymbol{a}} - \boldsymbol{a}\|_{\max} \leq & 2 C \| \mathbf{S}\|_{\max} \sqrt{\frac{\log p}{n}}. 
\end{split}
\end{equation*}

\noindent
Since we are restricting to the event $C_{1,n}$, it follows by Assumption \ref{Assumption3: Bounded eigenvalues} that
\begin{equation}\label{Lemma 4 proof 3}
\begin{split}
\| \hat{\boldsymbol{a}} - \boldsymbol{a}\|_{\max} & \leq  \frac{5C}{\tilde{\varepsilon}_0}\sqrt{\frac{\log p}{n}}
\end{split}
\end{equation}

\noindent
By combining (\ref{Lemma 4 proof 1}), (\ref{Lemma 4 proof 2}), and (\ref{Lemma 4 proof 3}), we get that 
\begin{equation*}
\|\mathbf{\Phi} \boldsymbol{\xi}^0 + \hat{\boldsymbol{a}} \|_{\max} \leq \left( 2c_1 + \frac{5C}{\tilde{\varepsilon}_0} \right) 
\sqrt{\frac{\log p}{n}}. 
\end{equation*}

\noindent
The result follows by letting $c_0 = 2c_1 + \frac{5C}{\tilde{\varepsilon}_0}$.
\end{proof}

\noindent
For ease of presentation, we denote the ratio of the posterior probabilities of any sparsity pattern/model 
$\boldsymbol{{\boldsymbol{l}}}$ and the true sparsity pattern/model $\boldsymbol{t}$,  by $PR\left( \boldsymbol{l}, 
\boldsymbol{t}\right)$, i.e.
\begin{equation}
PR\left( \boldsymbol{l}, \boldsymbol{t}\right) = \frac{P \left\{ \boldsymbol{{\boldsymbol{l}}} | \hat{\boldsymbol{\delta}}, 
\mathcal{Y}\right\}}{P \left\{ \boldsymbol{t} | \hat{\boldsymbol{\delta}}, \mathcal{Y}\right\}}, \quad \quad \text{for any sparsity 
pattern}\quad  \boldsymbol{l} \neq \boldsymbol{t}.
\end{equation}

\begin{Lemma} \label{Lemma 1 PR}
The ratio of the posterior probabilities of any sparsity pattern/model $\boldsymbol{{\boldsymbol{l}}}$ and 
the true sparsity pattern/model $\boldsymbol{t}$ satisfies:

\begin{equation}\label{PR}
\begin{split}
PR(\boldsymbol{l}, \boldsymbol{t}) &= \frac{\pi \left\{ \boldsymbol{{\boldsymbol{l}}} | \hat{\boldsymbol{\delta}}, \mathcal{Y} \right\}}{\pi \left\{ \boldsymbol{t} | \hat{\boldsymbol{\delta}}, \mathcal{Y} \right\}} = \frac{ q^{d_{\boldsymbol{{\boldsymbol{l}}}}}(1-q)^{\binom{p}{2} - d_{\boldsymbol{{\boldsymbol{l}}}}}}{q^{d_{\boldsymbol{t}}}(1-q)^{\binom{p}{2} - d_{\boldsymbol{t}}}}\\
&\times \frac{|\mathbf{\Lambda}_{\boldsymbol{l}\boldsymbol{l}}|^{\frac{1}{2}}}{|\mathbf{\Lambda}_{\boldsymbol{l}\boldsymbol{t}}|^{\frac{1}{2}}} \frac{| \left( n\mathbf{\Phi} + \mathbf{\Lambda}\right)_{tt}|^{\frac{1}{2}}}{| \left( n\mathbf{\Phi} + \mathbf{\Lambda}\right)_{{\boldsymbol{l}}{\boldsymbol{l}}}|^{\frac{1}{2}}} \frac{\exp\left\{ \frac{n^2}{2}\hat{\boldsymbol{a}}_{\boldsymbol{l}}' \left( n\mathbf{\Phi} + \mathbf{\Lambda}\right)_{{\boldsymbol{l}}{\boldsymbol{l}}}^{-1}\hat{\boldsymbol{a}}_{\boldsymbol{l}}\right\}}{\exp\left\{ \frac{n^2}{2}\hat{\boldsymbol{a}}_t' \left( n\mathbf{\Phi} + \mathbf{\Lambda}\right)_{tt}^{-1}\hat{\boldsymbol{a}}_t\right\}}.
\end{split}
\end{equation}
\end{Lemma}
\begin{proof} 
We note that
\begin{equation*}
\begin{split}
\pi& \left\{ \boldsymbol{l} | \hat{\boldsymbol{\delta}}, \mathcal{Y}\right\} = \pi \left\{ \boldsymbol{\xi}\in \mathcal{M}_{\boldsymbol{l}}| \hat{\boldsymbol{\delta}}, \mathcal{Y} \right\} = \int_{\mathcal{M}_{\boldsymbol{l}}} \pi \left(\boldsymbol{\xi}| \hat{\boldsymbol{\delta}}, \mathcal{Y}\right) d\boldsymbol{\xi}, 
\end{split}
\end{equation*}

\noindent
Hence, in view of (\ref{posterior theta given delta}), 
\begin{equation*}
\begin{split}
\pi& \left\{ \boldsymbol{l} | \hat{\boldsymbol{\delta}}, \mathcal{Y} \right\} =C_0 q^{d_{\boldsymbol{l}}}(1-q)^{\binom{p}{2} - d_{\boldsymbol{l}}} \frac{|\mathbf{\Lambda}_{\boldsymbol{l}\boldsymbol{l}}|^{\frac{1}{2}}}{| \left( n\mathbf{\Phi} + \mathbf{\Lambda}\right)_{{\boldsymbol{l}}{\boldsymbol{l}}}|^{\frac{1}{2}}} \exp\left\{ \frac{n^2}{2}\hat{\boldsymbol{a}}_{\boldsymbol{l}}' \left( n\mathbf{\Phi} + \mathbf{\Lambda}\right)_{{\boldsymbol{l}}{\boldsymbol{l}}}^{-1}\hat{\boldsymbol{a}}_{\boldsymbol{l}}\right\},\\
\end{split}
\end{equation*}
where the last equality is achieved using the properties of the multivariate normal distribution. 
\end{proof}
In the next series of lemmas, we will show that for any sparsity pattern $\boldsymbol{l} \in \mathcal{L}$, the posterior 
probability ratio $PR(\boldsymbol{l}, \boldsymbol{t})$ is approaching zero, as $n$ goes to $\infty$. Specifically, we consider 
four cases of underfitted (${\boldsymbol{l}} \subset \boldsymbol{t}$), overfitted (${\boldsymbol{t}} \subset \boldsymbol{l}$ 
with $d_{\boldsymbol{l}} < \tau_n$), and non-inclusive ( ${\boldsymbol{t}} \not \subseteq \boldsymbol{l}$ and 
${\boldsymbol{l}} \not \subseteq \boldsymbol{t}$ ) models.

\begin{Lemma} \label{Lemma 5}
Suppose ${\boldsymbol{l}} \subset \boldsymbol{t}$ then, under 
Assumptions \ref{Assumption2: d_t} - \ref{Assumption6: Decay rate} 
\begin{equation}
PR({\boldsymbol{l}}, \boldsymbol{t}) \to 0, \quad \quad \quad \text{as} \quad n \to \infty . 
\end{equation}
\end{Lemma}

\begin{proof}
By Assumption 2, $d_{\boldsymbol{t}} < \tau$, hence $d_{\boldsymbol{l}} < d_{\boldsymbol{t}} < \tau$. Now,

\begin{equation*}
\begin{split}
PR\left({\boldsymbol{l}}, \boldsymbol{t}\right) =& \frac{\| \mathbf{\Lambda}_{{\boldsymbol{l}} {\boldsymbol{l}}} \|^{\frac{1}{2}}}{\| \mathbf{\Lambda}_{\boldsymbol{t}\boldsymbol{t}} \|^{\frac{1}{2}}} \left( \frac{q}{1-q}\right)^{d_{\boldsymbol{l}} - d_{\boldsymbol{t}}} \frac{\|\left( n\mathbf{\Phi} + \mathbf{\Lambda}\right)_{\boldsymbol{t}\boldsymbol{t}} \|^{\frac{1}{2}}}{\|\left( n\mathbf{\Phi} + \mathbf{\Lambda}\right)_{{\boldsymbol{l}} {\boldsymbol{l}}} \|^{\frac{1}{2}}} \frac{\exp\left\{ \frac{n^2}{2} {\hat{\boldsymbol{a}}_{{\boldsymbol{l}}}}' \left( n\mathbf{\Phi} + \mathbf{\Lambda}\right)^{-1}_{{\boldsymbol{l}} {\boldsymbol{l}}} \hat{\boldsymbol{a}}_{\boldsymbol{l}} \right\}}{\exp\left\{ \frac{n^2}{2} {\hat{\boldsymbol{a}}_{\boldsymbol{t}}}' \left( n\mathbf{\Phi} + \mathbf{\Lambda}\right)^{-1}_{\boldsymbol{t}\boldsymbol{t}} \hat{\boldsymbol{a}}_{\boldsymbol{t}}\right\}}\\
= & \frac{\| \mathbf{\Lambda}_{{\boldsymbol{l}} {\boldsymbol{l}}} \|^{\frac{1}{2}}}{\| \mathbf{\Lambda}_{\boldsymbol{t}\boldsymbol{t}} \|^{\frac{1}{2}}} \left( \frac{q}{1-q}\right)^{d_{\boldsymbol{l}} - d_{\boldsymbol{t}}} \frac{\|\left( n\mathbf{\Phi} + \mathbf{\Lambda}\right)_{\boldsymbol{t}\boldsymbol{t}} \|^{\frac{1}{2}}}{\|\left( n\mathbf{\Phi} + \mathbf{\Lambda}\right)_{{\boldsymbol{l}} {\boldsymbol{l}}} \|^{\frac{1}{2}}}\\
&\exp\left\{- \frac{n^2}{2}\left[\hat{\boldsymbol{a}}_{{\boldsymbol{l}}^c} - n \mathbf{\Phi}_{{\boldsymbol{l}}^c{\boldsymbol{l}}}\left( n\mathbf{\Phi} + \mathbf{\Lambda}\right)_{{\boldsymbol{l}}{\boldsymbol{l}}}^{-1}\hat{\boldsymbol{a}}_{\boldsymbol{l}}\right]'\left( n\mathbf{\Phi} + \mathbf{\Lambda}\right)_{\boldsymbol{t}|{\boldsymbol{l}}}^{-1} \left[\hat{\boldsymbol{a}}_{{\boldsymbol{l}}^c} - n \mathbf{\Phi}_{{\boldsymbol{l}}^c{\boldsymbol{l}}}\left( n\mathbf{\Phi} + \mathbf{\Lambda}\right)_{{\boldsymbol{l}}{\boldsymbol{l}}}^{-1}\hat{\boldsymbol{a}}_{\boldsymbol{l}}\right]\right\},
\end{split}
\end{equation*}

\noindent
where $\boldsymbol{l}^c = \boldsymbol{t} \setminus \boldsymbol{l}$. It follows that 
\begin{equation*}
\begin{split}
PR\left({\boldsymbol{l}}, \boldsymbol{t}\right) \leq& \frac{\| \mathbf{\Lambda}_{{\boldsymbol{l}} {\boldsymbol{l}}} \|^{\frac{1}{2}}}{\| \mathbf{\Lambda}_{\boldsymbol{t}\boldsymbol{t}} \|^{\frac{1}{2}}} \left( \frac{q}{1-q}\right)^{d_{\boldsymbol{l}} - d_{\boldsymbol{t}}} \frac{\|\left( n\mathbf{\Phi} + \mathbf{\Lambda}\right)_{\boldsymbol{t}\boldsymbol{t}} \|^{\frac{1}{2}}}{\|\left( n\mathbf{\Phi} + \mathbf{\Lambda}\right)_{{\boldsymbol{l}} {\boldsymbol{l}}} \|^{\frac{1}{2}}} \exp \left\{ - \frac{n^2 \| \hat{\boldsymbol{a}}_{{\boldsymbol{l}}^c} - n \mathbf{\Phi}_{{\boldsymbol{l}}^c{\boldsymbol{l}}}\left( n\mathbf{\Phi} + \mathbf{\Lambda}\right)_{{\boldsymbol{l}}{\boldsymbol{l}}}^{-1}\hat{\boldsymbol{a}}_{\boldsymbol{l}}\|^2}{2 \text{eig}_{max}\left(n\mathbf{\Phi} + \mathbf{\Lambda} \right)_{\boldsymbol{t}\boldsymbol{t}}}\right\},
\end{split}
\end{equation*}

Now, by the triangular inequality,
\begin{equation}\label{Lemma 5.1}
\begin{split}
\| \hat{\boldsymbol{a}}_{{\boldsymbol{l}}^c} - n \mathbf{\Phi}_{{\boldsymbol{l}}^c{\boldsymbol{l}}}&\left( n\mathbf{\Phi} + \mathbf{\Lambda}\right)_{{\boldsymbol{l}}{\boldsymbol{l}}}^{-1}\hat{\boldsymbol{a}}_{\boldsymbol{l}}\| \\
\geq &\| \boldsymbol{a}_{{\boldsymbol{l}}^c} - n \mathbf{\Phi}_{{\boldsymbol{l}}^c{\boldsymbol{l}}}\left( n\mathbf{\Phi} + \mathbf{\Lambda}\right)_{{\boldsymbol{l}}{\boldsymbol{l}}}^{-1}\boldsymbol{a}_{\boldsymbol{l}}\|\\
&-\|\left( \hat{\boldsymbol{a}}_{{\boldsymbol{l}}^c} - \boldsymbol{a}_{{\boldsymbol{l}}^c}\right) - n \mathbf{\Phi}_{{\boldsymbol{l}}^c{\boldsymbol{l}}}\left( n\mathbf{\Phi} + \mathbf{\Lambda}\right)_{{\boldsymbol{l}}{\boldsymbol{l}}}^{-1}\left( \hat{\boldsymbol{a}}_{{\boldsymbol{l}}} - \boldsymbol{a}_{{\boldsymbol{l}}} \right) \|\\
=&\| \left( \pm \mathbf{\Phi} \boldsymbol{\xi}^0 + \boldsymbol{a}\right)_{{\boldsymbol{l}}^c} - n \mathbf{\Phi}_{{\boldsymbol{l}}^c{\boldsymbol{l}}}\left( n\mathbf{\Phi} + \mathbf{\Lambda}\right)_{{\boldsymbol{l}}{\boldsymbol{l}}}^{-1}\left( \pm \mathbf{\Phi} \boldsymbol{\xi}^0 + \boldsymbol{a} \right)_{{\boldsymbol{l}}}\| \\
&-\|\left( \hat{\boldsymbol{a}}_{{\boldsymbol{l}}^c} - \boldsymbol{a}_{{\boldsymbol{l}}^c}\right) - n \mathbf{\Phi}_{{\boldsymbol{l}}^c{\boldsymbol{l}}}\left( n\mathbf{\Phi} + \mathbf{\Lambda}\right)_{{\boldsymbol{l}}{\boldsymbol{l}}}^{-1}\left( \hat{\boldsymbol{a}}_{{\boldsymbol{l}}} - \boldsymbol{a}_{{\boldsymbol{l}}} \right) \|\\
\ge & \| \left( \mathbf{\Phi} \boldsymbol{\xi}^0\right)_{{\boldsymbol{l}}^c} - n \mathbf{\Phi}_{{\boldsymbol{l}}^c{\boldsymbol{l}}}\left( n\mathbf{\Phi} + \mathbf{\Lambda}\right)_{{\boldsymbol{l}}{\boldsymbol{l}}}^{-1}\left( \mathbf{\Phi} \boldsymbol{\xi}^0 \right)_{{\boldsymbol{l}}}\| \\
&- \| \left( \mathbf{\Phi} \boldsymbol{\xi}^0 + \boldsymbol{a}\right)_{{\boldsymbol{l}}^c} - n \mathbf{\Phi}_{{\boldsymbol{l}}^c{\boldsymbol{l}}}\left( n\mathbf{\Phi} + \mathbf{\Lambda}\right)_{{\boldsymbol{l}}{\boldsymbol{l}}}^{-1}\left( \mathbf{\Phi} \boldsymbol{\xi}^0 + \boldsymbol{a}\right)_{{\boldsymbol{l}}}\| \\
&-\|\left( \hat{\boldsymbol{a}}_{{\boldsymbol{l}}^c} - \boldsymbol{a}_{{\boldsymbol{l}}^c}\right) - n \mathbf{\Phi}_{{\boldsymbol{l}}^c{\boldsymbol{l}}}\left( n\mathbf{\Phi} + \mathbf{\Lambda}\right)_{{\boldsymbol{l}}{\boldsymbol{l}}}^{-1}\left( \hat{\boldsymbol{a}}_{{\boldsymbol{l}}} - \boldsymbol{a}_{{\boldsymbol{l}}} \right) \|.
\end{split}
\end{equation}
Now, by appropriately partitioning $\mathbf{\Phi}$, we can write $\left( \mathbf{\Phi} \boldsymbol{\xi}^0\right)_{{\boldsymbol{l}}^c} = \mathbf{\Phi}_{{\boldsymbol{l}}^c{\boldsymbol{l}}} \boldsymbol{\xi}^0_{{\boldsymbol{l}}} + \mathbf{\Phi}_{{\boldsymbol{l}}^c{\boldsymbol{l}}^c} \boldsymbol{\xi}^0_{{\boldsymbol{l}}^c}$ and $\left( \mathbf{\Phi} \boldsymbol{\xi}^0\right)_{{\boldsymbol{l}}} = \mathbf{\Phi}_{{\boldsymbol{l}}{\boldsymbol{l}}} \boldsymbol{\xi}^0_{{\boldsymbol{l}}} + \mathbf{\Phi}_{{\boldsymbol{l}}{\boldsymbol{l}}^c} \boldsymbol{\xi}^0_{{\boldsymbol{l}}^c}$. Hence, for large enough $n$,

\begin{equation}\label{Lemma 5.2}
\begin{split}
\| \left(\mathbf{\Phi} \boldsymbol{\xi}^0\right)_{{\boldsymbol{l}}^c} - n \mathbf{\Phi}_{{\boldsymbol{l}}^c{\boldsymbol{l}}}\left( n\mathbf{\Phi} + \mathbf{\Lambda}\right)_{{\boldsymbol{l}}{\boldsymbol{l}}}^{-1}\left( \mathbf{\Phi} \boldsymbol{\xi}^0 \right)_{{\boldsymbol{l}}}\| &= \| \frac{1}{n}\left( n\mathbf{\Phi} + \mathbf{\Lambda}\right)_{\boldsymbol{t}|{\boldsymbol{l}}}\boldsymbol{\xi}_{{\boldsymbol{l}}^c}^0 - \mathbf{\Phi}_{{\boldsymbol{l}}^c{\boldsymbol{l}}}\left( n\mathbf{\Phi} + \mathbf{\Lambda}\right)_{{\boldsymbol{l}}{\boldsymbol{l}}}^{-1}\mathbf{\Lambda}_{{\boldsymbol{l}}{\boldsymbol{l}}}\boldsymbol{\xi}_{{\boldsymbol{l}}}^0\| \\
&\geq \| \frac{1}{n}\left( n\mathbf{\Phi} + \mathbf{\Lambda}\right)_{\boldsymbol{t}|{\boldsymbol{l}}}\boldsymbol{\xi}_{{\boldsymbol{l}}^c}^0 \| - \| \mathbf{\Phi}_{{\boldsymbol{l}}^c{\boldsymbol{l}}}\left( n\mathbf{\Phi} + \mathbf{\Lambda}\right)_{{\boldsymbol{l}}{\boldsymbol{l}}}^{-1}\mathbf{\Lambda}_{{\boldsymbol{l}}{\boldsymbol{l}}}\boldsymbol{\xi}_{{\boldsymbol{l}}}^0\| \\
&\geq \| \frac{1}{n}\left( n\mathbf{\Phi} + \mathbf{\Lambda}\right)_{\boldsymbol{t}|{\boldsymbol{l}}}\boldsymbol{\xi}_{{\boldsymbol{l}}^c}^0 \| - \frac{\text{eig}_{\min}\left( \mathbf{\Phi}_{{\boldsymbol{l}}^c{\boldsymbol{l}}}\right) \| \mathbf{\Lambda}_{{\boldsymbol{l}}{\boldsymbol{l}}}\boldsymbol{\xi}_{{\boldsymbol{l}}}^0\|}{\text{eig}_{\min}\left( n\mathbf{\Phi} + \mathbf{\Lambda}\right)_{{\boldsymbol{l}} {\boldsymbol{l}}}} \\
&\geq \| \frac{1}{n}\left( n\mathbf{\Phi} + \mathbf{\Lambda}\right)_{\boldsymbol{t}|{\boldsymbol{l}}}\boldsymbol{\xi}_{{\boldsymbol{l}}^c}^0 \| - \frac{2 \| \mathbf{\Lambda}_{{\boldsymbol{l}}{\boldsymbol{l}}}\boldsymbol{\xi}_{{\boldsymbol{l}}}^0\|}{n\tilde{\varepsilon}_0^2} \\
& \geq \frac{1}{2} \| \frac{1}{n}\left( n\mathbf{\Phi} + \mathbf{\Lambda}\right)_{\boldsymbol{t}|{\boldsymbol{l}}}\boldsymbol{\xi}_{{\boldsymbol{l}}^c}^0 \| \\
&\ge \frac{1}{2} \frac{1}{n} \text{eig}_{\text{min}}\left( n\mathbf{\Phi} + \mathbf{\Lambda}\right)_{\boldsymbol{t}\boldsymbol{t}}s_n\sqrt{(d_{\boldsymbol{t}} - d_{\boldsymbol{l}})} \\
&\ge \frac{1}{2} \frac{1}{n} n\text{eig}_{\text{min}}\left(\mathbf{\Phi} \right)_{tt}s_n\sqrt{(d_{\boldsymbol{t}} - d_{\boldsymbol{l}})} \\
&\ge \frac{3}{8} \tilde{\varepsilon}_0 s_n \sqrt{(d_{\boldsymbol{t}} - d_{\boldsymbol{l}})}
\end{split}
\end{equation} 
Moving onto the second term in the right hand side of (\ref{Lemma 5.1}),
\begin{equation}\label{Lemma 5.3}
\begin{split}
& \| \left( \mathbf{\Phi} \boldsymbol{\xi}^0 + \boldsymbol{a}\right)_{{\boldsymbol{l}}^c} - n \mathbf{\Phi}_{{\boldsymbol{l}}^c{\boldsymbol{l}}}\left( n\mathbf{\Phi} + \mathbf{\Lambda}\right)_{{\boldsymbol{l}}{\boldsymbol{l}}}^{-1}\left( \mathbf{\Phi} \boldsymbol{\xi}^0 + \boldsymbol{a}\right)_{{\boldsymbol{l}}}\| \\
& \le \| \left( \mathbf{\Phi} \boldsymbol{\xi}^0 + \boldsymbol{a}\right)_{{\boldsymbol{l}}^c} \| + \| n \mathbf{\Phi}_{{\boldsymbol{l}}^c{\boldsymbol{l}}}\left( n\mathbf{\Phi} + \mathbf{\Lambda}\right)_{{\boldsymbol{l}}{\boldsymbol{l}}}^{-1}\left( \mathbf{\Phi} \boldsymbol{\xi}^0 + \boldsymbol{a}\right)_{{\boldsymbol{l}}}\|\\
& \le \| \left( \mathbf{\Phi} \boldsymbol{\xi}^0 + \boldsymbol{a}\right)_{{\boldsymbol{l}}^c} \| + \frac{n\text{eig}_{\max} \left(\mathbf{\Phi}_{\boldsymbol{l}^c \boldsymbol{l}}\right) \|\left( \mathbf{\Phi} \boldsymbol{\xi}^0 + \boldsymbol{a}\right)_{{\boldsymbol{l}}}\| }{\text{eig}_{\min} \left( n\mathbf{\Phi} + \mathbf{\Lambda}\right)_{{\boldsymbol{l}} {\boldsymbol{l}}} }\\
& \le \| \left( \mathbf{\Phi} \boldsymbol{\xi}^0 + \boldsymbol{a}\right)_{{\boldsymbol{l}}^c} \| + \frac{2\|\left( \mathbf{\Phi} \boldsymbol{\xi}^0 + \boldsymbol{a}\right)_{{\boldsymbol{l}}}\|}{\tilde{\varepsilon}_0^2} \\
& \leq c_0\sqrt{\frac{\log p}{n}} \left(\sqrt{d_{\boldsymbol{t}} - d_{\boldsymbol{l}}} + \frac{2\sqrt{d_{\boldsymbol{l}}}}{\tilde{\varepsilon}_0^2} \right),
\end{split}
\end{equation}
where the last equality was achieved by Lemma \ref{Lemma 4}. 
Further, regarding the third term in the right hand side of (\ref{Lemma 5.1}) we can express it as 
\begin{equation}\label{Lemma 5.4}
\begin{split}
\|\left( \hat{\boldsymbol{a}}_{{\boldsymbol{l}}^c} - \boldsymbol{a}_{{\boldsymbol{l}}^c}\right) - n \mathbf{\Phi}_{{\boldsymbol{l}}^c{\boldsymbol{l}}}\left( n\mathbf{\Phi} + \mathbf{\Lambda}\right)_{{\boldsymbol{l}}{\boldsymbol{l}}}^{-1}\left( \hat{\boldsymbol{a}}_{{\boldsymbol{l}}} - \boldsymbol{a}_{{\boldsymbol{l}}} \right) \| & \leq \| \hat{\boldsymbol{a}}_{{\boldsymbol{l}}^c} - \boldsymbol{a}_{{\boldsymbol{l}}^c} \| + \frac{n\text{eig}_{\max} \left(\mathbf{\Phi}_{\boldsymbol{l}^c \boldsymbol{l}}\right) \|\left( \hat{\boldsymbol{a}}_{{\boldsymbol{l}}} - \boldsymbol{a}_{{\boldsymbol{l}}} \right) \| }{\text{eig}_{\min} \left( n\mathbf{\Phi} + \mathbf{\Lambda}\right)_{{\boldsymbol{l}} {\boldsymbol{l}}} }\\
& \leq \frac{3C}{\tilde{\varepsilon}_0} \sqrt{\frac{\log p}{n}}  \left(\sqrt{d_{\boldsymbol{t}} - d_{\boldsymbol{l}}} + \frac{2\sqrt{d_{\boldsymbol{l}}}}{\tilde{\varepsilon}_0^2} \right),
\end{split}
\end{equation}

\noindent
Hence, by combining (\ref{Lemma 5.1}), (\ref{Lemma 5.2}), (\ref{Lemma 5.3}), and (\ref{Lemma 5.4}), for sufficiently large 
$n$, we have that
\begin{equation*}
\begin{split}
\| \hat{\boldsymbol{a}}_{{\boldsymbol{l}}^c} - n \mathbf{\Phi}_{{\boldsymbol{l}}^c{\boldsymbol{l}}}\left( n\mathbf{\Phi} + \mathbf{\Lambda}\right)_{{\boldsymbol{l}}{\boldsymbol{l}}}^{-1}\hat{\boldsymbol{a}}_{\boldsymbol{l}}\| \geq & \frac{3}{8} \tilde{\varepsilon}_0 s_n \sqrt{(d_{\boldsymbol{t}} - d_{\boldsymbol{l}})} \\
&- c_0\sqrt{\frac{\log p}{n}}  \left(\sqrt{d_{\boldsymbol{t}} - d_{\boldsymbol{l}}} + \frac{2\sqrt{d_{\boldsymbol{l}}}}{\tilde{\varepsilon}_0^2} \right) \\
&- \frac{3C}{\tilde{\varepsilon}_0} \sqrt{\frac{\log p}{n}}  \left(\sqrt{d_{\boldsymbol{t}} - d_{\boldsymbol{l}}} + \frac{2\sqrt{d_{\boldsymbol{l}}}}{\tilde{\varepsilon}_0^2} \right)\\
\geq & \frac{1}{2} \tilde{\varepsilon}_0 s_n \sqrt{(d_{\boldsymbol{t}} - d_{\boldsymbol{l}})} \\
& - \left(c_0 +  \frac{3C}{\tilde{\varepsilon}_0} \right) \sqrt{\frac{\log p}{n}}  \left(\sqrt{d_{\boldsymbol{t}} - d_{\boldsymbol{l}}} + \frac{2\sqrt{d_{\boldsymbol{l}}}}{\tilde{\varepsilon}_0^2} \right)\\
\geq & \frac{1}{2} \tilde{\varepsilon}_0 s_n - \left(c_0 +  \frac{3C}{\tilde{\varepsilon}_0} \right) \sqrt{\frac{\log p}{n}}  \left(\frac{2\sqrt{d_{\boldsymbol{t}}}}{\tilde{\varepsilon}_0^2} \right), \\
\end{split}
\end{equation*}
in view of Assumption 5, $\frac{ \frac{3}{8} \tilde{\varepsilon}_0 s_n }{ \left(c_0 +  \frac{3C}{\tilde{\varepsilon}_0} \right) \sqrt{\frac{\log p}{n}}  \left(\frac{2\sqrt{d_{\boldsymbol{t}}}}{\tilde{\varepsilon}_0^2} \right)} \to \infty$, as $n \to \infty$, hence, for all large $n$, we can write,
\begin{equation*}
\| \hat{\boldsymbol{a}}_{{\boldsymbol{l}}^c} - n \mathbf{\Phi}_{{\boldsymbol{l}}^c{\boldsymbol{l}}}\left( n\mathbf{\Phi} + \mathbf{\Lambda}\right)_{{\boldsymbol{l}}{\boldsymbol{l}}}^{-1}\hat{\boldsymbol{a}}_{\boldsymbol{l}}\| \geq \frac{1}{4} \tilde{\varepsilon}_0 s_n 
\end{equation*}
Now, once again by Lemma \ref{Lemma 4}
\begin{equation*}
\begin{split}
PR({\boldsymbol{l}}, \boldsymbol{t}) & \leq \frac{\| \mathbf{\Lambda}_{{\boldsymbol{l}} {\boldsymbol{l}}}\|^{\frac{1}{2}}}{\| \mathbf{\Lambda}_{\boldsymbol{t}\boldsymbol{t}}\|^{\frac{1}{2}}} (2q)^{d_{\boldsymbol{l}} - d_{\boldsymbol{t}}} n^{\frac{d_{\boldsymbol{t}} - d_{\boldsymbol{l}} }{2}} \exp \left\{-\frac{n^2 \frac{1}{64} \tilde{\varepsilon}_0^2s_n^2}{6n\tilde{\varepsilon}_0^{-1}} \right\} \\
& = \frac{\| \mathbf{\Lambda}_{{\boldsymbol{l}} {\boldsymbol{l}}}\|^{\frac{1}{2}}}{\| \mathbf{\Lambda}_{\boldsymbol{t}\boldsymbol{t}}\|^{\frac{1}{2}}} 2^{d_{{\boldsymbol{l}}} - d_{\boldsymbol{t}}} \left( \frac{\sqrt{n}}{q} \exp \left\{ -\frac{n\tilde{\varepsilon}_0^3 s_n^2}{384}\right\} \right)^{d_{\boldsymbol{t}} - d_{\boldsymbol{l}} }. 
\end{split}
\end{equation*}

\noindent
Since the diagonal entries of $\mathbf{\Lambda}$ are uniformly bounded, it follows by 
Assumption \ref{Assumption5: Signal Strength} that for large enough $n$ 
\begin{equation*}
\begin{split}
PR({\boldsymbol{l}}, \boldsymbol{t}) & \leq \frac{\| \mathbf{\Lambda}_{{\boldsymbol{l}} {\boldsymbol{l}}}\|^{\frac{1}{2}}}{\| \mathbf{\Lambda}_{\boldsymbol{t}\boldsymbol{t}}\|^{\frac{1}{2}}} 2^{d_{{\boldsymbol{l}}} - d_{\boldsymbol{t}}} \left( \frac{\sqrt{n}}{q} \exp \left\{ -\log n -2a_2d_{\boldsymbol{t}} \log p \right\} \right)^{d_{\boldsymbol{t}} - d_{\boldsymbol{l}} }\\
&= \frac{\| \mathbf{\Lambda}_{{\boldsymbol{l}} {\boldsymbol{l}}}\|^{\frac{1}{2}}}{\| \mathbf{\Lambda}_{\boldsymbol{t}\boldsymbol{t}}\|^{\frac{1}{2}}} 2^{d_{{\boldsymbol{l}}} - d_{\boldsymbol{t}}} \left( \frac{p^{-2a_2 d_{\boldsymbol{t}}}}{\sqrt{n}q}  \right)^{d_{\boldsymbol{t}} - d_{\boldsymbol{l}} }\\
&= \frac{\| \mathbf{\Lambda}_{{\boldsymbol{l}} {\boldsymbol{l}}}\|^{\frac{1}{2}}}{\| \mathbf{\Lambda}_{\boldsymbol{t}\boldsymbol{t}}\|^{\frac{1}{2}}} 2^{d_{{\boldsymbol{l}}} - d_{\boldsymbol{t}}} \left( \frac{p^{-a_2 d_{\boldsymbol{t}}}}{\sqrt{n}}  \right)^{d_{\boldsymbol{t}} - d_{\boldsymbol{l}} }\\
&= \left( \frac{ 2 C_1 p^{-a_2 d_{\boldsymbol{t}}}}{\sqrt{n}} \right)^{d_{\boldsymbol{t}} - 
d_{\boldsymbol{l}}} 
\end{split}
\end{equation*}

\noindent
where $C_1$ is an appropriate constant. 
\end{proof}

\begin{Lemma} \label{Lemma 6}
Suppose ${\boldsymbol{l}} \supset \boldsymbol{t}$, and $d_{\boldsymbol{l}} < \tau_n$ then, under 
Assumptions \ref{Assumption2: d_t} - \ref{Assumption6: Decay rate} 
\begin{equation*}
PR({\boldsymbol{l}}, \boldsymbol{t}) \to 0, \quad \quad \quad \text{as} \quad n \to \infty . 
\end{equation*}
\end{Lemma}

\begin{proof}
In this case,
\begin{equation*}
PR(\boldsymbol{l}, \boldsymbol{t}) = \frac{\|\mathbf{\Lambda}_{{\boldsymbol{l}} {\boldsymbol{l}}} \|^{\frac{1}{2}}}{\| \mathbf{\Lambda}_{\boldsymbol{t}\boldsymbol{t}}\|^{\frac{1}{2}}} \frac{\left(2q\right)^{d_{\boldsymbol{l}} - d_{\boldsymbol{t}}}}{\| \left( n\mathbf{\Phi} + \mathbf{\Lambda}\right)_{{\boldsymbol{l}} | \boldsymbol{t}}\|^{\frac{1}{2}}} \exp \left\{ \frac{n^2}{2}\hat{\boldsymbol{a}}_{\boldsymbol{l}}' \left( n\mathbf{\Phi} + \mathbf{\Lambda}\right)_{{\boldsymbol{l}}{\boldsymbol{l}}}^{-1}\hat{\boldsymbol{a}}_{\boldsymbol{l}} - \frac{n^2}{2} \hat{\boldsymbol{a}}_{\boldsymbol{t}}' \left( n\mathbf{\Phi} + \mathbf{\Lambda}\right)_{{\boldsymbol{t}}{\boldsymbol{t}}}^{-1}\hat{\boldsymbol{a}}_{\boldsymbol{t}} \right\}.
\end{equation*}

\noindent
Using the fact that $\boldsymbol{\xi}_{{\boldsymbol{l} \setminus \boldsymbol{t}}}^0 = {\bf 0}$, we get 
\begin{eqnarray*}
& & \left[ \left( n\mathbf{\Phi + \mathbf{\Lambda}}\right)_{{\boldsymbol{l}} {\boldsymbol{l}}} \boldsymbol{\xi}_{{\boldsymbol{l}}}^0 + n{\hat{\boldsymbol{a}}}_{\boldsymbol{l}} \right]' \left(n\mathbf{\Phi} + \mathbf{\Lambda} \right)_{{\boldsymbol{l}} {\boldsymbol{l}}}^{-1} \left[ \left( n\mathbf{\Phi + \mathbf{\Lambda}}\right)_{{\boldsymbol{l}} {\boldsymbol{l}}} \boldsymbol{\xi}_{{\boldsymbol{l}}}^0 + n{\hat{\boldsymbol{a}}}_{\boldsymbol{l}} \right] - n^2 \hat{\boldsymbol{a}}_{\boldsymbol{l}}' \left( n\mathbf{\Phi} + \mathbf{\Lambda}\right)_{{\boldsymbol{l}}{\boldsymbol{l}}}^{-1}\hat{\boldsymbol{a}}_{\boldsymbol{l}}\\
&=& 2 n \left( \boldsymbol{\xi}_{{\boldsymbol{l}}}^0 \right)' \hat{\boldsymbol{a}}_{\boldsymbol{l}} + 
\left( \boldsymbol{\xi}_{{\boldsymbol{l}}}^0 \right)' \left(n\mathbf{\Phi} + 
\mathbf{\Lambda} \right)_{{\boldsymbol{l}} {\boldsymbol{l}}} \left( \boldsymbol{\xi}_{{\boldsymbol{l}}}^0 \right)\\
&=& 2 n \left( \boldsymbol{\xi}_{{\boldsymbol{t}}}^0 \right)' \hat{\boldsymbol{a}}_{\boldsymbol{t}} + 
\left( \boldsymbol{\xi}_{{\boldsymbol{t}}}^0 \right)' \left(n\mathbf{\Phi} + 
\mathbf{\Lambda} \right)_{{\boldsymbol{t}} {\boldsymbol{t}}} \left( \boldsymbol{\xi}_{{\boldsymbol{t}}}^0 \right)\\
&=& \left[ \left( n\mathbf{\Phi + \mathbf{\Lambda}}\right)_{{\boldsymbol{t}} {\boldsymbol{t}}} \boldsymbol{\xi}_{{\boldsymbol{t}}}^0 + n{\hat{\boldsymbol{a}}}_{\boldsymbol{t}} \right]' \left(n\mathbf{\Phi} + \mathbf{\Lambda} \right)_{{\boldsymbol{t}} {\boldsymbol{t}}}^{-1} \left[ \left( n\mathbf{\Phi + \mathbf{\Lambda}}\right)_{{\boldsymbol{t}} {\boldsymbol{t}}} \boldsymbol{\xi}_{{\boldsymbol{t}}}^0 + n{\hat{\boldsymbol{a}}}_{\boldsymbol{t}} \right] - n^2 \hat{\boldsymbol{a}}_{\boldsymbol{t}}' \left( n\mathbf{\Phi} + \mathbf{\Lambda}\right)_{{\boldsymbol{t}}{\boldsymbol{t}}}^{-1}\hat{\boldsymbol{a}}_{\boldsymbol{t}}\\
&\geq& - n^2 \hat{\boldsymbol{a}}_{\boldsymbol{t}}' \left( n\mathbf{\Phi} + \mathbf{\Lambda}\right)_{{\boldsymbol{t}}{\boldsymbol{t}}}^{-1}\hat{\boldsymbol{a}}_{\boldsymbol{t}}. 
\end{eqnarray*}

\noindent
It follows that 
\begin{equation*}
PR(\boldsymbol{l}, \boldsymbol{t}) \leq \frac{\|\mathbf{\Lambda}_{{\boldsymbol{l}} {\boldsymbol{l}}} \|^{\frac{1}{2}}}{\| \mathbf{\Lambda}_{\boldsymbol{t}\boldsymbol{t}}\|^{\frac{1}{2}}} \frac{\left(2q\right)^{d_{\boldsymbol{l}} - d_{\boldsymbol{t}}}}{\| \left( n\mathbf{\Phi} + \mathbf{\Lambda}\right)_{{\boldsymbol{l}} | \boldsymbol{t}}\|^{\frac{1}{2}}}\\
\exp \left\{ \frac{1}{2} \left[ \left( n\mathbf{\Phi + \mathbf{\Lambda}}\right)_{{\boldsymbol{l}} {\boldsymbol{l}}} \boldsymbol{\xi}_{{\boldsymbol{l}}}^0 + n{\hat{\boldsymbol{a}}}_{\boldsymbol{l}} \right]' \left(n\mathbf{\Phi} + \mathbf{\Lambda} \right)_{{\boldsymbol{l}} {\boldsymbol{l}}}^{-1} \left[ \left( n\mathbf{\Phi + \mathbf{\Lambda}}\right)_{{\boldsymbol{l}} {\boldsymbol{l}}} \boldsymbol{\xi}_{{\boldsymbol{l}}}^0 + n{\hat{\boldsymbol{a}}}_{\boldsymbol{l}} \right] \right\}. 
\end{equation*}

\noindent
Now, we note by Lemma \ref{Lemma 4} that each entry of 
\begin{equation*}
\left( n\mathbf{\Phi + \mathbf{\Lambda}}\right)_{{\boldsymbol{l}} {\boldsymbol{l}}} \boldsymbol{\xi}_{{\boldsymbol{l}}}^0 + n{\hat{\boldsymbol{a}}}_{\boldsymbol{l}} = n\left( \mathbf{\Phi} \boldsymbol{\xi}^0 + {{\boldsymbol{a}}}\right)_{\boldsymbol{l}} + \mathbf{\Lambda}_{{\boldsymbol{l}} {\boldsymbol{l}}} \boldsymbol{\xi}_{{\boldsymbol{l}}}^0, + n \left( {\hat{\boldsymbol{a}}} - {{\boldsymbol{a}}}\right)_{\boldsymbol{l}}
\end{equation*}

\noindent
is smaller in absolute value than 
\begin{equation*}
nc_0 \sqrt{\frac{\log p}{n}} + \frac{\| \mathbf{\Lambda} \|_{\max}}{\tilde{\varepsilon}_0} + \frac{3C}{\tilde{\varepsilon}_0} \sqrt{\frac{\log p}{n}}\leq \frac{3nc_0}{2} \sqrt{\frac{\log p}{n}},
\end{equation*}
hence, by Lemma \ref{Lemma 3},
\begin{equation*}
\begin{split}
&\left[ \left( n\mathbf{\Phi + \mathbf{\Lambda}}\right)_{{\boldsymbol{l}} {\boldsymbol{l}}} \boldsymbol{\xi}_{{\boldsymbol{l}}}^0 + n{\hat{\boldsymbol{a}}}_{\boldsymbol{l}} \right]' \left(n\mathbf{\Phi} + \mathbf{\Lambda} \right)_{{\boldsymbol{l}} {\boldsymbol{l}}}^{-1} \left[ \left( n\mathbf{\Phi + \mathbf{\Lambda}}\right)_{{\boldsymbol{l}} {\boldsymbol{l}}} \boldsymbol{\xi}_{{\boldsymbol{l}}}^0 + n{\hat{\boldsymbol{a}}}_{\boldsymbol{l}} \right] \\
& \leq \frac{1}{n\tilde{\varepsilon}_0} d_{\boldsymbol{l}} \frac{4n^2 c_0^2 \log p}{n} = \frac{4c_0^2 d_{\boldsymbol{l}} \log p}{\tilde{\varepsilon}_0}.
\end{split}
\end{equation*}
Hence,
\begin{eqnarray}
PR(\boldsymbol{l}, \boldsymbol{t}) 
&\leq& \frac{\|\mathbf{\Lambda}_{{\boldsymbol{l}} {\boldsymbol{l}}} \|^{\frac{1}{2}}}{\| \mathbf{\Lambda}_{\boldsymbol{t}\boldsymbol{t}}\|^{\frac{1}{2}}} \frac{\left(2q\right)^{d_{\boldsymbol{l}} - d_t}}{\left(\frac{n\tilde{\varepsilon}_0}{2}\right)^{\frac{d_{\boldsymbol{l}} - d_{\boldsymbol{t}}}{2}}} \exp \left\{\frac{2 c_0^2}{\tilde{\varepsilon}_0}d_{\boldsymbol{l}} \log p \right\} 
\nonumber\\
&=& \frac{2^{d_{\boldsymbol{l}} - d_{\boldsymbol{t}}}\|\mathbf{\Lambda}_{{\boldsymbol{l}} {\boldsymbol{l}}} \|^{\frac{1}{2}}}{\| \mathbf{\Lambda}_{\boldsymbol{t}\boldsymbol{t}}\|^{\frac{1}{2}}} \frac{q^{d_{\boldsymbol{l}} - d_t} p^{\frac{2 c_0^2 d_{\boldsymbol{l}}}{\tilde{\varepsilon}_0}}}{\left(\frac{n\tilde{\varepsilon}_0}{2}\right)^{\frac{d_{\boldsymbol{l}} - d_{\boldsymbol{t}}}{2}}} \nonumber\\
&\leq& \frac{2^{d_{\boldsymbol{l}} - d_{\boldsymbol{t}}}\|\mathbf{\Lambda}_{{\boldsymbol{l}} {\boldsymbol{l}}} \|^{\frac{1}{2}}}{\| \mathbf{\Lambda}_{\boldsymbol{t}\boldsymbol{t}}\|^{\frac{1}{2}} \left(\frac{n\tilde{\varepsilon}_0}{2}\right)^{\frac{d_{\boldsymbol{l}} - d_{\boldsymbol{t}}}{2}}} p^{-a_2 d_{\boldsymbol{t}} (d_{\boldsymbol{l}} - 
d_{\boldsymbol{t}}) + a_2 d_{\boldsymbol{l}}/4} \nonumber\\
&=& \frac{2^{d_{\boldsymbol{l}} - d_{\boldsymbol{t}}}\|\mathbf{\Lambda}_{{\boldsymbol{l}} {\boldsymbol{l}}} \|^{\frac{1}{2}}}{\| \mathbf{\Lambda}_{\boldsymbol{t}\boldsymbol{t}}\|^{\frac{1}{2}} \left(\frac{n\tilde{\varepsilon}_0}{2}\right)^{\frac{d_{\boldsymbol{l}} - d_{\boldsymbol{t}}}{2}}} \left( p^{-a_2 d_{\boldsymbol{t}} + a_2 \frac{d_{\boldsymbol{l}}}
{4 (d_{\boldsymbol{l}} - d_{\boldsymbol{t}})}} \right)^{d_{\boldsymbol{l}} - d_{\boldsymbol{t}}} \nonumber\\
&=& \frac{2^{d_{\boldsymbol{l}} - d_{\boldsymbol{t}}}\|\mathbf{\Lambda}_{{\boldsymbol{l}} {\boldsymbol{l}}} \|^{\frac{1}{2}}}{\| \mathbf{\Lambda}_{\boldsymbol{t}\boldsymbol{t}}\|^{\frac{1}{2}} \left(\frac{n\tilde{\varepsilon}_0}{2}\right)^{\frac{d_{\boldsymbol{l}} - d_{\boldsymbol{t}}}{2}}} \left( p^{-a_2 d_{\boldsymbol{t}} + a_2 \frac{d_{\boldsymbol{t}}}
{2}} \right)^{d_{\boldsymbol{l}} - d_{\boldsymbol{t}}} \nonumber\\
&\leq& \left( \frac{4 C_1 p^{\frac{-a_2 d_{\boldsymbol{t}}}{2}}}{\sqrt{n\tilde{\varepsilon}_0}} \right)^{d_{\boldsymbol{l}} - d_{\boldsymbol{t}}}
\end{eqnarray}

\noindent
for an appropriate constant $C_2$. 
\end{proof}

\begin{Lemma}\label{Lemma 7}
Let $\boldsymbol{l} \in \mathcal{L}$ such that $\boldsymbol{l} \not\subset \boldsymbol{t}$, $\boldsymbol{t} \not\subset 
\boldsymbol{l}$, $\boldsymbol{l} \neq \boldsymbol{t}$ and $d_{\boldsymbol{l}} \leq \tau_n$, then under 
Assumptions \ref{Assumption2: d_t} - \ref{Assumption6: Decay rate} for sufficiently large n, 
\begin{equation*}
PR({\boldsymbol{l}}, \boldsymbol{t}) \to 0, \quad \quad \quad \text{as} \quad n \to \infty . 
\end{equation*}
\end{Lemma}
\begin{proof}
By using the formula for the inverse of a partitioned matrix, it can be shown that 
\begin{equation} \label{partition:quadratic:form}
\left( \begin{matrix} {\bf x}_1' & {\bf x}_2' \end{matrix} \right) \left( \begin{matrix} {\bf A}_{11} & {\bf A}_{12} \cr 
{\bf A}_{12}' & {\bf A}_{22} \end{matrix} \right)^{-1} \left( \begin{matrix} {\bf x}_1 \cr {\bf x}_2 \end{matrix} \right) - {\bf x}_1' 
{\bf A}_{11}^{-1} {\bf x}_1 = \left( {\bf x}_2 - {\bf A}_{12}' {\bf A}_{11}^{-1} {\bf x}_1 \right)' {\bf F}^{-1} \left( {\bf x}_2 - 
{\bf A}_{12}' {\bf A}_{11}^{-1} {\bf x}_1 \right) \geq 0, 
\end{equation}

\noindent
where ${\bf F} = {\bf A}_{22} - {\bf A}_{12}' {\bf A}_{11}^{-1} {\bf A}_{12}$. 

Suppose $\boldsymbol{l}$ is such that $d_{\boldsymbol{l}} > d_{\boldsymbol{t}}$. Let $\tilde{\boldsymbol{l}}$ denote the 
union of $\boldsymbol{l}$ and $\boldsymbol{t}$. Then $\tilde{\boldsymbol{l}} \supset \boldsymbol{t}$ and 
$$
d_{\tilde{\boldsymbol{l}}} \leq d_{\boldsymbol{l}} + d_{\boldsymbol{t}} \leq \tau_n + d_{\boldsymbol{t}}. 
$$

\noindent
As in the proof of Lemma \ref{Lemma 6}, using the fact ${\boldsymbol \xi}_{\tilde{\boldsymbol{l}} \setminus \boldsymbol{t}} = 
{\bf 0}$ and (\ref{partition:quadratic:form}), we get that 
\begin{eqnarray*}
& & n^2 \hat{\boldsymbol{a}}_{\boldsymbol{l}}' \left( n\mathbf{\Phi} + \mathbf{\Lambda}\right)_{\boldsymbol{l} \boldsymbol{l}}^{-1}\hat{\boldsymbol{a}}_{\boldsymbol{l}} - n^2 \hat{\boldsymbol{a}}_{\boldsymbol{t}}' \left( n\mathbf{\Phi} + \mathbf{\Lambda}\right)_{{\boldsymbol{t}}{\boldsymbol{t}}}^{-1}\hat{\boldsymbol{a}}_{\boldsymbol{t}}\\
&\leq& n^2 \hat{\boldsymbol{a}}_{\tilde{\boldsymbol{l}}}' \left( n\mathbf{\Phi} + \mathbf{\Lambda}\right)_{\tilde{\boldsymbol{l}}\tilde{\boldsymbol{l}}}^{-1}\hat{\boldsymbol{a}}_{\tilde{\boldsymbol{l}}} - n^2 \hat{\boldsymbol{a}}_{\boldsymbol{t}}' \left( n\mathbf{\Phi} + \mathbf{\Lambda}\right)_{{\boldsymbol{t}}{\boldsymbol{t}}}^{-1}\hat{\boldsymbol{a}}_{\boldsymbol{t}}\\
&\leq& \left[ \left( n\mathbf{\Phi + \mathbf{\Lambda}}\right)_{\tilde{\boldsymbol{l}} \tilde{\boldsymbol{l}}} \boldsymbol{\xi}_{\tilde{\boldsymbol{l}}}^0 + n{\hat{\boldsymbol{a}}}_{\tilde{\boldsymbol{l}}} \right]' \left(n\mathbf{\Phi} + \mathbf{\Lambda} \right)_{\tilde{\boldsymbol{l}} \tilde{\boldsymbol{l}}}^{-1} \left[ \left( n\mathbf{\Phi + \mathbf{\Lambda}}\right)_{\tilde{\boldsymbol{l}} \tilde{\boldsymbol{l}}} \boldsymbol{\xi}_{\tilde{\boldsymbol{l}}}^0 + n{\hat{\boldsymbol{a}}}_{\tilde{\boldsymbol{l}}} \right] - \\
& & \left[ \left( n\mathbf{\Phi + \mathbf{\Lambda}}\right)_{{\boldsymbol{t}} {\boldsymbol{t}}} \boldsymbol{\xi}_{{\boldsymbol{t}}}^0 + n{\hat{\boldsymbol{a}}}_{\boldsymbol{t}} \right]' \left(n\mathbf{\Phi} + \mathbf{\Lambda} \right)_{{\boldsymbol{t}} {\boldsymbol{t}}}^{-1} \left[ \left( n\mathbf{\Phi + \mathbf{\Lambda}}\right)_{{\boldsymbol{t}} {\boldsymbol{t}}} \boldsymbol{\xi}_{{\boldsymbol{t}}}^0 + n{\hat{\boldsymbol{a}}}_{\boldsymbol{t}} \right]\\
&\leq& \frac{4 c_0^2 d_{\tilde{\boldsymbol{l}}} \log p}{\tilde{\varepsilon_0}} \leq \frac{4 c_0^2 (d_{\boldsymbol{l}} + 
d_{\boldsymbol{t}}) \log p}{\tilde{\varepsilon_0}}
\end{eqnarray*}

\noindent
It follows that 
\begin{eqnarray}
PR(\boldsymbol{l}, \boldsymbol{t}) 
&\leq& \frac{\|\mathbf{\Lambda}_{{\boldsymbol{l}} {\boldsymbol{l}}} \|^{\frac{1}{2}}}{\| \mathbf{\Lambda}_{\boldsymbol{t}\boldsymbol{t}}\|^{\frac{1}{2}}} \frac{\left(2q\right)^{d_{\boldsymbol{l}} - d_t}}{\left(\frac{n\tilde{\varepsilon}_0}{2}\right)^{\frac{d_{\boldsymbol{l}} - d_{\boldsymbol{t}}}{2}}} \exp \left\{\frac{2 c_0^2}{\tilde{\varepsilon}_0} (d_{\boldsymbol{l}} + 
d_{\boldsymbol{t}}) \log p \right\} \nonumber\\
&=& \frac{2^{d_{\boldsymbol{l}} - d_{\boldsymbol{t}}}\|\mathbf{\Lambda}_{{\boldsymbol{l}} {\boldsymbol{l}}} \|^{\frac{1}{2}}}{\| \mathbf{\Lambda}_{\boldsymbol{t}\boldsymbol{t}}\|^{\frac{1}{2}}} \frac{q^{d_{\boldsymbol{l}} - d_t} p^{\frac{2 c_0^2 (d_{\boldsymbol{l}} + d_{\boldsymbol{t}})}{\tilde{\varepsilon}_0}}}{\left(\frac{n\tilde{\varepsilon}_0}{2}\right)^{\frac{d_{\boldsymbol{l}} - d_{\boldsymbol{t}}}{2}}} \nonumber\\
&=& \frac{2^{d_{\boldsymbol{l}} - d_{\boldsymbol{t}}}\|\mathbf{\Lambda}_{{\boldsymbol{l}} {\boldsymbol{l}}} \|^{\frac{1}{2}}}{\| \mathbf{\Lambda}_{\boldsymbol{t}\boldsymbol{t}}\|^{\frac{1}{2}} \left(\frac{n\tilde{\varepsilon}_0}{2}\right)^{\frac{d_{\boldsymbol{l}} - d_{\boldsymbol{t}}}{2}}} \left( p^{-a_2 d_{\boldsymbol{t}} + a_2 
\frac{d_{\boldsymbol{l}} + d_{\boldsymbol{t}}}{4 (d_{\boldsymbol{l}} - d_{\boldsymbol{t}})}} \right)^{d_{\boldsymbol{l}} - d_{\boldsymbol{t}}} \nonumber\\
&=& \frac{2^{d_{\boldsymbol{l}} - d_{\boldsymbol{t}}}\|\mathbf{\Lambda}_{{\boldsymbol{l}} {\boldsymbol{l}}} \|^{\frac{1}{2}}}{\| \mathbf{\Lambda}_{\boldsymbol{t}\boldsymbol{t}}\|^{\frac{1}{2}} \left(\frac{n\tilde{\varepsilon}_0}{2}\right)^{\frac{d_{\boldsymbol{l}} - d_{\boldsymbol{t}}}{2}}} \left( p^{-a_2 d_{\boldsymbol{t}} + a_2 \frac{3 d_{\boldsymbol{t}}}
{4}} \right)^{d_{\boldsymbol{l}} - d_{\boldsymbol{t}}} \nonumber\\
&\leq& \left( \frac{4 C_2 p^{\frac{-a_2 d_{\boldsymbol{t}}}{4}}}{\sqrt{n\tilde{\varepsilon}_0}} \right)^{d_{\boldsymbol{l}} - 
d_{\boldsymbol{t}}}
\end{eqnarray}

\noindent
where $C_2$ is as in the proof of Lemma \ref{Lemma 6}. Let $D(\boldsymbol{l}, \boldsymbol{t})$ denotes the total number of 
disagreements between $\boldsymbol{l}$ and $\boldsymbol{t}$. Note that if $d_{\boldsymbol{l}} > d_{\boldsymbol{t}}$, then 
$$
D(\boldsymbol{l}, \boldsymbol{t}) \leq 2 d_{\boldsymbol{t}} (d_{\boldsymbol{l}} - d_{\boldsymbol{t}}). 
$$

\noindent
Hence 
$$
PR(\boldsymbol{l}, \boldsymbol{t}) \leq \left( \frac{4 C_2 p^{\frac{-a_2}{8}}}{\sqrt{n\tilde{\varepsilon}_0}} 
\right)^{D(\boldsymbol{l}, \boldsymbol{t})}. 
$$

Suppose $\boldsymbol{l}$ is such that $d_{\boldsymbol{l}} \leq d_{\boldsymbol{t}}$. Note that 
\begin{eqnarray*}
& & \hat{\boldsymbol{a}}_{\boldsymbol{l}}' \left( n\mathbf{\Phi} + \mathbf{\Lambda}\right)_{\boldsymbol{l} \boldsymbol{l}}^{-1}\hat{\boldsymbol{a}}_{\boldsymbol{l}} - \hat{\boldsymbol{a}}_{\boldsymbol{t}}' \left( n\mathbf{\Phi} + \mathbf{\Lambda}\right)_{{\boldsymbol{t}}{\boldsymbol{t}}}^{-1}\hat{\boldsymbol{a}}_{\boldsymbol{t}}\\
&=& \hat{\boldsymbol{a}}_{\boldsymbol{l}}' \left( \left( n\mathbf{\Phi} + \mathbf{\Lambda}\right)_{\boldsymbol{l} \boldsymbol{l}}^{-1} - \left( n\mathbf{\Phi} \right)_{\boldsymbol{l} \boldsymbol{l}}^{-1} \right) \hat{\boldsymbol{a}}_{\boldsymbol{l}} - \hat{\boldsymbol{a}}_{\boldsymbol{t}}' \left( \left( n\mathbf{\Phi} + \mathbf{\Lambda}\right)_{{\boldsymbol{t}}{\boldsymbol{t}}}^{-1} - \left( n\mathbf{\Phi} \right)_{{\boldsymbol{t}}{\boldsymbol{t}}}^{-1} \right) \hat{\boldsymbol{a}}_{\boldsymbol{t}} +\\
& & \frac{1}{n} \hat{\boldsymbol{a}}_{\boldsymbol{l}}' \left( \mathbf{\Phi} \right)_{\boldsymbol{l} \boldsymbol{l}}^{-1}\hat{\boldsymbol{a}}_{\boldsymbol{l}} - \frac{1}{n} \hat{\boldsymbol{a}}_{\boldsymbol{t}}' \left( \mathbf{\Phi}\right)_{{\boldsymbol{t}}{\boldsymbol{t}}}^{-1}\hat{\boldsymbol{a}}_{\boldsymbol{t}}\\
&=& O \left( \frac{d_{\boldsymbol{t}}}{n^2} \right) + \frac{1}{n} \hat{\boldsymbol{a}}_{\boldsymbol{l}}' \left( \mathbf{\Phi} \right)_{\boldsymbol{l} \boldsymbol{l}}^{-1}\hat{\boldsymbol{a}}_{\boldsymbol{l}} - \frac{1}{n} \hat{\boldsymbol{a}}_{\boldsymbol{t}}' \left( \mathbf{\Phi}\right)_{{\boldsymbol{t}}{\boldsymbol{t}}}^{-1}\hat{\boldsymbol{a}}_{\boldsymbol{t}} 
\end{eqnarray*}

\noindent
and by Lemma \ref{Lemma 4} 
$$
\left( \mathbf{\Phi} \boldsymbol{\xi}^0 + \hat{\boldsymbol{a}} \right)_{{\boldsymbol{l}}}' \left( \mathbf{\Phi} 
\right)_{\boldsymbol{l} \boldsymbol{l}}^{-1} \left( \mathbf{\Phi} \boldsymbol{\xi}^0 + \hat{\boldsymbol{a}}
\right)_{{\boldsymbol{l}}} + \left( \mathbf{\Phi} \boldsymbol{\xi}^0 + \hat{\boldsymbol{a}} \right)_{{\boldsymbol{t}}}' 
\left( \mathbf{\Phi} \right)_{\boldsymbol{t} \boldsymbol{t}}^{-1} \left( \mathbf{\Phi} \boldsymbol{\xi}^0 + \hat{\boldsymbol{a}} 
\right)_{{\boldsymbol{t}}} = O \left( \frac{d_{\boldsymbol{t}} \log p}{n} \right) 
$$

\noindent
on $C_{1,n}$. Let $\boldsymbol{l}^c$ denote the sparsity pattern which has a zero/one whenever the corresponding entry 
in $\boldsymbol{l}$ is one/zero. Using $\boldsymbol{\xi}^0_{\boldsymbol{t}^c} = {\bf 0}$ and Lemma \ref{Lemma 3}, it 
follows that 
\begin{eqnarray*}
& & \left( \mathbf{\Phi} \boldsymbol{\xi}^0 \right)_{{\boldsymbol{l}}}' \left( \mathbf{\Phi} 
\right)_{\boldsymbol{l} \boldsymbol{l}}^{-1} \left( \mathbf{\Phi} \boldsymbol{\xi}^0 \right)_{{\boldsymbol{l}}} - \left( \mathbf{\Phi} 
\boldsymbol{\xi}^0 \right)_{{\boldsymbol{t}}}' \left( \mathbf{\Phi} \right)_{\boldsymbol{t} \boldsymbol{t}}^{-1} \left( \mathbf{\Phi} 
\boldsymbol{\xi}^0 \right)_{{\boldsymbol{t}}}\\ 
&=& {\boldsymbol{\xi}^0_{\boldsymbol{l}}}' \mathbf{\Phi}_{\boldsymbol{l} \boldsymbol{l}} \boldsymbol{\xi}^0_{\boldsymbol{l}} 
+ {\boldsymbol{\xi}^0_{\boldsymbol{l}}}' \mathbf{\Phi}_{\boldsymbol{l} \boldsymbol{l}^c} \boldsymbol{\xi}^0_{\boldsymbol{l}^c} 
+ {\boldsymbol{\xi}^0_{\boldsymbol{l}^c}}' \mathbf{\Phi}_{\boldsymbol{l}^c \boldsymbol{l}} \mathbf{\Phi}_{\boldsymbol{l} 
\boldsymbol{l}}^{-1} \mathbf{\Phi}_{\boldsymbol{l} \boldsymbol{l}^c} \boldsymbol{\xi}^0_{\boldsymbol{l}^c} - 
{\boldsymbol{\xi}^0_{\boldsymbol{t}}}' \mathbf{\Phi}_{\boldsymbol{t} \boldsymbol{t}} \boldsymbol{\xi}^0_{\boldsymbol{t}}\\
&=& {\boldsymbol{\xi}^0}' \mathbf{\Phi} \boldsymbol{\xi}^0 - {\boldsymbol{\xi}^0}' \mathbf{\Phi} \boldsymbol{\xi}^0 
- {\boldsymbol{\xi}^0_{\boldsymbol{l}^c}}' \left( \mathbf{\Phi}_{\boldsymbol{l}^c \boldsymbol{l}^c} - 
\mathbf{\Phi}_{\boldsymbol{l}^c \boldsymbol{l}} \mathbf{\Phi}_{\boldsymbol{l} \boldsymbol{l}}^{-1} 
\mathbf{\Phi}_{\boldsymbol{l} \boldsymbol{l}^c} \right) \boldsymbol{\xi}^0_{\boldsymbol{l}^c}\\
&\leq& - \frac{3 d_{\boldsymbol{t} \cap \boldsymbol{l}^c} \tilde{\varepsilon}_0 s_n^2}{4, }
\end{eqnarray*}

\noindent
since exactly $d_{\boldsymbol{t} \cap \boldsymbol{l}^c}$ entries in $\boldsymbol{\xi}^0_{\boldsymbol{l}^c}$ are non-zero. 
Since $d_{\boldsymbol{t} \cap \boldsymbol{l}^c} \geq d_{\boldsymbol{t}} - d_{\boldsymbol{l}}$ and $D(\boldsymbol{l}, 
\boldsymbol{t}) \leq 2 d_{\boldsymbol{t}}$ similar arguments to those at the end of Lemma \ref{Lemma 5} can be used to 
obtain 
$$
PR(\boldsymbol{l}, \boldsymbol{t}) \leq \left( \frac{2 C_1 p^{-a_2 d_{\boldsymbol{t}}}}{\sqrt{n}} \right)^{d_{\boldsymbol{t} \cap 
\boldsymbol{l}^c}} \leq \left( \frac{2 C_1 p^{-a_2/2}}{\sqrt{n}} \right)^{D(\boldsymbol{l}, \boldsymbol{t})}. 
$$

\end{proof}

\noindent
It follows by Lemmas \ref{Lemma 5}, \ref{Lemma 6} and \ref{Lemma 7} that for every $\boldsymbol{l} \neq \boldsymbol{t}$ 
with $d_{\boldsymbol{l}} \leq \tau_n$, 
$$
PR(\boldsymbol{l}, \boldsymbol{t}) \leq f_n^{D(\boldsymbol{l}, \boldsymbol{t})} 
$$

\noindent
where 
$$
f_n = \max \left\{  \left( \frac{2 C_1 p^{-a_2/2}}{\sqrt{n}}  \right), \quad  \left( \frac{4 C_2 p^{\frac{-a_2}{8}}}
{\sqrt{n\tilde{\varepsilon}_0}} \right) \right\}. 
$$

The first part of Theorem \ref{theorem (strong consistency)} a is straightforward application of Lemmas \ref{Lemma 5}, 
\ref{Lemma 6}, and \ref{Lemma 7}. Note that $a_2 \geq 16$, which implies $p^2 f_n \to 0$ as $n \to \infty$. It follows that 
\begin{equation}
\begin{split}
\frac{1 - P \left\{ \boldsymbol{\xi}\in\mathcal{M}_{\boldsymbol{t}} | \hat{\boldsymbol{\delta}}, \mathcal{Y}\right\}}{P \left\{ \boldsymbol{\xi}\in\mathcal{M}_{\boldsymbol{t}} | \hat{\boldsymbol{\delta}}, \mathcal{Y}\right\}}& = \sum\limits_{\boldsymbol{l}\neq \boldsymbol{t}} PR(\boldsymbol{l}, \boldsymbol{t}) \\
&= \sum\limits_{\boldsymbol{l}\neq \boldsymbol{t}} \sum\limits_{j=1}^{\binom{p}{2}} PR(\boldsymbol{l}, \boldsymbol{t})I_{\{D(\boldsymbol{l}, \boldsymbol{t})=j\}}\\
&\leq \sum\limits_{j=1}^{\binom{p}{2}} \binom{\binom{p}{2}}{j} f_n^{j}\\
& \leq  \sum\limits_{j=1}^{\binom{p}{2}} {\binom{p}{2}}^j f_n^{j}\\
& \leq \sum\limits_{j=1}^{p^2} \left(p^2 f_n\right)^j \\
& \leq \frac{p^2 f_n}{1 - p^2 f_n} \to 0 \quad \text{as} \quad n \to \infty.
\end{split}
\end{equation}

\noindent
on $C_{1,n}$. Since $\mathbb{P}_0 (C_{1,n}) \rightarrow 1$ as $n \rightarrow \infty$, we get that 
$$
\pi \left\{ \boldsymbol{t} | \hat{\boldsymbol{\delta}},\mathcal{Y} \right\} \xrightarrow{\text{$\mathbb{P}_0$}} 1, 
\quad \quad \text{as} \quad n \to \infty. 
$$

\noindent
We now prove the second part of Theorem \ref{theorem (strong consistency)}. For every pair $(j,k)$, let $\pi_{jk}$ 
denote the posterior probability that $\omega_{jk}$ is non-zero. Note that 
$$
\pi_{jk} \geq \pi \left\{ \boldsymbol{t} | \hat{\boldsymbol{\delta}},\mathcal{Y}\right\} 
$$

\noindent
for $(j,k) \in E^0$, and 
$$
\pi_{jk} \leq 1 - \pi \left\{ \boldsymbol{t} | \hat{\boldsymbol{\delta}},\mathcal{Y}\right\} 
$$

\noindent
for $(j,k) \notin E^0$. It follows by the first part of Theorem \ref{theorem (strong consistency)} that 
\begin{eqnarray*}
\mathbb{P}_0 \left( \hat{\boldsymbol{l}}_{\upsilon, BSSC} = \boldsymbol{t} \right) 
&=& \mathbb{P}_0 \left( \left\{ \cap_{(j,k): (j,k) \in E^0} \{\pi_{jk} \geq v\} \right\} \cap \left\{ 
\cap_{(j,k): (j,k) \notin E^0} \{\pi_{jk} < v\} \right\} \right)\\
&\geq& \mathbb{P}_0 \left( \pi \left\{ \boldsymbol{t} | \hat{\boldsymbol{\delta}},\mathcal{Y} \right\} \geq 
\max \left( v, 1- \frac{v}{2} \right) \right) \to 1 \quad \text{as} \quad n \to \infty. 
\end{eqnarray*}

\section{Appendix B: Proofs of Lemmas 1 and 2} \label{Proofs of Lemmas 1 and 2}

\subsection*{Proof of Lemma 1}

\noindent
Note that any principal submatrix of $\mathbf{S}$ of size less than $n$ is positive definite. Let $\lambda > 0$ be the smallest 
number in the collection of eigenvalues of all principal submatrices of $S$ of size less than $n$. By assumption, if 
$\mathbf{\Omega} \in \mathbb{M}_{\hat{G}}$, then the $i^{th}$ column of $\mathbf{\Omega}$, denoted by 
$\mathbf{\Omega}_{\cdot i}$, has at most $n-1$ zeros. It follows that 
\begin{eqnarray*}
& & \int_{\mathbb{M}_{\hat{G}}} \exp \left\{ n tr(\mathbf{\Omega}) - \frac{n}{2} tr(\mathbf{\Omega}^2 \mathbf{S}) \right\} 
d \mathbf{\Omega}\\
&=& \int_{\mathbb{M}_{\hat{G}}} \exp \left\{ n tr(\mathbf{\Omega}) - \frac{n}{2} \sum_{i=1}^n \mathbf{\Omega}_{\cdot i}^t  
\mathbf{S} \mathbf{\Omega}_{\cdot i} \right\} d \mathbf{\Omega}\\
&\leq& \int_{\mathbb{M}_{\hat{G}}} \exp \left\{ n tr(\mathbf{\Omega}) - \frac{n \lambda}{2} \sum_{i=1}^n 
\mathbf{\Omega}_{\cdot i}^t \mathbf{\Omega}_{\cdot i} \right\} d \mathbf{\Omega}\\
&=& \int_{\mathbb{R}^{|\hat{E}|} \times \mathbb{R}_+^p} \exp \left\{ n \sum_{i=1}^p \omega_{ii} - \frac{n \lambda}{2} 
\sum_{i=1}^p \sum_{j=1}^p \omega_{ij}^2 \right\} \prod_{i < j, (i,j) \in \hat{E}} d \omega_{ij}\\
&=& \left( \prod_{i=1}^p \int_{\mathbb{R}_+} \exp \left\{ n \omega_{ii} - \frac{n \lambda}{2} \omega_{ii}^2 \right\} d \omega_{ii} 
\right) \left( \prod_{i<j, (i,j) \in \hat{E}} \int_\mathbb{R} \exp \left\{ - n \lambda \omega_{ij}^2 \right\} d \omega_{ij} \right)\\
&<& \infty. 
\end{eqnarray*}

\subsection*{Proof of Lemma 2}

\noindent
Let $\boldsymbol{\omega}_{\hat{G}}$ denote the vectorized version of the non-zero entries in $\mathbf{\Omega} \in 
\mathbb{M}_{\hat{G}}$. It follows that 
$$
h(\boldsymbol{\omega}_{\hat{G}}) = \frac{n}{2} tr(\mathbf{\Omega}^2 \mathbf{K}^{-1}) - n tr(\mathbf{\Omega}) = \frac{n}{2} 
\boldsymbol{\omega}_{\hat{G}}' \tilde{K} \boldsymbol{\omega}_{\hat{G}} - n \boldsymbol{\omega}_{\hat{G}}' {\bf u}, 
$$

\noindent
for an appropriate matrix $\tilde{K}$ and an appropriate vector ${\bf u}$. By a similar analysis as in the proof of 
Lemma \ref{Lemma 2}, it can be shown that the eigenvalues of $\tilde{K}$ are bounded below by the smallest eigenvalue of 
$K^{-1}$. It follows that $\tilde{K}$ is invertible and $h(\boldsymbol{\omega}_{\hat{G}})$ is uniquely minimized at $\tilde{K}^{-1} 
{\bf u}$. Note that for $(i,j) \in \hat{E}$, 
$$
\frac{\partial}{\partial \omega_{ij}} h(\boldsymbol{\omega}_{\hat{G}}) = n \sum_{i'=1}^p \omega_{i'j} K^{-1}_{ii'} + n \sum_{j'=1}^p 
\omega_{ij'} K^{-1}_{jj'}, 
$$

\noindent
and 
$$
\frac{\partial}{\partial \omega_{ii}} h(\boldsymbol{\omega}_{\hat{G}}) = n \sum_{i'=1}^p \omega_{ii'} K^{-1}_{ii'} - n 
$$

\noindent
for $1 \leq i \leq p$. It follows that the vectorized version of the non-zero entries (corresponding to entries in $\hat{E}$) 
of the matrix $\mathbf{K}$ satisfies the above first derivative equations, and must coincide with the unique minimum 
$\tilde{K}^{-1} {\bf u}$. Since $\mathbf{K} \in \mathbb{M}_{\hat{G}}$, the result follows.

\section{Appendix C: Proof of  Theorem 2}\label{Proof of Theorem 2}

\noindent
Note by the definition of $\hat{G}$ in (\ref{estimated:edge:set}) and Theorem \ref{theorem (strong consistency)} that 
$$
\mathbb{P}_0 (\hat{G} = G^0) = \mathbb{P}_0 \left( \hat{\boldsymbol{l}}_{\upsilon, BSSC} = \boldsymbol{t} \right) \rightarrow 
1
$$

\noindent
as $n \rightarrow \infty$. For ease of presentation, let $\epsilon_n = \sqrt{\frac{(p + d_{\boldsymbol{t}} \log p}{n}}$.  First note that for any constant $K'$, 
\begin{eqnarray}
& & \mathbb{E}_0 \left[ \pi_{refitted}  \left(\| \boldsymbol{\Omega}- \boldsymbol{\Omega}^0 \|_{\max} > K' \nu_{\max} 
\sqrt{d_{\boldsymbol{t}} \frac{\log p}{n}} \mid \mathcal{Y} \right)\right] \nonumber\\
&\leq& \mathbb{E}_0 \left[ \pi_{refitted}  \left(\| \boldsymbol{\Omega}- \boldsymbol{\Omega}^0 \|_{\max} > K' \nu_{\max} 
\sqrt{\frac{d_{\boldsymbol{t}} \log p}{n}} \mid \mathcal{Y} \right) 1_{\{\hat{G} = G^0\}} \right] + \mathbb{P}_0 (G \neq G^0). 
\label{refitted:posterior:bound}
\end{eqnarray}

\noindent
Hence, it is sufficient to prove that $\mathbb{E}_0 \left[ \pi_{refitted}  \left( 
\|\boldsymbol{\Omega}- \boldsymbol{\Omega}^0\|_{FM} > K \sqrt{\frac{\log p}{n}} \mid \mathcal{Y} \right) 1_{\{\hat{G} = 
G^0\}} \right]$ as $n \to \infty$. As in the proof of Lemma \ref{optimizer}, let $\boldsymbol{\omega}$ denote the vectorized 
version of the non-zero entries in $\mathbf{\Omega} \in \mathbb{M}_{G^0}$, such that the first $p$ entries of 
$\mathbf{\omega}$ correspond to the diagonal entries of $\mathbf{\Omega}$, and the last $d_{\boldsymbol{t}}$ entries of 
$\mathbf{\omega}$ correspond to the (structurally) non-zero off-diagonal entries of $\mathbf{\Omega}$. Similarly, let 
$\boldsymbol{\omega}^0$ denote the vectorized version of the non-zero entries in $\mathbf{\Omega} \in 
\mathbb{M}_{G^0}$. It follows that 
$$
h(\boldsymbol{\omega}) = \frac{n}{2} tr(\mathbf{\Omega}^2 S) - n tr(\mathbf{\Omega}) = \frac{n}{2} 
\boldsymbol{\omega}' \tilde{\mathbf{K}} \boldsymbol{\omega} - n \boldsymbol{\omega}' {\bf u}, 
$$

\noindent
for an appropriate matrix $\tilde{K}$ and an appropriate vector ${\bf u}$. It can be shown by straightforward calculations that 
the first $p$ entries of ${\bf u}$ (corresponding to the diagonals) are $1$, and the rest are $0$. By a similar analysis as in the 
proof of Lemma \ref{Lemma 2}, it can be shown that 
\begin{equation} \label{refitted:covariance}
\tilde{\mathbf{K}} = \mathbf{Q}_t' \tilde{\mathbf{P}}_t' \tilde{\mathbf{S}}_t \tilde{\mathbf{P}}_t \mathbf{Q}_t, 
\end{equation}

\noindent
wherein $\tilde{\mathbf{S}}_t$ is a block diagonal matrix with $p$ blocks. For every $1 \leq i \leq p$, the $i^{th}$ block is a 
principal sub-matrix of the sample covariance matrix $S$ corresponding to indices which are neighbors of $i$ in $G^0$, 
$\tilde{\mathbf{P}}_t$ is an appropriate $(p+2 d_{\boldsymbol{t}}) \times (p+d_{\boldsymbol{t}})$ matrix of zeros and ones 
such that each row has exactly one entry equal to $1$, and each column has at most $2$ entries equal to $1$, and 
$\mathbf{Q}_t$ is an appropriate $(p+d_{\boldsymbol{t}}) \times (p+d_{\boldsymbol{t}})$ permutation matrix. It follows 
from (\ref{refitted:covariance}) and the structure of $\tilde{\mathbf{P}}_t$ that $\lambda_{min} (\tilde{\mathbf{K}}) \geq 
\lambda_{min} (\tilde{\mathbf{S}}_t)$. Note that $\tilde{\mathbf{S}}_t$ is a block diagonal matrix, and each diagonal block 
is a sub-matrix of the sample covariance matrix $\mathbf{S}$ of size less than $d_{\boldsymbol{t}}$. It follows by 
Assumption \ref{Assumption3: Bounded eigenvalues} that on $C_{1,n}$, $\lambda_{\min} (\tilde{\mathbf{S}}_t) \geq 
\tilde{\varepsilon}_0 - d_{\boldsymbol{t}} \sqrt{\frac{\log p}{n}} \geq \tilde{\varepsilon}_0/2$ for large enough $n$. Hence, 
$\tilde{\mathbf{K}}$ is invertible, and the function $h(\boldsymbol{\omega})$ is uniquely minimized at 
$\hat{\boldsymbol{\omega}} = \tilde{\mathbf{K}}^{-1} {\bf u}$. 

Note that by Lemma \ref{optimizer}, the function $h_0 (\boldsymbol{\omega})$ defined by 
$$
h_0 (\boldsymbol{\omega}) = \frac{n}{2} tr(\mathbf{\Omega}^2 \mathbf{\Sigma}^0) - n tr(\mathbf{\Omega}) = \frac{n}{2} 
\boldsymbol{\omega}' \tilde{\mathbf{K}}^0 \boldsymbol{\omega} - n \boldsymbol{\omega}' {\bf u} 
$$

\noindent
is uniquely minimized at $\boldsymbol{\omega}^0$. It can be shown that $\boldsymbol{\omega}^0 = 
(\tilde{\mathbf{K}}^0)^{-1} {\bf u}$, where $\tilde{\mathbf{K}} = \mathbf{Q}_t' \tilde{\mathbf{P}}_t' \tilde{\mathbf{\Sigma}}_t 
\tilde{\mathbf{P}}_t \mathbf{Q}_t$, and $\tilde{\mathbf{\Sigma}}_t$ is a block diagonal matrix with $p$ blocks. For every $1 
\leq i \leq p$, the $i^{th}$ block is a principal sub-matrix of the true covariance matrix $\mathbf{\Sigma}^0$ corresponding to 
indices which are neighbors of $i$ in $G^0$. 

Next, we show that on $C_{1,n}$, $\| \hat{\boldsymbol{\omega}} - \boldsymbol{\omega}^0 \|_{max} \leq K' 
\sqrt{d_{\boldsymbol{t}} \frac{\log p}{n}}$ for a large enough constant $K'$ (not depending on $n$). Let $d = \{(i,j): 1 \leq i=j 
\leq p\}$, $o = \{(i,j): i < j, (i,j) \in E^0\}$ and $\bar{o} = \{(i,j): i > j, (i,j) \in E^0\}$. Note that $|d| = p$, and $|o| = |\bar{o}| = 
d_{\boldsymbol{t}}$. After straightforward calculations, it can be shown that 
$$
\tilde{\mathbf{K}} = \left( \begin{matrix}
\tilde{\mathbf{S}}_{t,dd} & \tilde{\mathbf{S}}_{t,do} + \tilde{\mathbf{S}}_{t,d \bar{o}} \cr
\tilde{\mathbf{S}}_{t,od} + \tilde{\mathbf{S}}_{t,\bar{o} d} & \tilde{\mathbf{S}}_{t,oo} + \tilde{\mathbf{S}}_{t, \bar{o} \bar{o}} 
\end{matrix} \right) \mbox{ and } \tilde{\mathbf{K}}^0 = \left( \begin{matrix}
\tilde{\mathbf{\Sigma}}_{t,dd} & \tilde{\mathbf{\Sigma}}_{t,do} + \tilde{\mathbf{\Sigma}}_{t,d \bar{o}} \cr
\tilde{\mathbf{\Sigma}}_{t,od} + \tilde{\mathbf{\Sigma}}_{t,\bar{o} d} & \tilde{\mathbf{\Sigma}}_{t,oo} + 
\tilde{\mathbf{\Sigma}}_{t, \bar{o} \bar{o}} 
\end{matrix} \right),
$$

\noindent
wherein $\tilde{\mathbf{S}}_{t,do}$ denotes the sub-matrix of $\tilde{\mathbf{S}}$ corresponding to the rows in $d$ and 
columns in $o$. Other sub-matrices are similarly defined. 

Using the form of the inverse of a partitioned matrix, along with $\hat{\boldsymbol{\omega}} = \tilde{\mathbf{K}}^{-1} {\bf u}$ 
and $\boldsymbol{\omega}^0 = (\tilde{\mathbf{K}}^0)^{-1} {\bf u}$, it follows that for every $1 \leq i \leq p$, we have 
$$
\hat{\omega}_{ii} = {\bf e}_i' \tilde{\mathbf{S}}_{t,dd}^{-1} {\bf 1}_p  + {\bf e}_i' \tilde{\mathbf{S}}_{t,dd}^{-1} 
(\tilde{\mathbf{S}}_{t,do} + \tilde{\mathbf{S}}_{t,d \bar{o}}) \tilde{\mathbf{S}}_{Schur} (\tilde{\mathbf{S}}_{t,od} + 
\tilde{\mathbf{S}}_{t,\bar{o} d}) \tilde{\mathbf{S}}_{t,dd}^{-1} {\bf 1}_p, 
$$

\noindent
and 
$$
\omega^0_{ii} = {\bf e}_i' \tilde{\mathbf{\Sigma}}_{t,dd}^{-1} {\bf 1}_p  + {\bf e}_i' \tilde{\mathbf{\Sigma}}_{t,dd}^{-1} 
(\tilde{\mathbf{\Sigma}}_{t,do} + \tilde{\mathbf{\Sigma}}_{t,d \bar{o}}) \tilde{\mathbf{\Sigma}}_{Schur} 
(\tilde{\mathbf{\Sigma}}_{t,od} + \tilde{\mathbf{\Sigma}}_{t,\bar{o} d}) \tilde{\mathbf{\Sigma}}_{t,dd}^{-1} {\bf 1}_p. 
$$

\noindent
Here ${\bf e}_i$ is the $i^{th}$ unit vector in $\mathbb{R}^p$, ${\bf 1}_p \in \mathbb{R}^p$ has all entries equal to $1$, and 
$\tilde{\mathbf{S}}_{Schur}$ and $\tilde{\mathbf{\Sigma}}_{Schur}$ are the lower principal $d_{\boldsymbol{t}} \times 
d_{\boldsymbol{t}}$ submatrices of $\tilde{\mathbf{K}}^{-1}$ and $(\tilde{\mathbf{K}}^0)^{-1}$ respectively. We now make 
the following observations using the structure of $\tilde{\mathbf{S}}$ and $\tilde{\mathbf{\Sigma}}$. 
\begin{enumerate}
\item Since $\tilde{\mathbf{S}}_{t,dd}$ and $\tilde{\mathbf{\Sigma}}_{t,dd}$ are diagonal matrices with diagonal entries 
$\{S_{ii}\}_{i=1}^p$ and $\{\Sigma^0_{ii}\}_{i=1}^p$ respectively, it follows that 
$$
\left| {\bf e}_i' \tilde{\mathbf{S}}_{t,dd}^{-1} {\bf 1}_p - {\bf e}_i' \tilde{\mathbf{\Sigma}}_{t,dd}^{-1} {\bf 1}_p \right| 
= \left| \frac{1}{S_{ii}} - \frac{1}{\Sigma^0_{ii}} \right| \leq K_1' \sqrt{\frac{\log p}{n}} 
$$

\noindent
and 
$$
\left\| {\bf e}_i' \tilde{\mathbf{S}}_{t,dd}^{-1} - {\bf e}_i' \tilde{\mathbf{\Sigma}}_{t,dd}^{-1} \right\| = \left| \frac{1}{S_{ii}} - 
\frac{1}{\Sigma^0_{ii}} \right| \leq K_1' \sqrt{\frac{\log p}{n}} 
$$

\noindent
on $C_{1,n}$ for a large enough constant $K_1'$. 
\item Using Assumption \ref{Assumption3: Bounded eigenvalues} along with the structure of $\tilde{\mathbf{P}}_t$, 
$\tilde{\mathbf{S}}_t$ and $\tilde{\mathbf{\Sigma}}_t$, we get that for large enough $n$ 
$$
\|\tilde{\mathbf{S}}_{t,do} + \tilde{\mathbf{S}}_{t,d \bar{o}}\| \leq \|\tilde{K}\| \leq 2 \lambda_{\max} 
(\tilde{\mathbf{S}}_t) \leq \frac{4}{\tilde{\varepsilon}_0} 
$$

\noindent
and 
$$
\|\tilde{\mathbf{\Sigma}}_{t,do} + \tilde{\mathbf{\Sigma}}_{t,d \bar{o}}\| \leq \|\tilde{K}^0\| \leq 2 \lambda_{\max} 
(\tilde{\mathbf{\Sigma}}_t) \leq \frac{2}{\tilde{\varepsilon}_0}, 
$$

\noindent
on $C_{1,n}$. Also, since each structurally non-zero entry of $\tilde{\mathbf{S}}$ and $\tilde{\mathbf{\Sigma}}$ is  an 
appropriate entry of $\mathbf{S}$ and $\mathbf{\Sigma}^0$ respectively, each row of $\tilde{\mathbf{S}}_{t,do} + 
\tilde{\mathbf{S}}_{t,d \bar{o}}$ and $\tilde{\mathbf{\Sigma}}_{t,do} + \tilde{\mathbf{\Sigma}}_{t,d \bar{o}}$ has at most 
$\nu_{\max}$ non-zero entries, and each column of $\tilde{\mathbf{S}}_{t,do} + \tilde{\mathbf{S}}_{t,d \bar{o}}$ and 
$\tilde{\mathbf{\Sigma}}_{t,do} + \tilde{\mathbf{\Sigma}}_{t,d \bar{o}}$ has at most $2$ non-zero entries, it follows that 
for large enough $n$ and a large enough constant $K_2'$ 
$$
\|\tilde{\mathbf{S}}_{t,do} + \tilde{\mathbf{S}}_{t,d \bar{o}} - \tilde{\mathbf{\Sigma}}_{t,do} - 
\tilde{\mathbf{\Sigma}}_{t,d \bar{o}} \| \leq K_2' \sqrt{\frac{\nu_{\max} \log p}{n}} 
$$

\noindent
on $C_{1,n}$. 
\item It can be shown using the structure of $\tilde{\mathbf{S}}_t$ and $\tilde{\mathbf{\Sigma}}_t$ that 
$(\tilde{\mathbf{S}}_{t,od} + \tilde{\mathbf{S}}_{t,\bar{o} d}) \tilde{\mathbf{S}}_{t,dd}^{-1} {\bf 1}_p$ and 
$(\tilde{\mathbf{\Sigma}}_{t,od} + \tilde{\mathbf{\Sigma}}_{t,\bar{o} d}) \tilde{\mathbf{\Sigma}}_{t,dd}^{-1} {\bf 1}_p$ 
are both $d_{\boldsymbol{t}}$-dimensional vectors with the entry corresponding to $(i,j) \in E^0$ given by 
$2 S_{ij} \left( S_{ii}^{-1} + S_{jj}^{-1} \right)$ and $2 \Sigma^0_{ij} \left( (\Sigma^0_{ii})^{-1} + (\Sigma^0_{jj})^{-1} \right)$ 
respectively. It follows that for large enough $n$ and a large enough constant $K_3'$ 
$$
\left\| (\tilde{\mathbf{S}}_{t,od} + \tilde{\mathbf{S}}_{t,\bar{o} d}) \tilde{\mathbf{S}}_{t,dd}^{-1} {\bf 1}_p \right\| \leq K_3' 
\sqrt{d_{\boldsymbol{t}}}, \; \left\| (\tilde{\mathbf{\Sigma}}_{t,od} + \tilde{\mathbf{\Sigma}}_{t,\bar{o} d}) 
\tilde{\mathbf{\Sigma}}_{t,dd}^{-1} {\bf 1}_p \right\| \leq K_3' \sqrt{d_{\boldsymbol{t}}} 
$$

\noindent
and 
$$
\left\| (\tilde{\mathbf{S}}_{t,od} + \tilde{\mathbf{S}}_{t,\bar{o} d}) \tilde{\mathbf{S}}_{t,dd}^{-1} {\bf 1}_p - 
(\tilde{\mathbf{\Sigma}}_{t,od} + \tilde{\mathbf{\Sigma}}_{t,\bar{o} d}) \tilde{\mathbf{\Sigma}}_{t,dd}^{-1} {\bf 1}_p 
\right\| \leq K_3' \sqrt{\frac{d_{\boldsymbol{t}} \log p}{n}} 
$$

\noindent
on $C_{1,n}$. 
\item Using Assumption \ref{Assumption3: Bounded eigenvalues} along with the structure of $\tilde{\mathbf{P}}_t$, 
$\tilde{\mathbf{S}}_t$ and $\tilde{\mathbf{\Sigma}}_t$, we get that for large enough $n$ 
$$
\| \tilde{\mathbf{S}}_{Schur} \| \leq \| \tilde{\mathbf{K}}^{-1} \| \leq \frac{2}{\tilde{\varepsilon}_0}
$$

\noindent
and 
$$
\| \tilde{\mathbf{\Sigma}}_{Schur} \| \leq \| (\tilde{\mathbf{K}}^0)^{-1} \| \leq \frac{1}{\tilde{\varepsilon}_0}
$$

\noindent
on $C_{1,n}$. Using 
$$
\tilde{\mathbf{S}}_{Schur}^{-1} = (\tilde{\mathbf{\Sigma}}_{t,oo} + \tilde{\mathbf{\Sigma}}_{t,\bar{o} \bar{o}}) - 
(\tilde{\mathbf{\Sigma}}_{t,od} + \tilde{\mathbf{\Sigma}}_{t,\bar{o} d}) \tilde{\mathbf{\Sigma}}_{t,dd}^{-1} 
(\tilde{\mathbf{\Sigma}}_{t,do} + \tilde{\mathbf{\Sigma}}_{t,d \bar{o}}) 
$$

\noindent
$$
\tilde{\mathbf{\Sigma}}_{Schur}^{-1} = (\tilde{\mathbf{\Sigma}}_{t,oo} + \tilde{\mathbf{\Sigma}}_{t,\bar{o} \bar{o}}) - 
(\tilde{\mathbf{\Sigma}}_{t,od} + \tilde{\mathbf{\Sigma}}_{t,\bar{o} d}) \tilde{\mathbf{\Sigma}}_{t,dd}^{-1} 
(\tilde{\mathbf{\Sigma}}_{t,do} + \tilde{\mathbf{\Sigma}}_{t,d \bar{o}}), 
$$

\noindent
along with the fact that $\tilde{\mathbf{S}}_{t,oo} + \tilde{\mathbf{S}}_{t,\bar{o} \bar{o}}$ and 
$\tilde{\mathbf{\Sigma}}_{t,oo} + \tilde{\mathbf{\Sigma}}_{t,\bar{o} \bar{o}}$ have at most $2 \nu_{\max}$ structurally 
non-zero entries in each row (and column), we get that for large enough $n$ and a large enough constant $K_4'$ 
$$
\left\| \tilde{\mathbf{S}}_{Schur} - \tilde{\mathbf{\Sigma}}_{Schur} \right\| \leq \nu_{\max} \sqrt{\frac{\log p}{n}} 
$$

\noindent
on $C_{1,n}$. 
\end{enumerate}

\noindent
Using the observations above, it follows that for $K' > 2 \max \left( K_1', K_2', K_3', K_4' \right)$, and large enough $n$ 
\begin{equation} \label{mode:diag:bound}
\max_{1 \leq i \leq p} \left| \hat{\omega}_{ii} - \omega^0_{ii} \right| \leq \frac{K'}{2} \nu_{\max} \sqrt{\frac{d_{\boldsymbol{t}} 
\log p}{n}} 
\end{equation}

\noindent
on $C_{1,n}$. 

Now, for every $(i,j) \in E^0$, using the form of the inverse of a partitioned matrix, along with $\hat{\boldsymbol{\omega}} = 
\tilde{\mathbf{K}}^{-1} {\bf u}$ and $\boldsymbol{\omega}^0 = (\tilde{\mathbf{K}}^0)^{-1} {\bf u}$, we get that 
$$
\hat{\omega}_{ij} = {\bf v}_{ij}' \tilde{\mathbf{S}}_{Schur} (\tilde{\mathbf{S}}_{t,od} + \tilde{\mathbf{S}}_{t,\bar{o} d}) 
\tilde{\mathbf{S}}_{t,dd}^{-1} {\bf 1}_p, 
$$

\noindent
and 
$$
\omega^0_{ii} = {\bf v}_{ij}' \tilde{\mathbf{\Sigma}}_{Schur} (\tilde{\mathbf{\Sigma}}_{t,od} + 
\tilde{\mathbf{\Sigma}}_{t,\bar{o} d}) \tilde{\mathbf{\Sigma}}_{t,dd}^{-1} {\bf 1}_p 
$$

\noindent
for an appropriate unit vector ${bf v}_{ij} \in \mathbb{R}^{d_{\boldsymbol{t}}}$. Using the observations in 3. and 4. above, it 
follows that for large enough $n$ 
\begin{equation} \label{mode:offdiag:bound}
\max_{(i,j) \in E^0} \left| \hat{\omega}_{ij} - \omega^0_{ij} \right| \leq \frac{K'}{2} \nu_{\max} \sqrt{\frac{d_{\boldsymbol{t}} \log 
p}{n}} 
\end{equation}

\noindent
on $C_{1,n}$. Since $\mathbb{P}_0 (C_{1,n}) \rightarrow 0$ as $n \rightarrow \infty$, using (\ref{refitted:posterior:bound}), 
it is sufficient to prove that 
$$
\mathbb{E}_0 \left[ \pi_{refitted}  \left(\| \boldsymbol{\omega} - \hat{\boldsymbol{\omega}} \|_{\max} > K' \nu_{\max} 
\sqrt{\frac{d_{\boldsymbol{t}} \log p}{n}} \mid \mathcal{Y} \right) 1_{\{\hat{G} = G^0\}} \right] \rightarrow 0 
$$

\noindent
as $n \rightarrow \infty$. It follows from (\ref{refitted:posterior:density}) and the definition of $\mathbb{M}_{\hat{G}}$ that 
\begin{eqnarray}
& & \pi_{refitted}  \left(\| \boldsymbol{\omega} - \hat{\boldsymbol{\omega}} \|_{\max} > K' \nu_{\max} 
\sqrt{\frac{d_{\boldsymbol{t}} \log p}{n}} \mid \mathcal{Y} \right) 1_{\{\hat{G} = G^0\}} \nonumber\\ 
&=& P \left( \| {\bf Z} \|_{\max} > K' \nu_{\max} \sqrt{\frac{d_{\boldsymbol{t}} \log p}{n}} \mid Z_i > - \hat{\omega}_{ii}, \; 
\forall 1 \leq i \leq p \right) \nonumber\\
&\leq& \frac{P \left( \| {\bf Z} \|_{\max} > K' \nu_{\max} \sqrt{\frac{d_{\boldsymbol{t}} \log p}{n}} \right)}{P \left( Z_i > - 
\hat{\omega}_{ii}, \; \forall 1 \leq i \leq p \right)}. \label{refitted:multivariate:normal}
\end{eqnarray}

\noindent
where $P$ is a probability measure, and the $(p+d_{\boldsymbol{t}})$-dimensional random vector ${\bf Z}$ has a 
multivariate normal distribution with mean ${\bf 0}$ and covariance matrix $\frac{\tilde{\mathbf{K}}^{-1}}{n}$ under $P$. 
Using Assumption \ref{Assumption3: Bounded eigenvalues}, the structure of $\tilde{\mathbf{P}}_t, \tilde{\mathbf{S}}_t$, 
along with (\ref{mode:diag:bound}), we get that for large enough $n$ 
$$
\hat{\omega}_{ii} > \frac{\tilde{\varepsilon}_0}{2} \; \forall 1 \leq i \leq p 
$$

\noindent
and 
$$
\lambda_{\max} \left( \tilde{\mathbf{K}}^{-1} \right) < \frac{2}{\tilde{\varepsilon}_0} 
$$

\noindent
on $C_{1,n}$. It follows by the union-sum inequality that 
\begin{eqnarray*}
\frac{P \left( \| {\bf Z} \|_{\max} > K' \nu_{\max} \sqrt{\frac{d_{\boldsymbol{t}} \log p}{n}} \right)}{P \left( Z_i > - 
\hat{\omega}_{ii}, \; \forall 1 \leq i \leq p \right)} 
&\leq& \frac{\sum_{i=1}^{p+d_{\boldsymbol{t}}} P \left( |Z_i| > K' \nu_{\max} \sqrt{\frac{d_{\boldsymbol{t}} \log p}{n}} \right)}
{1 - \sum_{i=1}^p P(Z_i < - \hat{\omega}_{ii})}\\
&\leq& \frac{\sum_{i=1}^{p+d_{\boldsymbol{t}}} P \left( |\tilde{Z}_i| > K' \nu_{\max} \sqrt{\frac{\tilde{\varepsilon}_0 d_{\boldsymbol{t}} 
\log p}{2}} \right)}{1 - \sum_{i=1}^p P \left( |Z_i| > \frac{\varepsilon}{2} \right)}\\
&\leq& \frac{\sum_{i=1}^{p+d_{\boldsymbol{t}}} P \left( |\tilde{Z}_i| > K' \nu_{\max} \sqrt{\frac{\tilde{\varepsilon}_0 d_{\boldsymbol{t}} 
\log p}{2}} \right)}{1 - \sum_{i=1}^p P \left( |\tilde{Z}_i| > \sqrt{\frac{n \tilde{\varepsilon}_0^3}{8}} \right)}. 
\end{eqnarray*}

\noindent
where $\tilde{Z}_i$ has a standard normal distribution under the probability measure $P$ for every $1 \leq i \leq 
p+d_{\boldsymbol{t}}$. Using Markov's inequality with an appropriate function of $\tilde{Z}_i$, we get 
$$
\frac{P \left( \| {\bf Z} \|_{\max} > K' \nu_{\max} \sqrt{\frac{d_{\boldsymbol{t}} \log p}{n}} \right)}{P \left( Z_i > - 
\hat{\omega}_{ii}, \; \forall 1 \leq i \leq p \right)} \leq \frac{2 (p+d_{\boldsymbol{t}}) \exp \left( - (K')^2 \nu_{\max}^2 \varepsilon 
d_{\boldsymbol{t}} \log p/4 \right)}{1 - 2p \exp \left( - n \tilde{\varepsilon}_0^3/16 \right)} 
$$

\noindent
on $C_{1,n}$. It follows by Assumption \ref{Assumption2new}, (\ref{refitted:multivariate:normal}), and $\mathbb{P}_0 
(C_{1,n}) \rightarrow 0$ as $n \rightarrow \infty$ that 
$$
\mathbb{E}_0 \left[ \pi_{refitted}  \left(\| \boldsymbol{\omega} - \hat{\boldsymbol{\omega}} \|_{\max} > K' \nu_{\max} 
\sqrt{\frac{d_{\boldsymbol{t}} \log p}{n}} \mid \mathcal{Y} \right) 1_{\{\hat{G} = G^0\}} \right] \rightarrow 0 
$$

\noindent
as $n \rightarrow \infty$ for a large enough choice of $K'$. This establishes the first part of 
Theorem \ref{estimation:consistency:refitted:posterior} (with the $\| \cdot \|_{\max}$ norm). The second part follows by 
noting that on $\hat{G} = G^0$ 
$$
\left\| \mathbf{\Omega} - \mathbf{\Omega}^0 \right\| \leq \nu_{\max} \left\| \mathbf{\Omega} - \mathbf{\Omega}^0 
\right\|_{\max}. 
$$

\section{Appendix D: Continuous shrinkage priors - Horseshoe prior}

\noindent
Continuous shrinkage prior distributions are a popular alternative to spike-and-slab ones. Such prior distributions have a peak at zero and their tails decay at an appropriate rate. They serve as continuous approximations to the discrete mixture-based spike-and-slab prior distributions. Continuous shrinkage prior distributions are often a scale mixture of normals, such as Laplace-half-Cauchy, etc. (see \cite{polson2010shrink},\cite{bhattacharya2015dirichlet} and references therein). In the context of linear regression, the Bayesian lasso of \cite{park2008bayesian}, based on the interpretation of the well-known lasso estimator of the regression coefficients as the posterior mode in a Bayesian model which puts independent Laplace priors on the individual coefficients, has gained popularity in recent years. 

As mentioned in Section \ref{background}, the Bayesian Graphical lasso was proposed by \cite{wang2012bayesian}, as a 
Bayesian adaptation of the graphical lasso. The authors in \cite{wang2012bayesian} consider a Bayesian model which puts 
independent Laplace priors on the off-diagonal entries of $\mathbf{\Omega}$ and independent exponential priors on the 
diagonal entries of $\mathbf{\Omega}$ (restricted to $\mathbf{\Omega}$ begin positive definite). It follows that the graphical 
lasso estimator is the posterior mode of this Bayesian model. The Bayesian graphical lasso interpretation immediately yields 
credible regions for the graphical lasso estimate of $\mathbf{\Omega}$. Such estimates of uncertainty are not readily 
available in the frequentist setting. Alternatively, some practitioners also determine sparsity in $\mathbf{\Omega}$ based on 
whether zero is contained in the credible interval for the respective off-diagonal entries.

In principal, any continuous shrinkage prior distribution on the off-diagonal entries can be used in conjunction with the CONCORD generalized likelihood (\ref{eq2}). We will demonstrate this by choosing the popular horseshoe prior developed in 
\cite{carvalho2010horseshoe}. Consider the following hierarchical prior for every $\omega_{jk}$ with $j \neq k$:  
\begin{equation}\label{cauchy bayesian setting}
\begin{split}
\omega_{jk}|\lambda_{jk}^2, \tau^2 &\sim \mathcal{N}\left(0, \lambda_{jk}^2\tau^2 \right),\\
\lambda_{jk} & \sim \mathcal{C}^{+}\left(0, 1\right),\\
\tau & \sim \mathcal{C}^{+}\left(0, 1\right),
\end{split}
\end{equation}
where, $\mathcal{C}^{+}$ is the standard half-Cuachy distribution with probability density function
\begin{equation}
p\left( z \right) = \frac{2}{\pi \left(1+z^2\right)}, \quad z > 0,
\end{equation}
and $1 \leq j < k \leq p$. 

This hierarchical setting defines the horseshoe prior distribution, which is a global-local shrinkage distribution, wherein the local shrinkage for $\omega_{jk}$s is determined by $\lambda_{jk}$s and the overall 
level of shrinkage is determined by the hyperparameter $\tau$. The particular choice of the half-Cauchy distribution results in 
aggressive shrinkage of small in magnitude partial correlations and virtually no shrinkage of the sufficiently large ones. This is in contrast to other continuous shrinkage prior distributions, such as the Laplace (the Bayesian Lasso by \cite{park2008bayesian}) wherein the shrinkage effect is uniform across all values of the models parameters. For further studies regarding other properties of the horseshoe prior distribution, we refer the reader to \cite{carvalho2010horseshoe}, \cite{polson2012local} and \cite{polson2010shrink}. 

Employing the original form of the horseshoe distribution in (\ref{cauchy bayesian setting}) results in non-standard conditional posterior distributions for the hyperparameters $\lambda_{12}, ..., \lambda_{p-1p}, \tau$, which makes a standard Gibbs sampling algorithm difficult to implement. In the context of linear regression models, some studies have suggested the use of specialized algorithms, such as slice sampling for the hyperparameters, \cite{neal2003slice} and 
\cite{omre1989bayesian}. Recently, \cite{makalic2016simple} introduced an alternative sampling scheme for all model parameters based on auxiliary variables that leads to conjugate conditional posterior distributions for all parameters in various regression models. They make use of the following scale mixture representation of the Horse Shoe prior to construct a Gibbs sampler. Let $x$ and $a$ be random variables such that 
\begin{equation}
x^2 \sim \mathcal{IG}\left(1/2, 1/a\right), \quad \quad \quad a \sim \mathcal{IG}\left(1/2, 1/A^2\right); 
\end{equation}
then, $x\sim \mathcal{C}^{+}\left(0, A\right)$ (\cite{wand2011mean}) with $\mathcal{IG}(.,.)$ being the inverse-gamma distribution with probability density function 
\begin{equation}
p\left(z|\alpha, \beta\right) = \frac{\beta^\alpha}{\Gamma\left(\alpha\right)} z^{-\alpha - 1}\exp\left( -\frac{\beta}{z}\right), \quad z > 0.
\end{equation}
The above decomposition results in the following revised horseshoe hierarchy
\begin{equation}
\begin{split}
\omega_{jk}|\lambda_{jk}^2, \tau^2 &\sim \mathcal{N}\left(0, \lambda_{jk}^2\tau^2 \right),\\
\lambda_{jk}^2 | \nu_{jk} &\sim \mathcal{IG}\left(1/2, 1/\nu_{jk}\right), \\
\tau^2|\varepsilon & \sim \mathcal{IG}\left(1/2, 1/\varepsilon\right), \\
\nu_{12}, ..., \nu_{p-1p}, \varepsilon & \sim \mathcal{IG}\left(1/2, 1\right).
\end{split}
\end{equation}

Note that it becomes straightforward to construct a Gibbs sampling scheme. The conditional posterior distribution of the edge
parameters $\omega_{jk}$ is given by
\begin{equation}
\begin{split}
(\omega_{jk}| \mathbf{\Omega}_{-(jk)}, \mathcal{Y}) \sim N(-\frac{b_{jk}}{a_{jk}}, \frac{1}{na_{jk}}), \quad \quad 1 \le j < k \le p.
\end{split}
\end{equation}
with,
\begin{equation*}
\begin{split}
a_{jk} &= s_{jj} + s_{kk} + \frac{1}{n\lambda_{jk}^2\tau^2}, \quad \quad \quad 
b_{jk} = \mathbf{\Omega}_{-jk}'\mathbf{S}_{-jj} + \mathbf{\Omega}_{-kj}'\mathbf{S}_{-kk}, \end{split}
\end{equation*}
the conditional posterior probabilities of the local and global hyperparameters are inverse-gamma distributions 

\begin{equation}
\begin{split}
\lambda_{jk}^2|. &\sim \mathcal{IG}\left(1, \frac{1}{\nu_{jk}} + \frac{\omega_{jk}^2}{2\tau^2} \right),\\
\tau^2|. &\sim \mathcal{IG}\left(\frac{1}{2} + \frac{p(p-1)}{4}, \frac{1}{\varepsilon} + \sum\limits_{j=1}^{p-1} 
\sum\limits_{k=j+1}^{p}\frac{\omega_{jk}^2}{2\lambda_{jk}^2} \right).
\end{split}
\end{equation}

\noindent
Finally, the conditional posterior distribution for the auxiliary variables is given by
\begin{equation}
\begin{split}
\nu_{jk}|. &\sim \mathcal{IG}\left(1, 1 + \frac{1}{\lambda_{jk}^2} \right)\\
\varepsilon|. &\sim \mathcal{IG}\left(1, 1 + \frac{1}{\tau^2}\right).
\end{split}
\end{equation}
The resulting Gibbs sampler is summarized in the following algorithm.

\begin{algorithm}
\caption{Entry Wise Gibbs Sampler for BHSC}
\begin{algorithmic}
\Procedure{BHSC}{$\mathbf{y}_{1:}, ..., \mathbf{y}_{n:}$}\Comment{Input the data}
\For{$j=1, ..., p-1$}
\For{$k=j+1, ..., p$}
\State $a_{jk} \gets s_{jj} + s_{kk} + \frac{1}{n\lambda_{jk}\tau^2}$
\State $b_{jk} \gets \mathbf{\Omega}_{-jk}'\mathbf{S}_{-jj} + \mathbf{\Omega}_{-kj}'\mathbf{S}_{-kk}, $
\State $\omega_{jk} \sim N(-\frac{b_{jk}}{a_{jk}}, \frac{1}{na_{jk}})$
\State $\lambda_{jk}^2|. \sim \mathcal{IG}\left(1, \frac{1}{\nu_{jk}} + \frac{\omega_{jk}^2}{2\tau^2} \right)$
\State $\tau^2|. \sim \mathcal{IG}\left(\frac{1}{2} + \frac{p(p-1)}{4}, \frac{1}{\varepsilon} + \sum\limits_{j=1}^{p-1}\sum\limits_{k=j+1}^{p}\frac{\omega_{jk}^2}{2\lambda_{jk}^2} \right)$
\State $\nu_{jk}|. \sim \mathcal{IG}\left(1, 1 + \frac{1}{\lambda_{jk}^2} \right)$
\State $\varepsilon|. \sim \mathcal{IG}\left(1, 1 + \frac{1}{\tau^2}\right)$
\EndFor
\EndFor
\For{$j=1, ..., p$}
\State $\omega_{jj} \gets \frac{-(\lambda + n\mathbf{\Omega}_{-jj}'\mathbf{S}_{-jj}) + \sqrt{(\lambda + n\mathbf{\Omega}_{-jj}'\mathbf{S}_{-jj})^2 + 4n^2 s_{ii}^k}}{2n s_{ii}^k}$
\EndFor
\State \textbf{return} $ \mathbf{\Omega}$\Comment{Return $\mathbf{\Omega}$}
\EndProcedure
\end{algorithmic}\label{GibbsBHSC}
\end{algorithm}

\section{Appendix E: Background on the CONCORD regression based generalized likelihood} \label{background}

\noindent
Let $\mathcal{Y} :=\left(\{\mathbf{y}_{i:}\}_{i=1}^{n}\right)$ denote $i.i.d$ observations from a 
$p$-variate (continuous) distribution, with mean $\boldsymbol{0}$ and covariance matrix $\mathbf{\Omega}^{-1}$.  Let 
$\mathbf{S}$ denote the sample covariance matrix of the observations. In the frequentist setting, one of the standard 
methods to achieve a sparse estimate of $\bf{\Omega}$ is to minimize an objective function, comprising of the (negative) 
Gaussian log-likelihood and an $\ell_1$-penalty term for the off-diagonal entries of $\bf{\Omega}$, over the space of 
positive definite matrices. Equivalently, one can maximize the following weighted Gaussian likelihood: 
\begin{equation}\label{glasso}
\text{exp}\left( -\frac{n}{2} \left\{   \text{tr} \left( \mathbf{\Omega} \mathbf{S} \right) - \text{log det}\ \mathbf{\Omega} + 
\frac{\lambda}{n} \mathop{\sum\sum}\limits_{1\le j \le k \le p} {\left| \omega_{jk}  \right|}  \right\}\right). 
\end{equation}

\noindent
This approach and its variants are known as the {\it graphical lasso}, see 
\cite{yuan2007model, friedman2008sparse, banerjee2008model}. The function in (\ref{glasso}) can also be 
regarded as the posterior density of $\mathbf{\Omega}$ (up to proportionality) under Laplace priors for the off-diagonal entries, leading to a Bayesian inference and analysis framework \cite{wang2012bayesian}. 
Note that the requirement on $\mathbf{\Omega}$ being positive definite translates to the need of inverting 
$(p-1)\times (p-1)$ matrices in each iteration of the graphical lasso or Bayesian Markov Chain Monte Carlo (MCMC) based algorithms. This issue is mitigated in the graphical lasso algorithm by the small number of iterations required, but becomes critical for Bayesian approaches that require many iterations (in the thousands) of the corresponding MCMC scheme.

To address this problem in the frequentist setting, several works (see \cite{peng2009partial}, \cite{khare2015convex}) have 
considered replacing the Gaussian likelihood by a regression based generalized likelihood. Note that $-\frac{\omega_{jk}}
{\omega_{jj}}$ is the regression coefficient of $\mathbf{y}_{k:}$ when we regress $\mathbf{y}_{j:}$ on all other variables, and 
$\frac{1}{\omega_{jj}}$ is the residual variance. This is true even in non-Gaussian settings. \cite{peng2009partial} use this 
interpretation to define a generalized likelihood in terms of the (negative) exponent of the combined weighted squared error 
loss associated with all these regressions (corresponding to all $p$ variables) as follows. 
\begin{equation} \label{eq1}
\text{exp} \left( -\sum\limits_{j=1}^{p} {\omega_{jj}} \left\{ \sum\limits_{i=1}^{n} {\left( y_{ij} - \sum\limits_{k\ne j}
-{\frac{\omega_{jk}}{\omega_{jj}}y_{ik}} \right)^2} \right\} - \sum\limits_{j=1}^p \frac{n}{2} {\text{log}\omega_{jj}} \right).
\end{equation}

\noindent
Under Gaussianity, the expression in (\ref{eq1}) corresponds to the product of the conditional densities of each variable given 
all the other variables in the data set, and corresponds to Besag \cite{Besag:1975}'s pseudo-likelihood. \cite{peng2009partial} 
develop the SPACE algorithm which obtains a sparse estimator for $\mathbf{\Omega}$ by minimizing an objective function 
consisting of the (negative) log generalized likelihood and an $\ell_{1}$ penalty term for off-diagonal entries of 
$\mathbf{\Omega}$. However, this objective function is not jointly convex, which can lead to serious convergence issues for 
the corresponding minimization algorithm. 

\cite{khare2015convex} address this issue by appropriately re-weighting each of the $p$ regression terms in the exponent of 
(\ref{eq1}) and combining it with an $\ell_1$ penalty term to obtain a regression based loss function which is jointly convex in the elements of $\mathbf{\Omega}$. This loss function, referred by \cite{khare2015convex} as the CONCORD objective function, is given by
\begin{equation} \label{concord likelihood}
\begin{split}
Q_{con}\left(\mathbf{\Omega} \right) &= -n\sum\limits_{j=1}^{p}{\text{log}\omega_{jj}} + \frac{1}{2}\sum\limits_{j=1}^{p}
{\sum\limits_{i=1}^{n}{\left( \omega_{jj}y_{ij} + \sum\limits_{k \ne j}{\omega_{jk}y_{ik}} \right)^2}} + \lambda \mathop{\sum\sum}
\limits_{1 \le j < k \le p}|\omega_{jk}|  \\
&= -n\sum\limits_{j=1}^{p}{\text{log}\omega_{jj}} + \frac{n}{2} \text{tr} \left( \mathbf{\Omega}^2\mathbf{S}\right) + \lambda 
\mathop{\sum\sum}\limits_{1 \le j < k \le p}|\omega_{jk}|, 
\end{split}
\end{equation}

\noindent
where $\mathbf{S}$ denotes the sample covariance matrix. The joint convexity of $Q_{con}$ can be used to show that a
coordinate-wise minimization algorithm always converges to a global minimum. Note that both \cite{peng2009partial} and 
\cite{khare2015convex} relax the parameter space of $\mathbf{\Omega}$ from positive definite matrices to symmetric matrices with positive diagonal entries. The primary purpose of this relaxation is computational. Combined with the quadratic nature of the objective function, this relaxation leads to an {\it order of magnitude decrease in computational complexity} as compared to graphical lasso based approaches. Note that given the $\text{log det} \mathbf{\Omega}$ term in the Gaussian likelihood, such a parameter relaxation will not work for the graphical lasso. 

While the resulting minimizer may not be positive definite, its sparsity structure can be used to address the 
primary goal/challenge of selecting the sparsity pattern in $\mathbf{\Omega}$. \cite{khare2015convex} establish 
high-dimensional sparsity selection consistency of this approach, and also demonstrate that the CONCORD approach can 
outperform graphical lasso when the underlying data is not generated from a multivariate Gaussian distribution. This 
robustness is somewhat expected, since the regression based interpretation of $\mathbf{\Omega}$ does not depend on normality.

\newpage
\newpage
\small

\end{document}